\documentclass[10pt,a4wide]{article}
\usepackage{hyperref}
\usepackage{amsfonts}
\usepackage{amssymb}
\usepackage{amsthm}
\usepackage{amsmath} 
\usepackage{graphicx}
\usepackage{float}
\usepackage[labelfont=bf]{caption}
\usepackage[dvipsnames]{xcolor}

\usepackage{color}
\definecolor{newcolor}{rgb}{.8,.349,.1}

\usepackage{color}

\date{}
\numberwithin{equation}{section}

\newcommand{\sss}{ }

\newtheorem{theorem}{THEOREM}[section]
\renewcommand{\thetheorem}{\arabic{section}.\arabic{theorem}}
\newtheorem{lemma}[theorem]{LEMMA}

\newtheorem{remark}[theorem]{REMARK}

\newtheorem{prop}[theorem]{PROPOSITION}

\renewcommand{\theequation}{\arabic{section}.\arabic{equation}}

\newcommand\beq{\begin{equation}}
\newcommand\ene{\end{equation}}
\def \ds{\displaystyle}

\newcommand{\R}{\mathbb{R}}

\newcommand{\modifsF}{\textcolor{black}}

\begin{document}
%

\title{The \textcolor{black}{magnetized} Vlasov-Amp\`ere system and the Bernstein-Landau paradox
 \thanks{ Research partially supported by project  PAPIIT-DGAPA UNAM  IN103918 and by project SEP-CONACYT CB 2015, 254062. }}

\author{ Fr\'ed\'erique Charles\thanks{frederique.charles@sorbonne-universite.fr},  Bruno Despr\'es\thanks{IUF, bruno.despres@sorbonne-universite.fr}, Alexandre Rege\thanks{alexandre.rege@sorbonne-universite.fr } \\
Sorbonne-Universit\'e, CNRS, Universit\'e de Paris, \\
Laboratoire Jacques-Louis Lions (LJLL),  F-75005 Paris, France.\\
Ricardo Weder\thanks {Fellow, Sistema Nacional de Investigadores. weder@unam.mx. Home page: http://www.iimas.unam.mx/rweder/rweder.html} \\
Departamento de F\'{\i}sica Matem\'atica.\\
 Instituto de Investigaciones en Matem\'aticas Aplicadas y en
 Sistemas. \\
 Universidad Nacional Aut\'onoma de M\'exico.\\
  Apartado Postal 20-126,
Ciudad de M\'exico  01000, M\'exico.}

\maketitle


\vspace{.5cm}
 \centerline{{\bf Abstract}} 
We study the Bernstein-Landau paradox in the collisionless motion of an electrostatic plasma in the presence of a constant external magnetic field. The Bernstein-Landau paradox consists in that in the presence of the magnetic field, the electric field  and the charge density fluctuation have an oscillatory behavior in time.  This is radically different from Landau damping, in the case without magnetic field, where the electric field tends to zero for large times. We consider this problem from a new point of view. Instead of analyzing the linear \textcolor{black}{magnetized}  Vlasov-Poisson system, as it is usually done, we study the linear \textcolor{black}{magnetized} Vlasov-Amp\`ere system. We formulate the 
\textcolor{black}{magnetized} Vlasov-Amp\`ere system as a Schr\"odinger equation with a selfadjoint \textcolor{black}{magnetized} Vlasov-Amp\`ere operator in the Hilbert space of states with finite energy.
 The \textcolor{black}{magnetized} Vlasov-Amp\`ere operator has a complete set of orthonormal eigenfunctions, that include the Bernstein modes. The expansion of the solution of the \textcolor{black}{magnetized} Vlasov-Amp\`ere system in the eigenfunctions shows the oscillatory behavior in time. We prove the convergence of the expansion under optimal conditions, assuming only that the initial state has finite energy. This solves a problem that was recently posed in the literature. The Bernstein modes are not complete.  To have  a complete  system it is necessary to add eigenfunctions that are associated with eigenvalues at  all the integer multiples of the cyclotron frequency. These special plasma oscillations actually exist on their own, without the excitation  of the other modes. \textcolor{black}{ In the limit when the magnetic fields goes to zero the spectrum of the \textcolor{black}{magnetized} Vlasov-Amp\`ere operator changes drastically  from  pure point to absolutely continuous in the orthogonal complement to its kernel, due to a sharp change on its domain. This explains the Bernstein-Landau paradox.}  Furthermore, we present numerical simulations that illustrate the Bernstein-Landau paradox.  
In  Appendix~ B we provide  exact formulas for a family of time-independent solutions.
\bigskip

\section{Introduction}\sss

Collisionless motion of an electrostatic  plasma can exhibit wave damping, a  phenomenon
identified by  Landau in \cite{landau}, and that is called Landau damping. It consists in the decay for large times of  the electric field. There is a very extensive literature 
on Landau damping. See for example, \cite{degond,despres1,despres2,ryu,stix,swanson}, and the references quoted there.   For a recent deep mathematical study 
of Landau damping in the nonlinear case see  \cite{villani}. On the contrary, it is known that \textcolor{black}{\textcolor{black}{magnetized}} plasmas can 
prevent  Landau damping \cite{bernstein}.  In fact, it was shown by Bernstein \cite{bernstein} that in the presence of a constant magnetic field  the  electric field does not decay  for large times, and that, actually, it  has an oscillatory behavior as a function of time.  This  phenomenon is  called the Bernstein-Landau paradox, see for example  \cite{suku}, because it seems paradoxical that even an arbitrary   small, but nonzero, value of the external constant magnetic field can be the cause of this radical change in the behaviour of  the electric field for large times.  The standard  theory of the Bernstein-Landau paradox in the physics  literature is based on the representation  of the solutions to the \textcolor{black}{magnetized} Vlasov-Poisson system of equations in terms of a series of  Bernstein modes. See, for example, \cite{stix}[section 9.16] and \cite{swanson}[section 4.4.1]. 

It is the purpose of the present work to revisit  the Bernstein-Landau paradox from a new point of view. Instead of considering the \textcolor{black}{magnetized} Vlasov-Poisson system we  study the \textcolor{black}{magnetized} Vlasov-Amp\`ere system. We  write the  \textcolor{black}{magnetized} Vlasov-Amp\`ere system as a Schr\"odinger equation where   the \textcolor{black}{magnetized}  Vlasov-Amp\`ere operator plays the role of the  Hamiltonian. We construct a   realization of the \textcolor{black}{magnetized}  Vlasov-Amp\`ere operator as a selfadjoint operator in the Hilbert space, that we call $\mathcal H,$ that consists of the charge density functions  that are square integrable and of the electric fields that are square integrable and of mean zero. Actually, the square of the norm of $\mathcal H$ is the energy. From the physical point of view this permits us  to use the conservation of the energy  in a very explicit way. On the mathematical side, this allows us to bring into the fore the powerful methods of the spectral  theory of selfadjoint operators in a Hilbert space. There is a very extensive literature in spectral  theory, see for example \cite{ka, rs1,rs2,rs4,rs3}. This approach has previously been used  in the case without magnetic field to analyze the Landau-damping in \cite{despres1,despres2}.  \textcolor{black}{Within this framework the study of the  Bernstein-Landau paradox reduces to the  proof that the \textcolor{black}{magnetized}
 Vlasov-Amp\`ere operator only has pure point spectrum, i.e., that its spectrum consists only of eigenvalues. 
Then, the fact that the  \textcolor{black}{magnetized}  Vlasov-Amp\`ere operator  has a complete set of orthonormal  eigenfunctions follows from the abstract spectral theory of selfadjoint operators. 
We  expand the general solutions to the \textcolor{black}{magnetized} Vlasov-Amp\`ere   system   in the orthonormal basis of eigenfunctions of the  \textcolor{black}{magnetized} Vlasov-Amp\`ere operator. The coefficients of this expansion are the product of the scalar product of the initial state with the corresponding eigenfunction, and of   the phase $e^{-it \lambda},$ where $t$ is time and $\lambda$ is the eigenvalue of the eigenfunction. This representation of the solution shows the oscillatory behavior in time for $\lambda\neq 0$,  or constant in time for $\lambda =0, $ that is to say the Bernstein-Landau paradox.} Moreover,  our representation of the solution as an expansion in the orthonormal basis of eigenfunctions of the \textcolor{black}{magnetized}  Vlasov-Amp\`ere operator converges strongly in $\mathcal H$ for any initial data in $\mathcal H,$ that is to say for any square integrable initial state without any further restriction in regularity and decay. Note that our result is optimal, since square integrability is the minimum that we can  require, even to pose precisely the problem. A physical state has to have finite energy, i.e., it has to be square integrable.\textcolor{black}{Note moreover, that we prove that the Bernstein modes are not  complete. In fact, we prove in Theorem ~\ref{expansion} that for general initial data with finite energy, and that satisfy the Gauss law, the charge density fluctuation is a sum of two terms. One of them  is oscillatory in time (see  \eqref{4.152})  and it  coincides with the standard series of Bernstein modes given in \cite{bernstein,bedro}. The other term, given in \eqref{4.152.1}, is constant in  time, and it is a series of   eigenfunctions with eigenvalue zero, of  the magnetized Vlasov-Amp\' ere system, and that satisfy the Gauss law.  It appears that the fact that  the Bernstein modes  are not enough to expand the general charge density fluctuation, and that one has  also to consider static modes is a new result, that has not been observed previously in the literature. }  We prove that the spectrum of the  \textcolor{black}{magnetized} Vlasov-Amp\`ere operator is pure point in two  different ways.

 In the first one, we actually compute the eigenvalues and we explicitly construct a   orthonormal basis of eigenfunctions, i.e., a complete set of orthonormal eigenfunctions. This, of course, gives us much more than just that the expansion
of the charge density fluctuation, and is interesting in its own right, because it can be used for many other purposes. \textcolor{black}{ As we mentioned above,  our analysis shows that the Bernstein modes alone are not a complete orthonormal system. In addition to the eigenfunctions with eigenvalue zero that contribute to the  static part of the charge density fluctuation,  to have  a complete orthonormal system in the Hilbert space, $\mathcal H $  of configurations  with finite energy, it is necessary to add other eigenfunctions that are associated with eigenvalues at  all the integer multiples of the cyclotron frequency, including the zero eigenvalue. These other eigenfunctions have nontrivial density function, but the electric field  and the charge density fluctuation are zero. Recall that   the charge density fluctuation is obtained averaging the density function over the velocities. In consequence, these other eigenfunctions do not appear in the  expansion of the charge density fluctuation.  Anyhow, these eigenfunctions are physically interesting because they show that there are plasma oscillations such that at each point the charge density fluctuation and the electric field  are zero. Some of them are time independent. Note that since our eigenfunctions are orthonormal, these special plasma oscillations actually exist on their own, without the excitation  of the other modes.} \textcolor{black}{It appears to us that this fact, or at least the exact
form of   these eigenfunctions,  with zero charge density fluctuation and zero electric field, has not been  observed previously in the literature.} 

In the second one we use an abstract operator theoretical argument based on the celebrated Weyl theorem  on the invariance of the essential spectrum of a selfadjoint operators in Hilbert space.
This argument allows us to prove that the \textcolor{black}{magnetized} Vlasov-Amp\`ere operator has pure point spectrum. It gives a less detailed information about where the eigenvalues are located, and it tells us nothing about the eigenfunctions. However, it is enough for the proof of the existence of the Bernstein-Landau  paradox without going through  the detailed calculations of the first approach. It also tells us  why the Bernstein-Landau paradox exists from a general  principle in spectral theory. 
       
\textcolor{black}{On the contrary  in the case where the magnetic field is zero, it was proven  in \cite{despres1,despres2}  that  the spectrum of the  \textcolor{black}{magnetized} Vlasov-Amp\`ere operator is 
made of an  absolutely continuous part and of a kernel. 
The Landau damping follows from the well known fact that for  a selfadjoint operator  $\mathbf H$, the operator 
 $e^{-it\mathbf H}P_0$ goes weakly to zero as $ t \to \pm \infty$ (here $P_0$ is the projection on the absolutely continuous part of the spectrum).}
   It has been remarked in \cite{degond} that there are  "interesting analogies with Lax and Phillips scattering theory " \cite{lax:phillips}.
In fact, the results of \cite{despres1,despres2}  prove that it is not  just an analogy, but the  consequence of a convenient  reformulation of Landau damping in terms of the  \textcolor{black}{magnetized} Vlasov-Amp\`ere system.
\textcolor{black}{The sharp change in the spectrum of the \textcolor{black}{magnetized}  Vlasov-Amp\`ere operator when the magnetic field goes to zero, i.e. from pure point to absolutely continuous in the orthogonal complement to its kernel, may appear to be paradoxical because the  formal  \textcolor{black}{magnetized}  Vlasov-Amp\`ere operator is formally analytic in the magnetic field. The issue is that the domain of the selfadjoint realization  of the  \textcolor{black}{magnetized} Vlasov-Amp\`ere operator changes   abruptly when the magnetic field is zero. It is a well known fact in the spectral  theory of families of linear operators that  the spectrum can change sharply at values of the parameter where the domain of the operator  sharply changes.  For a comprehensive presentation of these results the reader can consult, for example, \cite{ka}. Summing up, this shows that there is no  paradox in the Bernstein-Landau  paradox, just a well known  fact of spectral  theory, but, of course, in the physics  literature the domains of the operators are usually not  taken into account. }
Perhaps the reason why the absence of Landau damping  for arbitrarily small magnetic fields is considered as paradoxical  is related to the fact that the 
  \textcolor{black}{magnetized} Vlasov-Poisson system somehow hides the underlying mathematical physics structure of our problem, in spite of the fact that it is a convenient tool, particularly for computational purposes.  Let us explain what we mean. The full Maxwell  equations consist of the Maxwell-Faraday equation, the Amp\`ere equation, the Gauss law, and the Gauss law for magnetism, i.e., the divergence of the magnetic field is zero. In our case the Maxwell-Faraday equation and the Gauss law for magnetism are automatically satisfied. So,  of the Vlasov-Maxwell equations,  only  the Vlasov equation remains, as well as  the Amp\`ere equation, and the Gauss law. Furthermore, the Gauss law is a constraint that is only necessary to impose at the initial time, since it is propagated  by the \textcolor{black}{magnetized}  Vlasov-Amp\`ere system. Further, both the Vlasov and the Amp\`ere equations are evolution equations.   So, the natural way to proceed is to solve the  \textcolor{black}{magnetized} Vlasov-Amp\`ere system as an evolution problem, and to restrict the initial data to those who satisfy the Gauss law. The situation with the \textcolor{black}{magnetized}   Vlasov-Poisson system is somewhat different because  the Amp\`ere  equation is not explicitly taken into account. So, one could think that the \textcolor{black}{magnetized}  Vlasov-Poisson system is incomplete. The remedy is that instead of imposing the Gauss law only at the  initial time, it is required at all times. We actually prove in Section  ~\ref{sectionvlasamp} that the   \textcolor{black}{magnetized} Vlasov-Poisson system is  indeed equivalent to the \textcolor{black}{magnetized}  Vlasov-Amp\`ere system plus the validity of the Gauss law at the initial time. However, the \textcolor{black}{magnetized} Vlasov-Poisson system is a hybrid one where the Vlasov equation is an evolution equation and the Poisson equation is an elliptic equation, without time derivative. This is one way to understand why in  the \textcolor{black}{magnetized} Vlasov-Poisson system the basic mathematical physics of our problem is not so apparent. On the contrary, as we mentioned above, the  \textcolor{black}{magnetized} Vlasov-Amp\`ere system is an evolution problem, that moreover, as we already mentioned, and as  we explain in Section ~\ref{sectionvlasamp},  has a conserved energy that is explicitly expressed in terms of the density function and the electric field that appear in the  \textcolor{black}{magnetized} Vlasov-Amp\`ere system.  These two facts are the reasons why the \textcolor{black}{magnetized}  Vlasov-Amp\`ere system has a selfadjoint formulation in Hilbert space,
and then, it is clear that there is no paradox in the Bernstein-Landau paradox, as we explained above.
 
  Once the selfadjointness of the \textcolor{black}{magnetized} Vlasov-Amp\`ere formulation is established, it is a matter of explicit calculations to determine the eigenfunctions. 
  The technicalities of the calculations  are related to the fact that three different natural decomposition are
  combined. The first one is based on Fourier decomposition (factors $e^{inx}$), the second one is based on a direct sum of the kernel of the operator and 
  its orthogonal (it will be denoted as  $\mathcal H= \mbox{Ker}[\mathbf H ]\oplus \mbox{Ker}[\mathbf H ]^\perp$) and the third one  starts from the determination of the 
  eigenfunctions with a vanishing electric field (it will be denoted as $F=0$). The combination of these  three decompositions is made compatible with
  convenient notations. 
  
\textcolor{black}{Passing to the limit $\omega_c\to 0$ in the representation formulae is possible in principle. This requires a careful analysis. We do not  consider this problem  in this work except in the very last remark in  Appendix~ B.
 Nevertheless,  in Lemma \ref{lemm4.1} two consecutive eigenvalues are different 
 by the constant $\omega_c$ and in Lemma \ref{lemm4.4} two consecutive eigenvalues are different 
 by a value that is smaller than $ 2 \, \omega_c$.
 Therefore,  at the limit $\omega_c\to 0$,   the discrete spectrum fills the entire  real line
 with density and so it approaches the spectrum of the limit problem without magnetic field.}
A recent mathematical work \cite{bedro} studied the Vlasov and the Vlasov-Fokker-Planck equations in a box, in three dimensions in configuration and in velocity space, and with a constant background magnetic field.  They consider the Landau damping, the Bernstein-Landau paradox and the enhanced collisional relaxation in the limit when the collision frequency goes to zero. Since in this paper we study the case when there are no collisions, we will only comment on the results of \cite{bedro} when the collision frequency is zero.  We denote by $\mathcal T^3$ the three-dimensional torus, ie., $\mathcal T^3:= [0,2\pi]^3_{\rm per}.$    In the collisionless case \cite{bedro} considers the following  linearized magnetized Vlasov-Poisson system for 
$ (\mathbf x,\mathbf v) \in  \mathcal T^3\times \mathbb R^3$, where $\mathbf x=(x,y,z),$ and $\mathbf v=(v_1,v_2,v_3),$
\begin{equation} \label{bw.1}
\left\{
\begin{aligned}
&\partial_t \mathcal G(t,\mathbf x, \mathbf v)+ \mathbf v \cdot \nabla_{\mathbf x} \mathcal G(t,\mathbf x, \mathbf v) + \frac{q}{m}   \mathcal F(t,\mathbf x)\cdot  
\nabla_{\mathbf v} {f^0( \mathbf v)}+ \frac{q}{m}   \mathbf v \times \mathbf B_0\cdot
\nabla_{\mathbf v} \mathcal G(t,\mathbf x,\mathbf v)= 0, \qquad \mathcal G(0,\mathbf x,\mathbf v)= \mathcal G_{\rm in}(\mathbf x, \mathbf v),\\
& \mathcal F(t,\mathbf x)=-\nabla_{\mathbf x}\, \int_{\mathcal T^3}\, W(\mathbf x-\mathbf y)\, \rho(t,\mathbf y)\, d^3\mathbf y,\\
&\int_{\mathcal T^3\times\mathbb R^3}\, \mathcal G(t,\mathbf x,\mathbf v)\, d^3\mathbf x\,d^3\mathbf v=0,
\end{aligned} 
\right.
\end{equation}
where,
$$
\rho(t,\mathbf x):= \int_{\mathbb R^3}\, \mathcal G(t,\mathbf x, \mathbf v)\, d^3\mathbf v,
$$
and
$$
W(\mathbf x):= \frac{q}{4\pi}\, \frac{1}{(2\pi)^3}\, \sum_{\mathbf k \in \mathbb Z^3\setminus\{0\}}\, \frac{1}{
|\mathbf k|^2}\, e^{ik\cdot\mathbf x},
$$
$q,m >0,$ and $\mathbf B_0=(0,0,B_0), B_0 >0.$ Moreover,
\begin{equation} \label{eq:ftilde}
f^0(v)= \frac{1}{2\pi}\, e^{-\frac{ v_1^2+v_2^2}{2}}\,\left( \frac{1}{\sqrt{2\pi T_\parallel}}\, e^{-\frac{-v_3^2}{2 T_\parallel}}+ \tilde{f}(v_3^2)\right),
\end{equation}
with   $T_\parallel >0.$ 

To state the results of \cite{bedro} we first introduce some notation. Let us define the Fourier coefficients of $ f \in L^2(\mathcal T^3)$ as follows,
$$
\hat{f}(\mathbf k):= \int_{\mathcal T^3}\, e^{-i\mathbf k\cdot \mathbf x}\, f(\mathbf x)\, d^3 \mathbf x,\qquad \mathbf k=(k_1,k_2,k_3) \in \mathbb Z^3.
$$
Then,
$$
f(\mathbf x)= \frac{1}{(2\pi)^3}\, \sum_{\mathbf k \in \mathbb Z^3}\, e^{i \mathbf k\cdot \mathbf x} \, \hat{f}(\mathbf k).
$$
 Denoting  $<\mathbf v>:=  \sqrt{1+|\mathbf v|^2},$ the $H^\sigma_m$ norm of a function f is defined as follows \cite{bedro},
 $$
 \|f\|^2_{H^\sigma_m}:= \int_{\mathcal T^3\times \mathbb R^3}\,  <\mathbf v>^{2 m}  | <\nabla>^\sigma \,f(\mathbf x,\mathbf v) |^2
\, d^3 \mathbf x d^3 \mathbf v,
 $$
with $\nabla=\nabla_{\mathbf x,\mathbf v}$ the differential operator both in $\mathbf x$ and $\mathbf v.$ 

In \cite{bedro} the following theorem is stated.
\begin{theorem} \label{theob}(Bedrossian and Wang, Theorem 1, \cite{bedro})
 Let $<\mathbf v>^m \mathcal G_{\rm in} \in H^\sigma,$ for $\sigma \geq 0,$ and $m > 2.$ Suppose that $\|\tilde{f}\|_{H^{\sigma'}_m}\leq \delta_0,$ with $\sigma' > \sigma+ 5/2,$ and
 $|T_\parallel-1|+ \delta_0$ sufficient small depending on universal constants and $B_0$. Then, the following holds.
 \begin{itemize}
 \item[a)]
 The Landau damping for $z-$dependent modes:
 $$
 \| |\partial_z|^{1/2}<\nabla, \partial_z t >^\sigma \rho(t)  \|_{L^2_t L^2_x}\lesssim_{\sigma,\sigma',m}\, \|\mathcal G_{\rm in}\|_{H^\sigma_m}.
 $$
 \item[b)]
 If $k_3=0,$and we additionally  have $\sigma > 5/2,$ then for all $\mathbf k_\perp:=(k_1,k_2,0)$ and $ n\in \mathbb N,$ $\exists !$ $b_{n,\mathbf k}= b_{n,\mathbf k}(\frac{q}{m} B_0)\in (n,n+1)$ and coefficients   $r_{\pm n,\mathbf k}$  depending on $\mathcal G_{\rm in}$ such that (with the convention that $r_{-0,\mathbf k}$ is distinct from $r_{0,\mathbf k}),$
 \beq\label{b.exp}
 \hat{\rho}(t, \mathbf k_\perp,0)= \sum_{n=0}^\infty\,r_{n,\mathbf k}\, e^{i \frac{b_{n,\mathbf k} q B_0}{m}t}+ r_{-n,\mathbf k}\, e^{-i\frac{b_{n,\mathbf k}q B_0 }{m} t}, 
 \ene 
 and further, there holds,
  $$
  \left|  r_{\pm n, \mathbf k} \right|\leq \frac{1}{<\mathbf k>^\alpha\, <n>^\gamma}\,  \|\mathcal G_{\rm in}\|_{H^\sigma_m},
  $$
  where $\alpha,\beta,\gamma$ are such that $ \gamma= \min (\beta+1, 2\beta-1), \alpha+2\beta-\frac{1}{2}\leq \sigma,$ and $\beta+1 < m.$
   \end{itemize}
 \end{theorem}
 Item a) of Theorem 1 of \cite{bedro} (Theorem~\ref{theob}  above) is concerned with the Landau damping of the $z-$dependent modes (see also \cite{bernstein}), that is to say, of the modes that depend on the coordinate, $z,$ along the direction of the magnetic field $\mathbf B_0$. In this work we do not consider this problem, since as we work in $1+2$ dimensions our modes only depend on the coordinate $x$ that is orthogonal to the direction of the magnetic field $\mathbf B_0.$
 
 Item b) of Theorem 1 of \cite{bedro} (Theorem~\ref{theob}  above) considers the expansion of the  Fourier coefficients, $\hat{\rho}(t, k_\perp,0)$ of the charge density  fluctuation  in terms of the Bernstein modes, that is to say the Bernstein-Landau paradox, in the case of $3+3$ dimensions. This is the problem that we consider in $1+2$ dimensions in this work. \textcolor{black}{We now proceed to discuss this result.} Let us first show that if $\hat{\rho}(t,k_\perp,0)$ in \eqref{b.exp} is time independent, then it has to be identically zero. Assume then, that $ \hat{\rho}(t,\mathbf k_\perp,0)=  \hat{\rho}(\mathbf k_\perp,0).$ Since all the $b_{n,\mathbf k}$ are different from each other, and different from zero, then one deduces
 \begin{equation}
 \begin{array}{lll}
 0&=& \lim_{T\to \infty}\,\int_0^T \,\hat{\rho}(\mathbf k_\perp,0)\, e^{\mp i\frac{b_{j,\mathbf k}q B_0 }{m} t}\, dt \\
 & = &\lim_{T\to \infty}\, \int_0^T\,\left(   \sum_{n=0}^\infty\,r_{n,\mathbf k}\, e^{i \frac{b_{n,\mathbf k} q B_0}{m}t}+ r_{-n,\mathbf k}\, e^{-i\frac{b_{n,\mathbf k}q B_0 }{m} t}     \right)  e^{\mp i\frac{b_{j,\mathbf k}q B_0 }{m} t}\, dt \\
 &=&
 r_{\pm j,\mathbf k}.
\end{array}
\end{equation}
 Hence, $r_{\pm n, \mathbf k}=0$ for all   $n\in \mathbb N,$ and all $\mathbf k=(\mathbf k_\perp,0).$ Therefore,
    item b) of Theorem 1 of \cite{bedro} (Theorem~\ref{theob}  above)    implies that for the time independent solutions to \eqref{bw.1} the Fourier coefficient of the charge density, $\hat{\rho} (\mathbf k_\perp 0),$ is identically zero for all $\mathbf k_\perp.$ \textcolor{black}{Recall that $\hat\rho(\mathbf k)$ is the Fourier coefficient of $\rho(\mathbf x),$ that is to say,
 $$
 \hat\rho(\mathbf k):= \int_{\mathcal T^3}\, \rho(\mathbf x) \,e^{-i\mathbf k \mathbf x} d\mathbf x.
 $$   
 Then, inverting the Fourier series,
 \beq\label{infou}
 \rho(\mathbf x)= \frac{1}{(2\pi)^3}\, \sum_{\mathbf k \in \mathbb Z^3}\, e^{i\mathbf k\mathbf x}\, \hat{\rho}(\mathbf k).
 \ene  
 Further, by \eqref{infou}
$$
\int_{\mathcal T}\,\rho(x, y, z)\, dz=\frac{1}{(2\pi)^2}\,  \sum_{ (k_1,k_2 )\in \mathbb Z^2}\,  \, e^{i(k_1x+k_2 y)}\, \hat{\rho}(k_1,k_2,0).
$$
Hence, as for time independent solutions, if \eqref{b.exp}  holds, $\hat{\rho}(\mathbf k_\perp,0)=0,$ we have proven that for the    solutions given in Item b) of Theorem 1 of \cite{bedro} (Theorem~\ref{theob}  above), if their are time independent, necessarily
\beq\label{average}
\int_{\mathcal T}\,\rho(x_\perp, x_3)\, dx_3=0.
\ene
  However, in Appendix~ B we construct an explicit family of time-independent solutions to \eqref{bw.1}, that satisfy the hypotheses 
     of item b) of Theorem 1 of \cite{bedro}  (Theorem~\ref{theob}  above)    with   $\hat{\rho} (k_\perp,0)$ that is not identically zero, and where \eqref{average} does not holds.} This shows that \textcolor{black}{the result of item b) of Theorem 1 of \cite{bedro}  (Theorem~\ref{theob}  above)  is devoted to the behaviour of a special class of solutions.}
     \textcolor{black}{This is clear by comparison with the physical literature \cite{bald:rol,bernstein,suku} which is based on the study of a dispersion
   relation. The dispersion relation  is mathematically justified  in \cite{bedro}, in particular with  a  summability argument 
   of the contributions of all poles of the dispersion relation. Nevertheless the  pole $1/z$, 
    that corresponds to  time-independent solutions, is discarded in equation (2.14)
   in \cite{bedro}, and in this sense, the works \cite{bald:rol,bernstein,bedro} focus on a subclass, or special class,  of solutions.
   In our work, we do not make such hypothesis or restriction and that is why we recover a time-independent solution 
   as in \cite[eq. (55)]{suku}.}
      Note that in Appendix~ B we write the model for ions, as in \cite{bedro}, for the purpose of making the comparison with the results of \cite{bedro} more transparent. In the rest of our paper we write the model for electrons. Actually, the models for ions and electrons are the same, up to a change in the sign of the cyclotron frequency and of the electric field. See Remark  \ref{rem:plusmoins}.
    This is actually in agreement with our theoretical expansions  in \eqref{4.152.0}-\eqref{4.152} that show that in general the charge density fluctuations   have a time-independent part and a time-dependent part. Further,  our result  in \eqref{4.152.0}-\eqref{4.152}  solves, in the case of one dimension in space and two dimensions in velocity, the  problem posed in Remark  3 of \cite{bedro} of justifying the expansion in the Bernstein modes of the charge density fluctuation, without the regularity in space  and decay in velocity that they  assume in  Theorem 1 of \cite{bedro} (Theorem~\ref{theob} above).  Our model is one dimensional in space and two dimensional in velocity. However,   there is no difficulty to write it  in three  dimensions  in space and velocity, because Maxwellian functions have a natural compatibility with separation of variables techniques. In principle, the extension of our  results  to three dimensions in space and  velocity is possible with due attention paid to the anisotropy introduced by the magnetic field. It is left for further research.

%

The organization of this work is as follows.  In Section \ref{sectionvlasamp} we introduce the \textcolor{black}{magnetized} Vlasov-Poisson and the \textcolor{black}{magnetized}  Vlasov-Amp\`ere systems, and we prove their equivalence. In Section \ref{notations}
we   give the notations and definitions that we use. In sections \ref{noef} 
we consider the case \modifsF{of a pure  \textcolor{black}{magnetized} Vlasov equation without coupling}. 
 We construct a selfadjoint realization of the  \textcolor{black}{magnetized}  Vlasov operator, we explicitly compute the eigenvalues and we explicitly construct an orthonormal system of eigenfunctions that is complete, i.e., it is a basis of the Hilbert space. 
In Section \ref{secvlasamp} we construct a selfadjoint realization of the \textcolor{black}{magnetized} Vlasov-Amp\`ere operator, we compute the eigenvalues, and we construct an orthonormal systems of eigenfunctions  that is complete, that is to say that is a basis of the Hilbert space. In Section \ref{paradox} 
we obtain  a representation of  the general solution to the \textcolor{black}{magnetized}  Vlasov-Amp\`ere system as an expansion in our orthonormal basis of eigenfunctions. In particular we prove the convergence of the Bernstein expansion \cite{bernstein}, \cite{bedro}, under  optimal conditions on the initial state. In Section \ref{optheor} we give a operator theoretical proof of the existence of the Bernstein-Landau paradox, with an argument based on the Weyl theorem for the invariance of the essential spectrum. In Section \ref{numerical} 
we illustrate our results with  numerical calculations. In  Appendix~A  we study the properties of the secular equation. Finally, in Appendix~ B we construct explicit families of time-independent solutions to  the linearized  magnetized Vlasov-Poisson system.
 
\section{The \textcolor{black}{magnetized} Vlasov-Poisson and the \textcolor{black}{magnetized} Vlasov-Amp\`ere systems}
\label{sectionvlasamp}\sss

We adopt the Klimontovitch approach \cite{klim,gopau} where the Newton equation of a very large number of charged particles with velocity $v$
moving in an electromagnetic field is 
approximated by a continuous density function  $f(t,x,v)\geq 0.$  The variable $t$ is time. We assume that the charged particles undergo a one dimensional motion, and that the real variable $x$  is the position of the charged particles. Furthermore, we suppose that the velocity, $v,$ of the charged particles is two dimensional, i.e., 
$v=(v_1,v_2) \in \mathbb R^2.$ Further, we take the motion of the charged particles along the first coordinate axis of the velocity of the charged particles. The density function is a solution  of a Vlasov equation,

\begin{equation} \label{bd:eq:1}
\partial_t f + v_1 \partial _ x f +\mathbf F \cdot \nabla_v f=0.
\end{equation}
We assume, for simplicity, that the motion of the charged particles  is a $2\pi$-periodic oscillation, that is  a usual assumption \cite{degond}.  Hence, we look for solutions to \eqref{bd:eq:1}, 
 $f(t, x,v),$ for $t \in \mathbb R,  x \in [0,2 \pi], v=(v_1,v_2) \in \mathbb R^2,$  that are periodic in $x,$  i.e., $f(t,0,v)= f(t, 2\pi, v).$ 
The electromagnetic Lorentz force,
 \beq\label{lorenz}
 \mathbf F(t,  x) = \frac{q}{m} \left( \mathbf E (t, x)+ v \times \mathbf B(t,\mathbf x)\right),
 \ene
  is divergence free with respect to the velocity variable, that is $
\nabla_ v \cdot \mathbf F=0$.   The Maxwell's equations are simplified,
assuming  that the magnetic field $\mathbf B (t,\mathbf x)=\mathbf B_0$ is constant in space-time. Following the convention adopted in \cite{bald:rol,suku}, we suppose that the two dimensional velocity $v$ is perpendicular to the constant magnetic field, i.e., $\mathbf B_0= (0,0,  B_0), B_0 >0.$ Moreover, we assume that the electric field is directed along the first coordinate axis, $\mathbf E(t,x)=(E(t,x), 0, 0 ).$
We adopt a  convenient normalization adapted to electrons, that  is $ q_{\rm ref}=-1$ and $m_{\rm ref}=1,$ where $ q_{\rm ref}$ is the charge of the electron, and $m_{\rm ref}$ is the mass of the electron. 
The electric field satisfies the Gauss law,
\beq\label{gauss-0}
\partial_x  E(t,x)= 2\pi - \int_{\mathbb R^2}\, f dv,
\ene
where    $2\pi $ is  the constant density of the heavy ions, that do not move. We take the density of the ions equal to $2\pi$ to simplify some of the calculations below.
The term 
$- \int_{\mathbb R^2} \, f\, dv
$
is the charge density of the  particles with charge $-1.$
 
With these notations  and normalizations \eqref{bd:eq:1}, and \eqref{gauss-0} are written as
the following system,
\begin{equation} \label{bd:eq:4}
\left\{
\begin{aligned}
&\partial_t f + v_1 \partial_x f -  E \partial_{v_1} f + \omega_c \left( -v_2 \partial_{v_1} +v_1 \partial_{v_2}  \right) f= 0, \\
&\partial_x  E(t,x)= 2\pi - \int_{\mathbb R^2}\, f dv.
\end{aligned}
\right.
\end{equation}
We  denote the cyclotron frequency by $\omega_{\text{\rm{c}}} := B_0$.

\begin{remark} \label{rem:plusmoins}
The  model written for  positive ions instead of (negatively charged) electrons is similar to (\ref{bd:eq:4}), with the only modification
that the sign in front of the electric field and $\omega_c$  is changed  in both equations.
\end{remark} 

We retain the potential part of the  electric field
\begin{equation} \label{bd:eq:2}
 E (t, x)= - \partial_x \varphi(t, x),
\ene
where the  potential $ \varphi(t,\mathbf x)$ is a solution to the Poison equation,
\beq\label{poisson}
 \quad - \Delta \varphi= 2\pi - \int_{\mathbb R^2}\, f dv.
\ene
The electric field and the potential are assumed to be periodic with period $2\pi,$ i.e.  $E(t,0)= E (t,2\pi), \varphi(t,0)= \varphi(t, 2\pi).$    
Note that since    the potential $\varphi(t,x)$ is periodic it follows from \eqref{bd:eq:2} that the mean value of the electric field is zero,
\beq\label{mean0}
\int_0^{2\pi}\, E(t,x)\, dx=0.
\ene

Two important properties of the \textcolor{black}{magnetized} Vlasov-Poisson system (\ref{bd:eq:4}, \eqref{bd:eq:2}, and \eqref{poisson} are that the density function satisfies the maximum principle 
$$
\inf_{  (x, v) \in [0,2\pi]\times \mathbb R^2} f_{\rm ini} (x, v) \leq f (t, x,  v) \leq \sup_{ (x, v) \in [0,2\pi]\times \mathbb R^2 } f_{\rm ini}(x, v), 
$$
where $ f_{\rm ini}$ is the initial value of the solution, $f(t,x,v),$
and that the total energy is constant in time,
\begin{equation} \label{bd:eq:20}
\frac d {dt} \left(
\int_{[0,2\pi]\times \mathbb R^2} \frac{ |v|^2}2 f dxdv + \int_{[0,2\pi]} \frac{|E|^2}2 dx
\right)=0.
\end{equation}


%

Following \cite{bernstein}, a linearization   of the equations around a homogeneous Maxwellian equilibrium state $ f_0(v),$  where, $f_0(v):= e^{\frac{-v^2}{2}}$ is performed.
Here the Maxwellian distribution is normalized for $T_{\rm ref}\, k_{\rm B}=1,$  where $T_{\rm ref}$ is the reference temperature and $k_{\rm B}$ is Boltzmann's constant. 
It corresponds  to  the expansion
\beq\label{1.1}
f(t,x,v)=f_0(v) + \varepsilon \sqrt{ f_0(v)} u(t,x,v) +O(\varepsilon^2), 
\ene
and 
\beq\label{1.2}
 E(t,x)=E_0+ \varepsilon F(t,x)+O(\varepsilon^2),
\ene
with a null reference electric field $E_0=0$.
Inserting \eqref{1.1} and \eqref{1.2} into \eqref{bd:eq:4}, and keeping  the terms up to linear in $\varepsilon, $ one gets the linearized \textcolor{black}{magnetized} Vlasov-Poisson system written as, 
\begin{equation} \label{bd:eq:7}
\left\{
\begin{aligned}
&\partial_t u + v_1 \partial_x u + F v_1 \sqrt {f_0}+ \omega_c \left( -v_2 \partial_{v_1} +v_1 \partial_{v_2}  \right) u= 0, \\
&\partial_x  F=- \int_{\mathbb R^2} u  \sqrt{f_0} dv , \\
&\int_{[0,2\pi]} F=0,
\end{aligned} 
\right.
\end{equation}
where in the third equation we have added the constraint that the mean value of the electric field $F$ is zero, as in \eqref{mean0}. 
 Moreover, the electric field $F(t,x) = - \partial_x \varphi(t,x)$ is obtained from a potential as in \eqref{bd:eq:2},
 where  the potential is periodic, $\varphi(t,0)= \varphi(t,2\pi),$  and it solves the Poisson equation,
\beq\label{poisson2}
 \quad -  \Delta \varphi= - \int_{\mathbb R^2}\, u  \,\sqrt{f_0} dv. 
\ene
Observe that the second equation in \eqref{bd:eq:7} is the Gauss law, 
\beq\label{gauss}
\partial_x F(t,x)= \rho(t,x),
\ene
where  $\rho(t,x)$ is the charge density fluctuation of the perturbation of the Maxwellian equilibrium state,
\beq\label{density}
\rho(t,x):=  - \int_{\mathbb R^2}\, u(t,x,v)  \,\sqrt{f_0}(v) dv.
\ene

The study of the solutions to the  \textcolor{black}{magnetized} Vlasov-Poisson system  is the standard method to analyze the dynamics of a very large number of charged particles moving in the presence of a constant external magnetic field.   For the case of the Bernstein-Landau paradox see, for example, \cite{bernstein}, \cite{suku}, \cite{stix}[section 9.16], \cite{swanson}[[section 4.4.1] and \cite{bedro}.    We now present an alternate method  to study this problem.
\modifsF{
In the full Maxwell equations one of the equation  is the Amp\`ere equation
\beq\label{amp}
\partial_t  F= \int_{\mathbb R^2}\, v_1\, u \, \sqrt{f_0} 
dv,
\ene
 where we have taken the dielectric constant $\varepsilon_0=1.$  We consider here the following modified Amp\`ere equation
 \beq\label{amp2}
\partial_t  F=I^\ast \int_{\mathbb R^2} \, v_1 \sqrt{f_0} \, u\, dv,
\ene
where $I^\ast$ is the space operator such that $I^*g = g-[ g]$ 		and   the mean value in space of a function $g$ is  denoted  by $[ g]$,
 that is to say,
$
I^\ast g(x):= g(x)- \frac{1}{2\pi}\, \int_{0}^{2\pi}\,g(y)\, dy
$. 
With this convention the  \textcolor{black}{magnetized}  Vlasov-Amp\`ere system is written as follows,
\begin{equation} \label{vlasamp}
\left\{
\begin{aligned} 
&\partial_t u + v_1 \partial_x u + F v_1 \sqrt {f_0}+ \omega_c \left( -v_2 \partial_{v_1} +v_1 \partial_{v_2}  \right) u= 0, \\
&\partial_t  F=I^\ast \int_{\mathbb R^2}\,  v_1\, \sqrt{f_0} \, u  
dv.
\end{aligned} 
\right.
\end{equation}
To the \textcolor{black}{magne:tized} Vlasov-Amp\`ere system (\ref{vlasamp}), we  add  conditions for $F_{\rm ini}:=F(0,\cdot)$ and $u_{\rm ini} = u(0,\cdot,\cdot)$ :  the integral constraint,
 \beq\label{a.1}
 \int_0^{2 \pi} \,  F_{\rm ini}\,dx=0,
 \ene 
is satisfied at initial time, and  the Gauss law  \eqref{gauss}, \eqref{density}  is also satisfied at the initial time,
\begin{equation} \label{a.2}
\frac d {dx} F_{\rm ini}=- \int _{\mathbb R^2}u_{\rm ini} \sqrt{f_0} dv.
\end{equation}
\begin{lemma}
The linearized \textcolor{black}{magnetized} Vlasov-Poisson system \eqref{bd:eq:7} is equivalent to the \textcolor{black}{magnetized} Vlasov-Amp\`ere system \eqref{vlasamp} with initial conditions that satisfy \eqref{a.1}, \eqref{a.2}.  
\end{lemma}
\begin{proof}
Let $(u,F)$ be a solution the  \textcolor{black}{magnetized} Vlasov-Amp\`ere system \eqref{vlasamp} that satisfy \eqref{a.1}, \eqref{a.2}.  It follows from the Amp\`ere equation that 
$$
\partial _t  \int_0^{2\pi}\, F(t,x)\,dx = \int_0^{2\pi}\,I^\ast \int_{\mathbb R^2}\,  v_1\, \sqrt{f_0} \, u=0 
$$
and consequently the integral constraint \eqref{a.1} is propagated to all times.
 The Gauss law \eqref{a.2}    is propagated  also to all times by the \textcolor{black}{magnetized}  Vlasov-Amp\`ere system, as we proceed to prove.
Multiplying the first equation in  \eqref{vlasamp} by $\sqrt{f_0},$ integrating  in $v$ over $\mathbb R^2,$ using that $f_0$ is an even function of $|v|$ and using integration by parts, we prove  the  following continuity  equation, 
\beq\label{a.3}
\partial_t \int_{\mathbb R^2}  u \sqrt {f_0} dv + \partial_x  \int_{\mathbb R^2}   v_1 u \sqrt {f_0} dv=0.
\ene
Deriving  
 \eqref{amp2} with respect to $x$  we obtain,
$
0=  \partial_x \left( \partial_t F - I^*  \int_{\mathbb R^2}  v_1 u \sqrt {f_0} dv \right)
= \partial_x \left( \partial_t F -   \int_{\mathbb R^2}   v_1 u \sqrt {f_0} dv \right)$, 
 because $\partial_x= \partial_x I^*$. Then, by \eqref{a.3}
 $$ 0 = \partial_t\left( \partial_x F + \int_{\mathbb R^2}\,u \,\sqrt{f_0}\, dv\right),$$
 from which the Gauss law follows  for all times. We have proven that a solution to the  \textcolor{black}{magnetized} Vlasov-Amp\`ere system \eqref{vlasamp} that satisfies the initial conditions \eqref{a.1}, \eqref{a.2} solves the \textcolor{black}{magnetized} Vlasov-Poisson system \eqref{bd:eq:7}.\\
On the contrary let $(u,F)$ be a solution to the \textcolor{black}{magnetized}  Vlasov-Poisson system  \eqref{bd:eq:7}. Then  by the second equation in \eqref{bd:eq:7} and \eqref{a.3},
$$
0= \partial_x  \partial_t F   +  \partial_t \int_{\mathbb R^2}  u  \sqrt{f_0} dv= \partial_x \left( \partial_t F - \int_{\mathbb R^2}   v_1 u \sqrt {f_0} dv \right).
$$
So 
$
 \partial_t F =  \int_{\mathbb R^2}   v_1 u \sqrt {f_0} dv+  C(t)$, 
  where $C(t)$ is  constant in space. Then,
  $
 \partial_t I^* F =  I^* \int_{\mathbb R^2}   v_1 u \sqrt {f_0} dv$. 
  But $F$ has zero mean value, so $I^* F =F$, and it follows that   the Amp\`ere law in (\ref{amp2}) holds. Hence, the \textcolor{black}{magnetized}Vlasov-Poisson system implies the \textcolor{black}{magnetized}  Vlasov-Amp\`ere system \eqref{vlasamp} and the initial conditions \eqref{a.1}, \eqref{a.2}.  
\end{proof}
}
From now on we only consider the \textcolor{black}{magnetized} Vlasov-Amp\`ere system \eqref{vlasamp} with conditions  \eqref{a.1}, \eqref{a.2}.
A fundamental energy relation is easily shown for solutions of the \textcolor{black}{magnetized} Vlasov-Amp\`ere formulation (\ref{vlasamp})
\begin{equation} \label{bd:eq:50}
\frac{d}{dt} \left(
\int_{[0,2 \pi]\times \mathbb R^2}\frac{u^2}{2} dxdv + \int_{[0,2\pi]} \frac{ F^2}2 dx
\right)=0.
\end{equation}
It is the counterpart of the energy identity (\ref{bd:eq:20}), so the term
$
\int_ {[0,2 \pi]\times \mathbb R^2}\int_\mathbf v \frac{u^2}{2} dxdv$  
is identified
with the kinetic energy of the negatively charged particles, and the term $
\int_{[0,2\pi]} \frac{ F^2}2 dx$
is the energy of the electric field.
 This identity is known since \cite{kruskal,antonov}.
As we  show in the next sections, the identity \eqref{bd:eq:50} is the basis of our formulation of the  \textcolor{black}{magnetized}  Vlasov-Amp\`ere system as a
 Schr\"odinger equation in Hilbert space, where \textcolor{black}{magnetized}  the Vlasov-Amp\`ere operator plays the role of the selfadjoint Hamiltonian.
 


\section{Notations and Definitions}
\label{notations}
\sss

We will  write the \textcolor{black}{magnetized}  Vlasov-Amp\`ere system as a Schr\"odinger equation with a selfadjoint Hamiltonian in an appropriate Hilbert space. We find it convenient to borrow some terminology from quantum mechanics.
For this purpose, we first introduce some notations and definitions. We designate by $\mathbb R^+$ the positive real semi-axis, i.e., $\mathbb R^+:= (0,\infty),$ and by $\mathbb R^2$ the plane.  The set of all integers is denoted by $\mathbb Z$ and the set of all nonzero integers by $\mathbb Z^\ast.$  The positive natural numbers are designated by $\mathbb N^*.$  By $\mathbb C$  we designate the complex numbers. We denote by $C$ a generic constant whose value does not have to be the same when it appears in different places. 
By  $C^\infty([0,2\pi])$ we designate the set of all infinitely differentiable functions in $[0,2\pi],$  and by $C^\infty_0(\mathbb R^2)$ we denote the set of all infinitely differentiable functions in $\mathbb R^2$ with compact support.
Let $\mathcal B$ be a set of vectors in a Hilbert space, $\mathbb H.$ We denote by $\text{\rm Span}[\mathcal B]$ the closure in the strong convergence in $\mathbb H$ of   all finite linear combinations of elements of $\mathcal B,$ in other words,
$$
\text{\rm Span}[\mathcal B]:= \text{\rm{closure}}\left\{ \sum_{j=1}^N \alpha_j X_j: \alpha_j \in \mathbb C, X_j \in \mathcal B, N \in \mathbb N^* \right\}.
$$
Let $ \mathcal M$ be a subset of a Hilbert space $\mathbb H.$ We define the orthogonal complement of $\mathcal M$,  in symbol, $\mathcal M^\perp,$ as follows,
$$
\mathcal M^\perp:= \{ f \in \mathbb H : (f,u)_{\mathbb H}=0, \, \text{\rm{ for all}}\, u \in \mathcal M \}.
$$ 
Let $ \mathbb H$ be  a Hilbert space, and let  $ \mathbb H_j, j=1,\dots, N, 2 \leq N \leq \infty, $ be mutually orthogonal closed subspaces of $\mathbb H$, that is to say,
$$
\mathbb H_j \subset \mathbb H_m^\perp, \,\text{\rm{and} }\, \mathbb H_m\subset \mathbb H_j^\perp, \qquad j\neq m, 1\leq j,m \leq N.
$$
Note that if $\mathbb H_j$ and $\mathbb H_m$ are mutually orthogonal, then one has
$
(f,u)_{\mathbb H}=0$, $ f \in \mathbb H_j$, $ u \in  \mathbb H_m$. 
We say that $\mathbb H$ is the direct sum of the  $ \mathbb H_j, j=1,\dots, N, 2 \leq N \leq \infty, $  mutually orthogonal closed subspaces of $\mathbb H$, and we write,
$$
\mathbb H=\oplus_{j=1}^N \,  \mathbb H_j ,
$$
if for any $f \in \mathbb H,$ there are $f_j \in \mathbb H_j,  j=1,\dots, N,$ such that,
$
f= \sum_{j=1}^N\, f_j$.  
Note that the $f_j, j=1,\dots, N$ are unique for a given $f,$ and that 
$
\|f\|^2_{\mathbb H}= \sum_{j=1}^N \, \|f_j\|^2_{\mathbb H}$.

Let  $A$  be an operator in a Hilbert space $\mathbb H$,   and let us denote by $D[A]$ the domain of $A$. We say that  the operator $B$ is an extension of  the operator $A,$ in symbol,  $A \subset B,$ if $D[A]\subset D[B],$ and if $Au= Bu,$ for all $ u \in D[A].$  Suppose that the domain of $A$ is dense in $\mathbb H.$ We denote by $A^\dagger$ the adjoint of $A$,  that is defined as follows,
$$
D[A^\dagger]:= \{ v \in \mathbb H : (Au,v)_{\mathbb H}= (u, f)_{\mathbb H}, \text{\rm for some}\, f \in \mathbb H, \text{\rm and for all} \,u \in D[A] \},
$$
and
$$
A^\dagger v= f, \qquad v \in D[A^\dagger].
$$
We say that $A$ is symmetric if $ A \subset A^\dagger,$ and that $A$ is selfadjoint if $ A=A^\dagger,$ that is to say  if $D[A]= D[A^\dagger],$ and $Au=A^\dagger u, u \in D[A]= D[A^\dagger].$ An essentially selfadjoint operator  has only one selfadjoint extension.  For any operator $A$ we denote by $\text{\rm Ker}[A]:= \{ u \in D[A] : Au=0 \}$ the set of all eigenvectors of $A$ with eigenvalue zero.
 For more information on  the theory of operators in Hilbert space the reader can consult \cite{ka} and \cite{rs1}.

 We denote by $L^2(0,2\pi)$ the standard Hilbert space of 
functions that are square integrable in $(0,2 \pi).$ Furthermore, we designate by $L^2_{\text{\rm 0}}(0, 2 \pi)$ the closed subspace of $L^2(0, 2\pi)$ consisting of all functions with zero mean value, i.e.,
\beq\label{2.1}
L^2_{\text{\rm 0}}(2,\pi):= \left\{ F \in  L^2(0,2\pi): \int_0^{2\pi}\, F(x)\,dx=0\right\}.
\ene
Note that since all the functions in $L^2(0,2\pi)$ are integrable over $(0,2 \pi)$ the space $L^2_{\text{\rm 0}}(0,2\pi)$ is well defined.
Further, we denote by $L^2(\mathbb R^2)$ the standard Hilbert space of all functions that are square integrable in $\mathbb R^2.$  Let us denote by $\mathcal  A$ the tensor product of $L^2(0,2\pi)$ and of $L^2(\mathbb R^2),$ namely,
\beq\label{2.2}
\mathcal A:= L^2(0,2\pi)\otimes  L^2(\mathbb R^2).
\ene
For the definition and the properties of  tensor products of Hilbert spaces the reader can consult  Section 4 of Chapter II of \cite{rs1}.
We often make use of the fact that the tensor product of an orthonormal basis in $L^2(0,2\pi)$ and an orthonormal basis in $L^2(\mathbb R^2)$ is an orthonormal basis in $\mathcal A.$ As shown in  Section 4 of Chapter II of \cite{rs1}, the space $\mathcal A$ can be identified with the standard Hilbert space $L^2((0, 2 \pi) \times \mathbb R^2)$ of square integrable functions in $(0, 2 \pi) \times \mathbb R^2$ with the scalar product,
$$
\left(u,f\right)_{L^2((0, 2 \pi) \times \mathbb R^2)}:= \int_{(0,2\pi) \times \mathbb R^2}\, u(x,v)\, \overline{f(x,v)}\, dx\, dv,
$$
where $ x \in (0,2\pi)$ and $v=(v_1,v_2)\in \mathbb R^2.$
Our space of physical states, that we denote by $\mathcal H,$ is defined as the direct sum of $\mathcal A$ and  $L^2_0(0,2\pi).$
\beq\label{2.3}
\mathcal H:= \mathcal A \oplus L^2_0(0,2\pi).
\ene
  We find it convenient to write 
 $\mathcal H$ as the space of the column vector-valued functions,
$
\begin{pmatrix}
u\\
F
\end{pmatrix}
$
where $u(x,v) \in \mathcal A$ and $F(x)\in L^2_0(0,2 \pi).$ The scalar product in $\mathcal H$ is given by,

$$
\left(\begin{pmatrix}
u\\
F
\end{pmatrix}, \begin{pmatrix}
f\\G\end{pmatrix}\right)_{\mathcal H}:=(u,f)_{\mathcal A}+ (F,G)_{L^2(0,2\pi)}.
$$
Note that by the identity \eqref{bd:eq:50} \textcolor{black}{the $\mathcal H$-norm of 
 the solutions to the  \textcolor{black}{magnetized} Vlasov-Amp\`ere system  is constant in time}. This is the underlying reason why we will be able in later sections to formulate the 
  \textcolor{black}{magnetized} Vlasov-Amp\`ere system as a Schr\"odinger equation in $\mathcal H$ with a selfadjoint realization of the \textcolor{black}{magnetized} Vlasov-Amp\`ere operator playing the role of the Hamiltonian. Moreover, the square of the norm of $\mathcal H$ is the constant energy of the solutions to the  \textcolor{black}{magnetized} Vlasov-Amp\`ere system. 

Let us denote by $H^{(1)}(0,2\pi)$ the standard Sobolev space \cite{af} of all functions in $L^2(0,2\pi)$
such that its derivative in the  distribution sense is a function in $L^2(0,2\pi),$ with the scalar product,
$$
(F,G)_{H^{(1)}(0,2\pi)}:= (F,G)_{L^2(0,2\pi)}+ (\partial_x F, \partial_x G)_{L^2(0,2\pi)}.
$$
We designate by $H^{(1,0)}(0,2\pi)$ the closed subspace of $H^{(1)}(0,2\pi)$ that consists of all functions in $ F\in   H^{(1)}(0,2\pi)$ such that $F(0)= F(2\pi)$ and that have mean zero. Namely, 
$$
H^{(1,0)}(0,2\pi):=\left\{ F \in H^{(1)}(0,2 \pi):  F(0)=F(2\pi), \,\text{\rm and}\, \int_0^{2\pi}\, F(x)\, dx=0   \right\}.
$$
Note \cite{af} that as  the functions in $ H^{(1)}(0,2 \pi)$  have a continuous extension to $[0,2\pi],$ the space $H^{(1,0)}(0,2\pi)$ is well defined.
 
 We denote by $L^2(\mathbb R^+, r dr)$ the standard Hilbert
space of functions defined on $\mathbb R^+$ with the scalar product,
$$
(\tau,\eta)_{L^2(\mathbb R^+, r dr)}:= \int_0^\infty\, \tau(r)\, \overline{\eta(r)}\, r \, dr.
$$ 
\section{\modifsF{The \textcolor{black}{magnetized}  Vlasov equation without coupling}}\label{noef}
\sss
In this section we consider the case without electric field, i.e. the \textcolor{black}{magnetized}  Vlasov equation.  The results of this section  will be useful in the study of the full  \textcolor{black}{magnetized} Vlasov-Amp\`ere system, that we carry over in  Sections~\ref{secvlasamp}.

The \textcolor{black}{magnetized}  Vlasov equation can be written as the following Schr\"odinger equation in $\mathcal A,$

\beq\label{3.1}
i\partial_t u=  i \left(-v_1 \partial_x+\omega_{\text{\rm c}}(v_2\,\partial_{v_1} - v_1\, \partial_{v_2})\right) u.
\ene
In the following proposition we obtain a complete orthonormal system of eigenfunctions for the \textcolor{black}{magnetized} Vlasov equation \eqref{3.1}. \modifsF{To this end, we introduce the polar coordinates $(r,\varphi)$ of the velocity $v\in \R^2$.}
\begin{prop}\label{prop3.1}
Let $ \{ \tau_j\}_{j=1}^\infty$ be an orthonormal basis of $L^2(\mathbb R^+, r dr).$  Let $ \varphi \in [0,2\pi), r >0,$ be polar coordinates in
 $\mathbb R^2, v_1= r \cos\varphi, v_2= r \sin\varphi.$ For $(n,m,j)\in \mathbb Z^2 \times \mathbb N^*$ we 
define,

\beq\label{3.4}
u_{n,m,j}:=  \frac{ e^{i n (x-\frac{v_2}{\omega_{\text{\rm c}}})
}}{\sqrt{2 \pi}}    \,\frac{e^{im \varphi}}{\sqrt{2\pi}}\, \tau_j(r).
\ene 
Then, the $u_{n,m,j}, (n,m,j) \in \mathbb Z^2 \times \mathbb N^*$ are an orthonormal basis in $\mathcal A.$ Furthermore, each  $ u_{n,m,j}$ is an eigenfunction for the  \textcolor{black}{magnetized} Vlasov equation \eqref{3.1} with  eigenvalue  $\lambda^{(0)}_{m}= m\, \omega_{\text{\rm c}},$
\beq\label{3.5}
 i \left(-v_1 \partial_x+\omega_{\text{\rm c}}(v_2\,\partial_{v_1} - v_1\, \partial_{v_2})\right) u_{n,m,j}= \,\lambda^{(0)}_m\, u_{n,m,j}, \qquad (n,m,j)\in \mathbb Z^2 \times \mathbb N^*.
\ene
Moreover, the eigenvalues $\lambda^{(0)}_m, m \in \mathbb Z,$ have infinite multiplicity.
\end{prop}

\begin{proof} We first prove that the  $u_{n,m,j}, (n,m,j)\in \mathbb Z^2 \times \mathbb N^* $ are an orthonormal basis in $\mathcal A.$ Clearly, it is an orthonormal system. To prove that it is a basis it is enough to prove that if a function in $\mathcal A$ is orthogonal to all the $u_{n,m,j}, (n,m,j)\in \mathbb Z^2 \times \mathbb N^*,$ then, it is the zero function. Hence, assume that $ u \in \mathcal A$ satisfies,
\beq\label{3.6}
\left(u, u_{n,m,j}\right)_{\mathcal A}=0,\qquad  (n,m,j) \in \mathbb Z^2 \times \mathbb N^*.
\ene
Denote $
g_n(v):= \int_0^{2\pi}\, e^{-inx}\,u(x,v)\,dx$. 
By the Cauchy-Schwarz inequality, one has
$
\left|g_n(v)\right|^2 \leq 2 \pi \int_0^{2\pi}\, |u(x,v)|^2 \,dx$. 
 Further, since $ u \in \mathcal A,$ it follows that $ g_n \in L^2(\mathbb R^2).$
By \eqref{3.6}, for each  fixed $ n \in \mathbb Z,$   
$$
\int_{(0,2\pi)\times \mathbb R^+}\, g_n(v)\, e^{ i n\frac{v_2}{\omega_{\text{\rm c}}}} \,e^{-im\varphi}\, \overline{\tau_j(r)}\, d\varphi\, r\, dr=0, \qquad  (m,j)\in \mathbb Z \times \mathbb N^*.
$$
As the functions $\frac{1}{\sqrt{2\pi}}\, e^{im\varphi}  \,  \tau_j(r), m \in \mathbb Z, j \in \mathbb N^*$ are an orthonormal basis in $L^2(\mathbb R^2),$ 
one has that 
$g_n(v)\, e^{ i n\frac{v_2}{\omega_{\text{\rm c}}}}=0$ 
 for a.e. $v \in \mathbb R^2.$ Moreover, as $e^{ i n\frac{v_2}{\omega_{\text{\rm c}}}}$ is never zero, we obtain,
 $g_n(v)=0,$  for a.e. $v \in \mathbb R^2,$ i.e.,
$
\int_0^{2\pi}\, e^{-inx}\,u(x,v)\,dx = 0, n \in \mathbb Z$.
As the functions  $\frac{1}{\sqrt{2\pi}} e^{inx}, n \in \mathbb Z$ are an orthonormal basis in $L^2(0,2\pi)$,   it follows that $u(x,v)=0.$ This completes the proof that the
$u_{n,m,j}, (n,m,j)\in \mathbb Z^2 \times \mathbb N^*,$ are an orthonormal basis of $\mathcal A.$
 Equation \eqref{3.5} follows from a simple calculation using  that  $\partial_{v_1}= \frac{v_1}{r}\, \partial_r - \frac{v_2}{r^2}\, \partial_\varphi ,$  $\partial_{v_2}= \frac{v_2}{r}\, \partial_r + \frac{v_1}{r^2}\, \partial_\varphi ,$ and $v_2\, \partial_{v_1}-v_1\, \partial_{v_2}=- \partial_\varphi.$ Note that the eigenvalues $\lambda^{(0)}_m$ have infinite multiplicity because all the $u_{n,m,j}$ with $m$ fixed and  $n\in \mathbb Z, j \in \mathbb N^*$ are orthogonal eigenfunctions for $\lambda^{(0)}_m.$  
\end{proof}

Let us denote by $h_0$  the formal  \textcolor{black}{magnetized} Vlasov  operator with periodic boundary conditions in $x,$ that we define as follows,
\beq\label{fv}
h_0 u:=  i \left(-v_1 \partial_x+\omega_{\text{\rm c}}(v_2\,\partial_{v_1} - v_1\, \partial_{v_2})\right) u,
\ene
with domain,
\beq\label{dfv}
D[h_0]:= \mathcal D,
\ene
where by $\mathcal D$ we denote the following space of test functions,
\beq\label{test}
\mathcal D:=\{ u\in C^\infty_0([0,2\pi]\times \mathbb R^2): \frac{d^j}{dx^j}u(0, v)= \frac{d^j}{dx^j}u(2\pi, v), j=1,\dots\},
\ene
where by $ C^\infty_0([0,2\pi]\times \mathbb R^2)$ we designate the space of all infinitely differentiable functions, defined in $[0,2\pi]\times \mathbb R^2,$ and that have compact support in $[0,2\pi]\times \mathbb R^2.$

We will construct a selfadjoint extension of $h_0.$ For this purpose, we first introduce some definitions. Let us denote by $l^2(\mathbb Z^2 \times \mathbb N^*)$ the standard Hilbert space of square summable sequences, $s= 
\left\{s_{n,m,j},   (n,m,j)\in \mathbb Z^2 \times \mathbb N^* \right\}$ with the scalar product,
$$
(s,d)_{l^2(\mathbb Z^2 \times \mathbb N^*)}:= \sum_{(n,m,j)\in \mathbb Z^2 \times \mathbb N^*}\, s_{n,m,j}\, \overline{ d_{n,m,j}}.
$$
Let $\mathbf U$ be the following unitary operator from $\mathcal A$ onto $ l^2(\mathbb Z^2 \times \mathbb N^*),$

\beq\label{3.8}
\mathbf U u:=\left\{ (u, u_{n,m,j})_{\mathcal A}, (n,m,j) \in \mathbb Z^2 \times \mathbb Z \right\}.
\ene
We denote by $\widehat{H_0}$ the following operator in $l^2(\mathbb Z^2 \times \mathbb N^*),$
\beq\label{3.9}
\left\{(\widehat{H_0}\, s )_{n,m,j},(n,m,j)  \in  \mathbb Z^2 \times \mathbb N^* \right\}:=\left\{ \lambda^{(0)}_m s_{n,m,j}, (n,m,j)  \in  \mathbb Z^2 \times \mathbb N^* \right\},
\ene
with domain, $D[\widehat{H_0]},$ given by,
\beq\label{3.10}
D[\widehat{H_0}]:= \left\{  \left\{ s_{n,m,j}, (n,m,j) \in  \mathbb Z^2 \times \mathbb N^*\right\}  \in l^2(\mathbb Z^2 \times \mathbb N^*) :  \left\{ \lambda^{(0)}_m\, s_{n,m,j}, (n,m,j)\in  \mathbb Z^2 \times \mathbb N^*\right\}  \in l^2(\mathbb Z^2 \times \mathbb N^*)\right\}.
\ene
The operator $\widehat{H_0}$ is selfadjoint because it is the  multiplication operator by the real eigenvalues $\lambda^{(0)}_{m}$ defined on its maximal domain.

\begin{prop}\label{prop3.2}
Let us define
\beq\label{3.11}
H_0= \mathbf U^\dagger \, \widehat{H_0}\, \mathbf U, \ \ \ D[H_0]:=\{ u \in \mathcal A :   \mathbf U u \in  D[\widehat{H_0}] \}.
\ene
Then, $H_0$ is selfadjoint. Its spectrum is pure point, and it consists of the eigenvalues $ \lambda^{(0)}_m, m \in \mathbb Z.$  Moreover, each eigenvalue $ \lambda^{(0)}_m, m \in \mathbb Z,$ has infinite multiplicity. Further, $h_0\subset H_0.$
\end{prop}

\begin{proof}  $H_0$ is unitarily equivalent to the selfadjoint operator $\widehat{H_0},$ and in consequence $H_0$ is selfadjoint. Let us prove that $h_0 \subset H_0.$ Suppose that $u \in D[h_0].$ Integrating by parts we obtain,
$$
 \left(h_0  u, u_{n,m,j} \right)_{\mathcal A}= \left(u, h_0\, u_{n,m,j} \right)_{\mathcal A}= \lambda^{(0)}_m\, \left( u, u_{n,m,j} \right)_{\mathcal A}, \qquad (n,m,j) \in \mathbb Z^2 \times \mathbb N^*.
$$
Hence,
\beq\label{3.12}
\mathbf U h_0 u:=\left\{ (h_0 u, u_{n,m,j})_{\mathcal A}, (n,m, j)\in \mathbb Z^2 \times \mathbb N^*\right \}=
\left\{ \lambda^{(0)}_m\, (u, u_{n,m,j})_{\mathcal A}, (n,m, j)\in \mathbb Z^2 \times \mathbb N^*
\right\} \in l^2\left(\mathbb Z^2 \times \mathbb N^*\right),
\ene
 where we used that $h_0 u \in \mathcal A,$ 
Hence,
$$
Uu \in D[\widehat{H_0}].
$$
Moreover,
$$
H_0 u =\mathbf  U^\dagger \widehat{H_0} \mathbf U u= \mathbf U^\dagger \left\{ \lambda^{(0)}_m\, (u, u_{n,m,j})_{\mathcal A}, (n,m, j)\in \mathbb Z^2 \times \mathbb N^*
\right\}= \mathbf U^\dagger\,\mathbf  U h_0 u= h_0 u.
$$
This completes the proof that $ h_0 \subset H_0.$ As $h_0 \subset H_0$  and  one has the completeness
of the eigenfunctions of $h_0$ by Proposition~\ref{prop3.1}, it follows that the spectrum of $H_0$  is pure point, it consists of the eigenvalues $ \lambda^{(0)}_m, m \in \mathbb Z,$  and  each eigenvalue $ \lambda^{(0)}_m, m \in \mathbb Z,$ has infinite multiplicity.
\end{proof}

We write the  \textcolor{black}{magnetized} Vlasov equation \eqref{3.1} as a Schr\"odinger equation with a selfadjoint Hamiltonian as follows,

$$
i \partial_t u = H_0 u.
$$
We call $H_0$ the  \textcolor{black}{magnetized}  Vlasov operator.

Actually, we can give more information on $h_0.$
\begin{prop}\label{prop3.3}
Let $h_0$ be the formal  \textcolor{black}{magnetized} Vlasov operator defined in \eqref{fv} and \eqref{dfv}, and let $H_0$ be the   \textcolor{black}{magnetized} Vlasov operator defined in \eqref{3.11}. We have that,
$$
h_0^\dagger = H_0,
$$
and, furthermore, $h_0$ is essentially selfadjoint, i.e., $H_0$ is the only selfadjoint extension of $h_0.$
\end{prop}

\begin{proof}   suppose that $ f \in D[h_0^\dagger].$ Then 
\beq\label{3.13}
(h_0u, f )_{\mathcal A} = (u,h_0^\dagger f)_{\mathcal A}.
\ene
Hence, by \eqref{3.12} and \eqref{3.13}
\beq\label{3.14}\begin{array}{l}
(h_0 u,f)_{\mathcal A}= (\mathbf U h_0 u, \mathbf Uf)_{l^2\left(\mathbb Z^2 \times \mathbb N^*\right)}=
\sum_{ (n,m, j) \in \mathbb Z^2 \times \mathbb N^*}\, \lambda^{(0)}_m  (u, u_{n,m,j})_{\mathcal A}\, \overline{ (f, u_{n,m,j})_{\mathcal A}}= \\\\
\sum_{ (n,m, j)\in \mathbb Z^2 \times \mathbb N^*}\,   (u, u_{n,m,j})_{\mathcal A}\, \overline{ (h_0^\dagger \, f, u_{n,m,j})_{\mathcal A}}.
\end{array}
\ene
Since \eqref{3.14} holds for all $u$ in the dense set $D[h_0]$ we obtain,
\beq\label{3.15}
\left\{  \lambda^{(0)}_m   (f, u_{n,m,j})_{\mathcal A}, (n,m, j) \in \mathbb Z^2 \times \mathbb N^*   \right\}= \left\{    ( h_0^\dagger f, u_{n,m,j})_{\mathcal A},  (n,m, j)\in \mathbb Z^2 \times \mathbb N^*   \right\}\in l^2( \mathbb Z^2 \times \mathbb N^*).
\ene
It follows that,
\beq\label{3.16}
\left\{  \lambda^{(0)}_m   (f, u_{n,m,j})_{\mathcal A}, (n,m, j)\in \mathbb Z^2 \times \mathbb N^*   \right\} \in  l^2( \mathbb Z^2 \times \mathbb N^*).
\ene
This implies that
$
f \in D[H_0]$ 
and that
$
h_0^\dagger f = H_0 f$. 
Then,  $h_0^\dagger \subset H_0.$ We prove in a similar way that if $ f \in D[H_0],$ then  $f \in D[h_0^\dagger]$ and that, $H_0 f= h_0^\dagger f.$ This implies that $ H_0 \subset h_0^\dagger.$ Hence the proof that $ h_0^\dagger= H_0$ is complete.
Finally let $A$ be a selfadjoint operator such that  $h_0 \subset A.$ Then,   $A^\dagger \subset h_0^\dagger= H_0.$ But as $A= A^\dagger,$ we obtain that $ A \subset H_0,$
and then, $H_0^\dagger \subset  A^\dagger,$ but as $A=A^\dagger, H_0= H_0^\dagger,$ we have $H_0 \subset A,$ and finally $A= H_0.$ This proves that $H_0$ is the only selfadjoint extension of $h_0.$
\end{proof}

 \section{\modifsF{The full \textcolor{black}{magnetized} Vlasov-Amp\`ere system with coupling}}
  \label{secvlasamp}
 \sss

 In this section we consider the full \textcolor{black}{magnetized} Vlasov-Amp\`ere system. We write the  system as a Schr\"odinger equation in the Hilbert space $\mathcal H$ as follows
  \beq\label{4.1}
i \partial_t \begin{pmatrix} u\\F\end{pmatrix}= \mathbf H \begin{pmatrix} u\\ F
\end{pmatrix},
 \ene
where  the  \textcolor{black}{magnetized}  Vlasov-Amp\`ere operator $\mathbf H$ is the following operator in $\mathcal H,$
\beq\label{4.2}
\mathbf H= 
\begin{bmatrix}H_0 &  - i v_1\, e^{\frac{-v^2}{4}} \\
 i I^\ast \int_{\mathbb R^2}\, v_1\, e^{\frac{-v^2}{4}}\,\cdot\, dv&0 \end{bmatrix}
 \qquad
 \left(\mbox{where we use the notation } e^{\frac{-v^2}{4}} =e^{- \frac{|v|^2}{4}} =e^{- \frac{v_1^2+v_2^2}{4}} \right).
  \ene
In a more detailed way, the right-hand side of \eqref{4.1} is defined as follows,
\beq\label{4.4}
 \mathbf H \begin{pmatrix} u\\ F
\end{pmatrix}:= \begin{pmatrix}H_0 u  - i v_1\, e^{\frac{-v^2}{4}} F \\
 i I^\ast \int_{\mathbb R^2}\, v_1\, e^{\frac{-v^2}{4}}\, u \,dv \end{pmatrix}.
\ene
We recall that $I^\ast$ gives zero when applied to constant functions in $L^2(0,2\pi).$
The domain of $\mathbf H$ is defined as follows,
\beq\label{4.5}
D[\mathbf H]:= D(H_0)\oplus L^2_0(0,2 \pi).
\ene
We write $\mathbf H$ in the following form,
\beq\label{4.6}
\mathbf H= \mathbf H_0+ \mathbf V,
\ene
where
\beq\label{4.7}
\mathbf H_0:= \begin{bmatrix} H_0& 0 \\ 0&0 \end{bmatrix},
\ene 
and
\beq\label{4.8}
\mathbf V:= \begin{bmatrix} 0 &  -i v_1\, e^{\frac{-v^2}{4}} \\
  i I^\ast \int_{\mathbb R^2}\, v_1\, e^{\frac{-v^2}{4}}\,\cdot\, dv&0 \end{bmatrix}.
  \ene
Clearly, $\mathbf H_0$ is selfadjoint with $D[\mathbf H_0]= D[\mathbf H].$ Moreover,  $\mathbf V,$ with $D[\mathbf V]= \mathcal H,$ is bounded  in $\mathcal H.$   Observe that the presence of $I^\ast$ in $\mathbf V$ assures us that $\mathbf V$ sends $\mathcal H$ in to $\mathcal H.$
Further,  it follows from a simple calculation that  $\mathbf V$ is  symmetric in $\mathbf H.$ Then, by the Kato-Rellich theorem, see  Theorem 4.3 in page 287 of \cite{ka}, the operator $\mathbf H$ is selfadjoint. We proceed to prove that $\mathbf H$ has pure point spectrum. Actually, we will explicitly compute the eigenvalues and a   basis of eigenfunctions. We do that in several steps.

\modifsF{
\begin{remark}{\rm
The Gauss law in strong sense for a function 
 $
 \begin{pmatrix} u(x,v)\\ F(x)\end{pmatrix}\in \mathcal H$
 reads,
\beq\label{4.147}
\int_{\mathbb R^2}\, u(x,v)\, e^{\frac{-v^2}{4}}\, dv+F^\prime(x)=0.
\ene
Later, in Remark ~\ref{remker}, we write the Gauss law in weak sense, and we show that it can, equivalently, be expressed as a orthogonality relation with a subset of the eigenfunctions in the kernel of the \textcolor{black}{magnetized} Vlasov-Amp\`ere operator $\mathbf H.$
}
\end{remark}
}

\subsection{The kernel of $\mathbf H$}\label{kernel}
In this subsection we compute a basis for the kernel of the  \textcolor{black}{magnetized} Vlasov-Amp\`ere operator $\mathbf H.$ We have to solve the equation
\beq\label{4.9}
\mathbf H \begin{pmatrix} u\\ F\end{pmatrix}=0.
\ene
Inserting  \eqref{4.4} in \eqref{4.9} we obtain,

\beq\label{4.10}
\left\{\begin{aligned} 
&  i \left(-v_1 \partial_x+\omega_{\text{\rm c}}(v_2\,\partial_{v_1} - v_1\, \partial_{v_2})\right) u - i v_1\, e^{\frac{-v^2}{4}} F =0, \\
 &i I^\ast \int_{\mathbb R^2}\, v_1\, e^{\frac{-v^2}{4}}\, u \,dv =0.
 \end{aligned} 
 \right.
 \ene
 Denote,
 \beq\label{4.11}
 \psi(x):= \int_x^{2\pi}\, F(y)\, dy-\frac{1}{2\pi}\, \int_0^{2\pi} y F(y) dy.
 \ene
 Then,  as $ F\in L^2_0(0,2\pi),$ we have that  $\psi \in H^{(1,0)}(0,2\pi).$  Further,
 \beq\label{4.12}
 F(x)=-  \psi^\prime(x).
 \ene
 Let us designate $
 \modifsF{\gamma}(x,v):=  u(x,v)- e^{\frac{-v^2}{4}}\, \psi(x)$. 
 Hence, the first equation in \eqref{4.10} is equivalent to the following equation
 \beq\label{4.13}
 H_0\, \modifsF{\gamma}=0.
 \ene
Then, the general solution to the first equation in \eqref{4.10} can be written as
\beq\label{4.14}
u(x,v)=  e^{\frac{-v^2}{4}}\, \psi(x) + \modifsF{\gamma}(x,v),
\ene
with $F= - \psi^\prime,$ where $\psi \in H^{(1,0)}(0,2\pi),$  and  $\modifsF{\gamma}$ solves \eqref{4.13}. Furthermore,  by \eqref{4.14}  the second equation is \eqref{4.10} is equivalent to,
\beq\label{4.15}
I^\ast \int_{\mathbb R^2}\, v_1\, e^{\frac{-v^2}{4}}\, \modifsF{\gamma}\,dv =0.
\ene
Then, we have proven that the general solution to \eqref{4.10} can be written as,
\beq\label{4.16}
\begin{pmatrix} u\\ F \end{pmatrix}= \begin{pmatrix}  e^{\frac{-v^2}{4}}\, \psi(x) + \modifsF{\gamma}(x,v) \\ - \psi^\prime(x)\end{pmatrix},
\ene 
where $\psi \in H^{(1,0)}(0,2\pi), F=- \psi^\prime, $ and $\modifsF{\gamma}$ solves \eqref{4.13}. By Proposition ~\ref{prop3.1} the general solution  can be written as
\beq\label{4.17}
\modifsF{\gamma}= \sum_{(n,j)\in \mathbb Z \times \mathbb N^*}\, (\modifsF{\gamma}, u_{n,0,j})_{\mathcal A}\, u_{n,0,j}.
\ene
Using \eqref{3.4} we prove by explicit calculation that 
$u_{n,0,j}$, $n\in \mathbb Z$ and $j\in \mathbb N^*$, satisfies \eqref{4.15}.
So the general solution  \eqref{4.17} 
satisfies  \eqref{4.13} and  \eqref{4.15}.

In the following lemma we construct a basis of $\text{\rm Ker}[\mathbf H],$ using the results above.

\begin{lemma} \label{lemm4.0}  Let $\mathbf H$ be the \textcolor{black}{magnetized} Vlasov-Amp\`ere operator defined in \eqref{4.4} and, \eqref{4.5}.
Let $u_{n,0,j}$ be the eigenfunctions defined in \eqref{3.4}. Then,  the following set of eigenfunctions of $\mathbf H$ with eigenvalue zero,
\beq\label{4.18}
\left\{\mathbf V^{(0)}_{n}:= \frac{1}{\sqrt{2\pi+n^2}} \, \frac{e^{inx}}{\sqrt{2\pi}}    \begin{pmatrix}   e^{\frac{-v^2}{4}} \\ - i\, n \end{pmatrix}, n \in \mathbb Z^\ast\right\} \ds \cup 
\left\{\mathbf M^{(0)}_{n,j}:= \begin{pmatrix} u_{n,0,j}\\0 \end{pmatrix}, {(n,j) \in \mathbb Z \times \mathbb N^*}\right\},
\ene
is linearly independent and it is a basis of $\text{\rm Ker}[\mathbf H].$ 
\end{lemma}

\begin{proof}
Let us first prove the linear independence of the sets of functions \eqref{4.18}. We have to prove that if a linear combination of the  eigenfunctions \eqref{4.18} is equal to zero then, each of the coefficients in the linear combination is equal to zero. For this purpose we write the general linear combination of the eigenfunctions in \eqref{4.18}  with a convenient notation.   Let $\mathbb M_1$ be any  finite subset of $\mathbb Z^\ast$ and let $\mathbb M_2$ be any finite subset of $\mathbb Z \times \mathbb N^*$. Then, the  general linear combination of the eigenfunctions in \eqref{4.18} can be written as follows,
$$
\sum_{n \in \mathbb M_1} \alpha_n\,  \frac{1}{\sqrt{2\pi+n^2}} \, \frac{e^{i n x}}{\sqrt{2\pi}}    \begin{pmatrix}   e^{\frac{-v^2}{4}} \\ - i\, n \end{pmatrix} + \sum_{ (l,p)\in \mathbb M_2} \,\beta_{(l , p) } \begin{pmatrix} u_{l,0,p}\\ 0\end{pmatrix},
$$
for some complex numbers  $\alpha_n, n \in \mathbb M_1,$ 
and $\beta_{(l,p)}, (l,p) \in \mathbb M_2.$   
Suppose that,  
$$
\sum_{n \in \mathbb M_1} \alpha_n\,  \frac{1}{\sqrt{2\pi+n^2}} \, \frac{e^{i n x}}{\sqrt{2\pi}}    \begin{pmatrix}   e^{\frac{-v^2}{4}} \\ - i\, n \end{pmatrix} + \sum_{ (l,p)\in \mathbb M_2} \,\beta_{(l , p) } \begin{pmatrix} u_{l,0,p}\\ 0\end{pmatrix}=0.
$$
Since the second component of the functions in the second sum is zero, we have
$
\sum_{n\in \mathbb M_1}\alpha_{n}\,  \frac{1}{\sqrt{2\pi+n^2}} \, \frac{e^{i nx}}{\sqrt{2\pi}} \, n  =0$. 
Further, as the $\ds \frac{e^{inx}}{\sqrt{2\pi}}, n \in \mathbb M_1$ are orthogonal to each other, we have that, $\alpha_{n}=0, n \in \mathbb M_1.$ Furthermore, as the $\alpha_n, n \in \mathbb M_1$ are equal to zero, we obtain
$
 \sum_{(l.p)\in \mathbb M_2}\, \beta_{(l,p)}  u_{l,0, p}=0$. 
 Moreover,  since the $u_{l,0,p},  (l,p)\in \mathbb M_2$ are an orthonormal set, $\beta_{(l,p)}=0, (l,p)\in \mathbb M_2.$
This proves the linear independence of the set \eqref{4.18}.  Moreover, by \eqref{4.16}  with $\psi(x)= \frac{e^{inx}}{\sqrt{2\pi}}, n \in \mathbb Z^\ast,$ and $ f=0,$ each of the functions
$$
 \frac{1}{\sqrt{2\pi+n^2}} \, \frac{e^{inx}}{\sqrt{2\pi}}    \begin{pmatrix}   e^{\frac{-v^2}{4}} \\ - i\, n \end{pmatrix} \qquad n \in \mathbb Z^\ast, 
$$
is an eigenvector of $\mathbf H$ with eigenvalue zero. Similarly, by \eqref{4.16}  with  $\psi(x)=0,$  and $f= u_{n,0,j},$ one has that each of the functions,
$$
\begin{pmatrix} u_{n,0,j}\\0 \end{pmatrix}, \qquad (n,j) \in \mathbb Z \times \mathbb N^*,
$$
is an eigenfunctions of $\mathbf H$ with eigenvalue zero. By the  Fourier transform, the set of functions,
$
\frac{e^{inx}}{ \sqrt{2 \pi}}$, $ n \in \mathbb Z$, 
is a complete orthonormal set in $L^2(0,2\pi).$  Then, in particular,  any $\psi \in H^{(1,0)}(0,2 \pi),$ can be represented as follows,
\beq\label{4.20}
\psi(x)= \sum_{n \in \mathbb Z^\ast}\, \left(\psi, \frac{e^{inx}}{ \sqrt{2 \pi}}\right)_{L^2(0,2\pi)}\, \frac{e^{inx}}{ \sqrt{2 \pi}},
\ene
where the series converges in the norm of $L^2(0,2\pi).$ Note that there is no term with $n=0$ because the mean value of $\psi$ is zero.
Then, by \eqref{4.20},
\beq\label{4.20.1}
 \begin{pmatrix}  e^{\frac{-v^2}{4}}\, \psi(x)  \\
  - \psi^\prime(x)\end{pmatrix}= \sum_{n \in \mathbb N^*}\,   \left(\psi, \frac{e^{inx}}{ \sqrt{2 \pi}}\right)_{L^2(0,2\pi)}\, \frac{e^{inx}}{\sqrt{2\pi}}    \begin{pmatrix}   e^{\frac{-v^2}{4}} \\ - i\, n \end{pmatrix}. 
\ene 
Finally, it follows from  \eqref{4.16}, \eqref{4.17} and \eqref{4.20.1} that the set \eqref{4.18} is a basis of the kernel of $\mathbf H.$ 
\end{proof}

\subsection{The eigenvalues of $\mathbf H$ different from zero and their eigenfunctions}
In this subsection we compute the non-zero eigenvalues of $\mathbf H$     and we give explicit formulae for the eigenfunctions that correspond to each eigenvalue.
By \eqref{4.4} we have to solve the system of equations
\beq\label{4.21}
\left\{
\begin{aligned} 
 &H_0 u  - i v_1\, e^{\frac{-v^2}{4}} F= \lambda u, \\
&i I^\ast \int_{\mathbb R^2}\, v_1\, e^{\frac{-v^2}{4}}\, u \,dv  = \lambda F,
\end{aligned} \right.
\ene
with $ \lambda \in \mathbb R \setminus\{0\},$ and
$
\begin{pmatrix}  u\\ F \end{pmatrix} \in  D[\mathbf H]$. 
 We first consider the case where the electric field, $F$, is zero, and then, when it is different from zero.
\subsubsection{The case with zero electric field}
We have to compute solutions to \eqref{4.21} of the form,
\beq\label{4.22}
\begin{pmatrix}  u\\ 0 \end{pmatrix} \in  D[\mathbf H],
\ene
with $ u \in D[H_0].$  Introducing \eqref{4.22} into the system \eqref{4.21} we obtain,

\beq\label{4.22.a}
\left\{
\begin{aligned} 
& H_0 u = \lambda u, \\
&i I^\ast \int_{\mathbb R^2}\, v_1\, e^{\frac{-v^2}{4}}\, u \,dv  = 0.
\end{aligned} \right. 
\ene
We seek for eigenfunctions of the form,
\beq\label{4.23}
u(x,v):=  \frac{1}{\sqrt{2\pi}} \, e^{in(x-\frac{v_2}{\omega_{\text{\rm c}}})}\,   \frac{1}{\sqrt{2\pi}}\, e^{im\varphi}\, \tau(r), \qquad (n,m) \in \mathbb Z^2,
\ene
where \modifsF{$(r,\varphi)$ are the polar coordinates of $v\in \R^2$, and} the function $\tau$  will be specified later. We first consider the case when $n=0.$ In this case the second equation in \eqref{4.22.a} is  satisfied because the operator $I^\ast$  gives zero when applied to functions that are independent of $x.$ Hence, we are left with the first equation only,  that   is the problem that we solved in Section \ref{noef}. Then, as we seek non zero eigenvalues we have to have $m \neq 0$ in \eqref{4.23}. Using the results of Section ~\ref{noef} we obtain the following lemma.


\begin{lemma}\label{lemm4.1}
 Let $\mathbf H$ be the  \textcolor{black}{magnetized} Vlasov-Amp\`ere operator defined in \eqref{4.4} and, \eqref{4.5}. Let $ \{ \tau_j\}_{j=1}^\infty$ be an orthonormal basis of $L^2(\mathbb R^+, r dr).$  Let $ \varphi \in [0,2\pi), r >0,$ be polar coordinates in
 $\mathbb R^2, v_1= r \cos\varphi, v_2= r \sin\varphi.$ For $(m,j)\in \mathbb Z^\ast \times \mathbb N^*$ let $u_{0,m,j}$ be the eigenfunction defined in \eqref{3.4}. 
Then, the set 
\beq\label{4.26}
\mathbf V_{m,j}:= \left\{\begin{pmatrix} u_{0,m,j}\\ 0\end{pmatrix}, (m,j)\in \mathbb Z^\ast \times \mathbb N^*\right\},
 \ene 
 is an orthonormal set in $\mathcal H.$ Furthermore, each function on this set 
 is an eigenvector of $\mathbf H$ corresponding  the eigenvalue  $\lambda^{(0)}_{ m}= m\, \omega_{\text{\rm c  }}\neq 0,$
\beq\label{4.27}
\mathbf H \mathbf V_{m,j} = \,\lambda^{(0)}_m\,  \mathbf V_{m,j}, \qquad (m,j)\in \mathbb Z^\ast \times \mathbb N^*.
\ene
Moreover, each eigenvalue $ \lambda^{(0)}_m$ has infinite multiplicity.
\end{lemma}

\begin{proof}
The lemma follows from Proposition ~\ref{prop3.1} and since the second equation in \eqref{4.22.a}  is always satisfied for functions that are independent  of $x.$
\end{proof}

Let us now study the second case, namely  $n \neq 0.$ We have to consider the second equation in the system \eqref{4.22.a}. We first prepare some results.
 For $ m \in \mathbb Z$ let $J_m(z), z \in \mathbb C,$ be the Bessel function. We have that 
 
\beq\label{4.28}
J_m(-z)=(-1)^m\, J_m(z), \qquad  J_{-m}(-z)= J_m(z).
\ene
For the first equation see formula 10.4.1 in page 222 of \cite{nist} and for the second see formula 9.1.5 in page 358 of \cite{as}.
The Jacobi-Anger formula, given in equation 10.12.1, page 226 of \cite{nist}, yields,
\beq\label{4.29}
e^{iz \sin\varphi} = \sum_{m \in \mathbb Z} e^{im\varphi}\, J_m(z).
\ene
The Parseval identity for the Fourier series applied to  \eqref{4.29} gives,
\beq\label{4.46}
\sum_{m \in \mathbb Z}\, J_m(z)^2=1, \qquad z \in \mathbb R.
\ene

 Differentiating  the Jacobi-Anger formula with respect to $\varphi$ we obtain,
 \beq\label{4.30}
 z \cos\varphi\, e^{i z \sin\varphi}=  \sum_{m \in \mathbb Z}\, m\, e^{im\varphi}\, J_m(z).
\ene
Taking in \eqref{4.30} $ z = - nr/ \omega_{\text{\rm c}}, $ with $ n \neq 0,$   recalling that $v_1= r \cos\varphi, v_2= r \sin\varphi,$ and using the first  equation in \eqref{4.28} we get,
\beq\label{4.31}
v_1\, e^{-i \frac{n v_2}{\omega_{\text{\rm c}}}}= - \frac{\omega_{\text{\rm c}}}{n}\,   \sum_{m \in \mathbb Z}\, m\, e^{im\varphi}\, (-1)^m\, J_m\left(\frac{nr}{\omega_{\text{\rm c}}}\right), \qquad n \neq 0.
\ene
From \eqref{4.31} we obtain, 
\beq\label{4.31.b}
\int_0^{2\pi}\,  v_1\, e^{     -i\frac{n v_2}{\omega_{\text{\rm c}}} }\,e^{im \varphi}\,  d\varphi=
 2\pi \frac{m \omega_{\text{\rm c}}}{n}\, (-1)^m\, J_{-m}\left(\frac{nr}{\omega_{\text{\rm c}}}\right)=
2\pi \frac{m \omega_{\text{\rm c}}}{n}\, J_{m}\left(\frac{nr}{\omega_{\text{\rm c}}}\right),\qquad n \neq 0,
\ene
where in the last equality we used both equations in \eqref{4.28}. Using \eqref{4.31} and taking $n ,m \neq 0$ we prove that the second equation in \eqref{4.22.a} \modifsF{with $u$ given by \eqref{4.23}} is equivalent to,
\beq\label{4.32}
\int_0^\infty\, e^{-\frac{r^2}{4}}\, J_{m}\left(\frac{nr}{\omega_{\text{\rm c}}}\right) \, \tau(r)\, r\, dr=0.
\ene
Taking $m=0$ is possible, but it will be discarded below in Lemma \ref{lemm4.2}.
Let us denote by $V_{n,m}$ the orthogonal complement in $L^2(\mathbb R^+, r dr)$ to the function, $ e^{-\frac{r^2}{4}}\, J_{m}\left(\frac{nr}{\omega_{\text{\rm c}}}\right),$ that is to say,
\beq\label{4.33}
V_{n,m}:=\left\{  f \in  L^2(\mathbb R^+, r dr) : \left(f,   e^{-\frac{r^2}{4}}\, J_{m}\left(\frac{nr}{\omega_{\text{\rm c}}}\right) \right)_{ L^2(\mathbb R^+, r dr) }=0  \right\}, n,m  \in \mathbb Z^\ast.
\ene
 Note  that $V_{n,m}$ is an infinite dimensional subspace of $ L^2(\mathbb R^+, r dr)$ of codimension equal to one. 
We prove the following lemma using the results above.

\begin{lemma}\label{lemm4.2}
 Let $\mathbf H$ be the  \textcolor{black}{magnetized} Vlasov-Amp\`ere operator defined in \eqref{4.4} and \eqref{4.5}.
Let $\tau_{n,m,j}, n,m \in \mathbb Z^\ast, j \in \mathbb N^*$ be an orthonormal basis in $V_{n,m}$ and define,
\beq\label{4.34}
f_{n,m,j}:= \frac{1}{\sqrt{2\pi}} \, e^{in(x-\frac{v_2}{\omega_{\text{\rm c}}})}\,   \frac{1}{\sqrt{2\pi}}\, e^{im\varphi}\, \tau_{n,m,j}(r),  n,m \in \mathbb Z^\ast, j \in \mathbb N^*.
\ene
Then, the set 
\beq\label{4.35}
\left\{\mathbf W_{n,m,j} := \begin{pmatrix} f_{n,m,j}\\ 0\end{pmatrix},   n,m \in \mathbb Z^\ast, j \in \mathbb N^*.\right\}
 \ene 
 is an orthonormal set in $\mathcal H.$ Furthermore, each function on this set 
 is an eigenvector of $\mathbf H$ corresponding  the eigenvalue  $\lambda_{m}^{(0)}= m\, \omega_{\text{\rm c  }}\neq 0,$
\beq\label{4.36}
\mathbf H \mathbf W_{n,m,j} = \,\lambda_m^{(0)}\, \mathbf W_{n,m,j}\, \qquad   n,m \in \mathbb Z^\ast, j \in \mathbb N^*.
\ene
\end{lemma}
Moreover, each eigenvalue $\lambda^{(0)}_m$ has infinite multiplicity. 

\begin{proof} The lemma follows from \eqref{4.22.a}, \eqref{4.32}, \eqref{4.33} and \eqref{4.34}. Note that the case $m=0$ does not appear because we are looking for eigenfunctions with eigenvalue different from zero. Furthermore, the eigenvalues $\lambda^{(0)}_m$ have infinite multiplicity because all the eigenfunctions  $\mathbf W_{n,m,j} $ with a fixed $m$ and all $n \in \mathbb Z^\ast, j \in \mathbb N^*,$ are orthogonal eigenfunctions for the eigenvalue, $\lambda^{(0)}_m$. 
 
\end{proof}

\subsubsection{The case with electric field different from zero}
From the physical point of view this is the most interesting situation, since it describes the interaction of the electrons with the electric field. Moreover, it is the most involved technically. We look for eigenfunctions of the form,
\beq\label{4.37}
\frac{1}{\sqrt{2\pi}}\,\begin{pmatrix}  e^{in(x-\frac{v_2}{\omega_{\text{\rm c}}}})\, \tau(v)\\    e^{inx} G\end{pmatrix},
\ene
\modifsF{where $G$ is a constant.}
Since we wish that the electric field is nonzero we must have $G\neq 0.$ Hence, to fulfill that 
$
\int_0^{2\pi}\, F(x)\, dx=0$, 
we must have $ n \neq 0.$
The eigenvalue system \eqref{4.21} recasts as,
\beq\label{4.38}
\left\{\begin{aligned} 
&(-i  \omega_{\text{\rm c}} \partial_\varphi- \lambda) \tau= i G\, v_1\, e^{\frac{-v^2}{4}}\,  e^{in\frac{v_2}{\omega_{\text{\rm c}}}}, \\
&\lambda G= i \int_{\mathbb R^2}\, v_1\,  e^{\frac{-v^2}{4}}\, e^{-in\frac{v_2}{\omega_{\text{\rm c}}}}\,\tau(v)\, dv.
\end{aligned} \right.
\ene
Changing  $n$ into $-n$ in \eqref{4.31} and using the first equation in \eqref{4.28} we obtain,
\beq\label{4.39}
v_1\, e^{i \frac{n v_2}{\omega_{\text{\rm c}}}}=  \frac{\omega_{\text{\rm c}}}{n}\,   \sum_{m \in \mathbb Z}\, m\, e^{im\varphi}\, J_m\left(\frac{nr}{\omega_{\text{\rm c}}}\right), \qquad n \neq 0.
\ene
Plugging \eqref{4.39} into  the first equation in the system  \eqref{4.38} we get,
\beq\label{4.40}
 (-i  \omega_{\text{\rm c}} \partial_\varphi- \lambda) \tau(r,\varphi)= i G\,e^{\frac{-r^2}{4}}\, \frac{\omega_{\text{\rm c}}}{n}\, \sum_{m \in \mathbb Z^\ast}\, m\, e^{im\varphi}\, J_m\left(\frac{nr}{\omega_{\text{\rm c}}}\right), \qquad n \neq 0.
\ene
A solution to \eqref{4.40} is given by
\beq\label{4.41}
\tau(r,\varphi)=  i G\,e^{\frac{-r^2}{4}}\,  \frac{1}{n}\, \sum_{m \in \mathbb Z^\ast}\, \frac{m\,\omega_{\text{\rm c}}}{m \omega_{\text{\rm c}}-\lambda}\, e^{im\varphi}\, J_m\left(\frac{nr}{\omega_{\text{\rm c}}}\right), \qquad n \neq 0,
\ene
for $\lambda \neq m \omega_{\text{\rm c}}, m \in \mathbb Z^\ast.$ Introducing \eqref{4.41} into the second equation in the system \eqref{4.38}, and simplifying by $ G\neq 0$ we get,
\beq\label{4.42}
\lambda =- \frac{1}{n}\, \sum_{m \in \mathbb Z^\ast}\, \frac{m\,\omega_{\text{\rm c}}}{m \omega_{\text{\rm c}}-\lambda}\, \int_{\mathbb R^2}\,e^{\frac{-r^2}{2}}\, e^{im\varphi}\,  e^{-in\frac{v_2}{\omega_{\text{\rm c}}}}\, J_m\left(\frac{nr}{\omega_{\text{\rm c}}}\right)\, v_1\,dv, \qquad n \neq 0, \quad \lambda \neq m \omega_{\text{\rm c}}, \quad m \in \mathbb Z^\ast.
\ene  
Plugging \eqref{4.31.b} into \eqref{4.42} and using that $ dv = r dr d\varphi,$ we obtain,
\modifsF{
\beq\label{4.43}
\lambda =- \frac{2 \pi}{n^2}\, \sum_{m \in \mathbb Z^\ast}\, \frac{m^2\,\omega^2_{\text{\rm c}}}{m \omega_{\text{\rm c}}-\lambda}\, a_{n,m}, \qquad n \neq 0, \quad \lambda \neq m \omega_{\text{\rm c}},\quad m \in \mathbb Z^\ast.
\ene  
where we denote
\beq\label{4.48}
a_{n,m}:= \int_0^\infty\, \,e^{\frac{-r^2}{2}}\,  J_m\left(\frac{nr}{\omega_{\text{\rm c}}}\right)^2\, r dr >0, \quad m\in \mathbb Z.
\ene 
Equation \eqref{4.43} is a secular equation that we will study to determine the possible values of $\lambda.$ Remark that \eqref{4.43} coincides with the secular equation obtained by \cite{bedro} and \cite{bernstein}. First we write it in  a more convenient form. 
Note that thanks to the two equations in \eqref{4.28} we have  $J_{-m}(z)= (-1)^m J_m(z)$ and then $a_{n,-m}=a_{n,m}$. Using also 
$ \frac{m\,\omega_{\text{\rm c}}}{m \omega_{\text{\rm c}}-\lambda}= 1+  \frac{\lambda}{m \omega_{\text{\rm c}}-\lambda}$, this allow to obtain that 
\begin{equation} \label{propserie}
 \sum_{m \in \mathbb Z^\ast}\, \frac{m^2\,\omega^2_{\text{\rm c}}}{m \omega_{\text{\rm c}}-\lambda}\, a_{n,m} 
=\sum_{m \in \mathbb Z^\ast}\, \left(m  \omega_{\text{\rm c}} +  \frac{m\,\omega_{\text{\rm c}}\lambda}{m \omega_{\text{\rm c}}-\lambda}\right)\, a_{n,m} 
= \lambda\sum_{m \in \mathbb Z^\ast}\, \frac{m\,\omega_{\text{\rm c}}}{m \omega_{\text{\rm c}}-\lambda}\, a_{n,m}.
\end{equation}
Simplifying by $ \lambda \neq 0$ and using \eqref{propserie} we write \eqref{4.43} as
\beq\label{4.45}
1 =- \frac{2 \pi}{n^2}\, \sum_{m \in \mathbb Z^\ast}\,  \frac{m\,\omega_{\text{\rm c}}}{m \omega_{\text{\rm c}}-\lambda}\,  a_{n,m} , \qquad n \neq 0, \quad \lambda \neq m \omega_{\text{\rm c}},\quad  m \in \mathbb Z^\ast.
\ene
By \eqref{4.46} we have $\displaystyle \sum_{m\in \mathbb Z^\ast} a_{n,m} <+\infty$ and thus the series in \eqref{4.45} is absolutely convergent.
Secondly we proceed to write \eqref{4.45} in another form that we find  convenient. Using again $a_{n,-m}=a_{n,m}$ we have
\beq\label{propseries2}
 \sum_{m \in \mathbb Z^\ast}\,  \frac{m\,\omega_{\text{\rm c}}}{m \omega_{\text{\rm c}}-\lambda}\,  a_{n,m}
=2 \sum_{m=1}^\infty\,  \frac{m^2\,\omega^2_{\text{\rm c}}}{m^2 \omega^2_{\text{\rm c}}-\lambda}\,  a_{n,m}.
\ene
 Let us denote
\beq\label{4.47}
g(\lambda):= 4 \pi\, \sum_{m=1}^\infty\,  \frac{m^2\,\omega^2_{\text{\rm c}}}{m^2 \omega^2_{\text{\rm c}}-\lambda^2} a_{n,m}, \, \qquad \lambda \neq m \omega_{\text{\rm c}},\qquad  m \in \mathbb Z^\ast.
\ene
Then using \eqref{propseries2},  \eqref{4.45} is equivalent to
\beq\label{4.49}
g(\lambda) = -n^2, n \in \mathbb Z^\ast, \qquad \lambda \neq m \omega_{\text{\rm c}},\qquad  m \in \mathbb Z^\ast.
\ene
}
Since the function $g$ is even  it is enough to study it for $\lambda \geq 0.$ It has simple poles as $\lambda= m \omega_{\text{\rm c}}, m \in \mathbb N^*.$ It is well defined for $ \lambda \in  \cup_{m=0}^\infty I_m,$ where,
\beq\label{4.50}
I_0:= [0,  \omega_{\text{\rm c}}), \qquad  I_m:= ( m  \omega_{\text{\rm c}}, (m+1)  \omega_{\text{\rm c}}), \qquad m \in \mathbb N^*.
\ene
\begin{lemma}\label{lemm4.3}
The function $g$ is positive in $I_0.$ For $ m \geq 1,$ $g$ is monotone increasing in  the interval $I_m$ and the following limits hold,
\beq\label{4.51}
\lim{\lambda \to  (m  \omega_{\text{\rm c}})^- }= + \infty, \qquad    \lim{\lambda \to \ (m \omega_{\text{\rm c}})^+ }= -\infty.
\ene
\end{lemma}
\begin{proof} The fact that $g$ is positive in $I_0$ follows from the definition of $g$ in \eqref{4.47}.  Furthermore, since $a_{n,m} >0, m \geq 1,$ and the functions
$
\lambda \mapsto  \frac{m^2\,\omega^2_{\text{\rm c}}}{m^2 \omega^2_{\text{\rm c}}-\lambda^2}
$  
are monotone  increasing away from the poles, we have that $g$ is increasing in $I_m, m \geq 1,$ and that  the limits in \eqref{4.51} hold.
\end{proof}

In the following lemma we obtain the solutions to \eqref{4.49}

\begin{lemma}\label{lemm4.4} 
For $ n \in \mathbb Z^\ast,$ the equation \eqref{4.49} has a countable number of real simple roots, $\lambda_{n,m}$ in $ (m \omega_{\text{\rm c}}, ( m+1) \omega_{\text{\rm c}}), m \geq 1.$ By parity $\lambda_{n,m} := - \lambda_{n,-m}, \, m \leq -1$ is also a root. There is no root in $(-\omega_{\text{\rm c}}, \omega_{\text{\rm c}}).$ Furthermore, $\lambda_{n_1,m_1} =\lambda_{n_2,m_2}$ if and only if $n_1=n_2,$  and $m_1=m_2,$
\end{lemma}
\begin{proof}The first two  items follow from Lemma~\ref{lemm4.3} and the parity of $g.$ The third  point is true because $g$ is positive in  $(-\omega_{\text{\rm c}}, \omega_{\text{\rm c}}).$ 
Finally, if  $\lambda_{n_1,m_1} =\lambda_{n_2,m_2},$ we have, $m_1=m_2,$ because $\lambda_{n_1,m_1} \in (m_1 \omega_{\text{\rm c}}, ( m_1+1) \omega_{\text{\rm c}})  $ and  $\lambda_{n_2,m_2} \in  (m_2 \omega_{\text{\rm c}}, ( m_2+1) \omega_{\text{\rm c}}).$ Furthermore, if  $n_1 \neq n_2,$ then, $ \lambda_{n_1,m} \neq \lambda_{n_2,m},$ because, otherwise,  $ -n_1^2 =g( \lambda_{n_1,m})= g( \lambda_{n_2,m})= -n_2^2,$ and this is impossible.
\end{proof}

Using \eqref{4.37} and \eqref{4.41} we define,
\beq\label{4.83}
\mathbf Y_{n,m}:=  \frac{1}{\sqrt{2\pi}} \, e^{inx}
\begin{pmatrix}   e^{-in \frac{v_2}{\omega_{\text{\rm c}}}}\, \eta_{n,m}(v)\\  -n \,  i\end{pmatrix}, \qquad n, m \in \mathbb Z^\ast,
\ene
where
\beq\label{4.84}
\eta_{n,m}(v):=  \,e^{\frac{-r^2}{4}}\,   \sum_{q \in \mathbb Z^\ast}\, \frac{q\,\omega_{\text{\rm c}}}{q \omega_{\text{\rm c}}-\lambda_{n,m}}\, e^{iq\varphi}\, J_q\left(\frac{nr}{\omega_{\text{\rm c}}}\right), \qquad n, m \in \mathbb Z^\ast.
\ene
For $m \in \mathbb Z^\ast,  \lambda_{n,m}$ is the root given in Lemma ~\ref{lemm4.4}.
Note that we have simplified the factor $\frac{i}{n}$ in \eqref{4.41} and we have taken $G=1.$ Remark  that, formally, $\mathbf Y_{n,m}$ is an eigenfunction
of $\mathbf H,$
\beq\label{4.84.b}
\mathbf H \mathbf Y_{n,m}=  \lambda_{n,m}\, \mathbf Y_{n,m}.
\ene
However, we have to verify that $ \mathbf Y_{n,m} \in \mathcal H$. 
\modifsF{We have
$$\|Y_{n,m} \|_{\mathcal H}= \sqrt{  \left\|\eta_{n,m}\right\|^2_{L^2(\mathbb R^2)}+ n^2},$$
and
\beq\label{4.85}
\left\|\eta_{n,m}(v)\right\|^2_{L^2(\mathbb R^2)}= 2\pi \sum_{q \in\mathbb Z^\ast} \left( \frac{q \omega_{\text{\rm c}}}{q \omega_{\text{\rm c}}-  \lambda_{n,m}}  \right)^2\, a_{n,q},
\ene
where we used the first equation in \eqref{4.28}.
We now prove that $\left\|\eta_{n,m}(v)\right\|^2_{L^2(\mathbb R^2)}<+\infty$ and exhibit an asymptotic  expansion of this quantity which will be used later.
}
\begin{lemma}\label{lemm4.8}
We have,
\beq \label{4.86}
  \left\|\eta_{n,m}(v)\right\|^2_{L^2(\mathbb R^2)}= \frac{n^4}{2\pi}\, \frac{1}{a_{n,m}}\, \left( 1+ O\left( \frac{1}{m^2}  \right)  \right)+ O\left( \frac{1}{m^2}\right) , \qquad  m \to \pm \infty.
\ene
\end{lemma}

\begin{proof} Recall that  for $m \leq -1, \lambda_{n,m}= - \lambda_{n,-m}.$ Then, 
\beq\label{4.87}
\left\|\eta_{n,m}(v)\right\|^2_{L^2(\mathbb R^2)}= \left\|\eta_{n,-m}(v)\right\|^2_{L^2(\mathbb R^2)}, \qquad m \leq -1.
\ene
Hence, it is enough to consider the case $ m \geq 1.$ We decompose the sum in \eqref{4.85} as follows,
\beq\label{4.88}
\left\|\eta_{n,m}(v)\right\|^2_{L^2(\mathbb R^2)}:= \sum_{j=1}^4 h^{(j)}( \lambda_{n,m}),
\ene
where,
\beq\label{4.89}
 h^{(1)}(\lambda_{n,m}):= 2\pi \sum_{q \leq -1} \left( \frac{q \omega_{\text{\rm c}}}{q \omega_{\text{\rm c}}-  \lambda_{n,m}}\right)^2\, a_{n,q},
\ene
\beq\label{4.90}
 h^{(2)}(\lambda_{n,m}):= 2\pi \sum_{ 1 \leq q \leq  m-1} \left( \frac{q \omega_{\text{\rm c}}}{q \omega_{\text{\rm c}}-  \lambda_{n,m}}  \right)^2\, a_{n,q},
\ene
\beq\label{4.91}
 h^{(3)}(\lambda_{n,m}):= 2\pi  \left( \frac{m\, \omega_{\text{\rm c}}}{m \omega_{\text{\rm c}}-  \lambda_{n,m}}  \right)^2\, a_{n,m},
\ene
and
\beq\label{4.92}
 h^{(4)}(\lambda_{n,m}):= 2\pi \sum_{ m+1 \leq q } \left( \frac{q \omega_{\text{\rm c}}}{q \omega_{\text{\rm c}}-  \lambda_{n,m}}  \right)^2\, a_{n,q}.
\ene
Since
$
 \left( \frac{q \omega_{\text{\rm c}}}{q \omega_{\text{\rm c}}-  \lambda_{n,m}}  \right)^2\leq   \frac{q^2}{(m+1)^2}$,   $ q \leq -1$,  
we have,
\beq\label{4.93}
\left|  h^{(1)}( \lambda_{n,m}) \right| \leq  2\pi \sum_{q \leq -1} \, \frac{q^2}{(m+1)^2} \, a_{n,-q} \leq C \frac{1}{(m+1)^2},
\ene
where in the last inequality we used \eqref{4.53}. Assuming that $m$ is even, we decompose $ h^{(2)}( \lambda_{n,m})$ as follows,
\beq\label{4.94}
h^{(2)}( \lambda_{n,m}):=h^{(2,1)}( \lambda_{n,m}) +h^{(2,2)}( \lambda_{n,m}),  
\ene
where,
\beq\label{4.95}
 h^{(2,1)}( \lambda_{n,m}):=  2\pi \sum_{ 1 \leq q \leq  m/2} \left( \frac{q \omega_{\text{\rm c}}}{q \omega_{\text{\rm c}}-  \lambda_{n,m}}  \right)^2\, a_{n,q},
\ene
and
\beq\label{4.96}
 h^{(2,2)}( \lambda_{n,m}):=  2\pi \sum_{ m/2 < q \leq  m-1} \left( \frac{q \omega_{\text{\rm c}}}{q \omega_{\text{\rm c}}-  \lambda_{n,m}}  \right)^2\, a_{n,q},
\ene
Since
$
 \left( \frac{q \omega_{\text{\rm c}}}{q \omega_{\text{\rm c}}-  \lambda_{n,m}}  \right)^2 \leq 4 \frac{q^2}{m^2}$, $ 1 \leq q \leq  \frac{m}{2}$, 
and, using \eqref{4.53} we obtain,
\beq\label{4.97}
\left| h^{(2,1)}( \lambda_{n,m})\right| \leq  2\pi \sum_{ 1 \leq q \leq  m/2} \, 4\,  \frac{q^2}{m^2}\, a_{n,q} \leq C \frac{1}{m^2}.
\ene
Furthermore, as,
$
\left( \frac{q \omega_{\text{\rm c}}}{q \omega_{\text{\rm c}}-  \lambda_{n,m}}  \right)^2 \leq q^2$, $   m/2 < q  \leq m-1$, 
and by \eqref{4.53}, we have
\beq\label{4.98}
 \left|h^{(2,2)}( \lambda_{n,m})\right| \leq  2\pi \sum_{ m/2 < q \leq  m-1} \,  q^2\,  a_{n,q} \leq C_p \, \frac{1}{m^p}
\ene
for all $ p >0.$ When $m$ is odd we decompose $h^{(2)}( \lambda_{n,m})$ as in \eqref{4.94} with 
\beq\label{4.99}
 h^{(2,1)}( \lambda_{n,m}):=  2\pi \sum_{ 1 \leq q \leq  (m-1)/2} \left( \frac{q \omega_{\text{\rm c}}}{q \omega_{\text{\rm c}}- \lambda_{n,m}}  \right)^2\, a_{n,q},
\ene
and
\beq\label{4.100}
 h^{(2,2)}( \lambda_{n,m}):=  2\pi \sum_{ (m-1)/2 < q \leq  m-1} \left( \frac{q \omega_{\text{\rm c}}}{q \omega_{\text{\rm c}}-  \lambda_{n,m}}  \right)^2\, a_{n,q},
\ene
and we prove that \eqref{4.97} and \eqref{4.98} hold arguing as in the case where $m$ is even. This proves that,
\beq\label{4.101}
 \left|
 h^{(2)}( \lambda_{n,m})\right|
  \leq C\, \frac{1}{m^2}.
\ene
The technical result \eqref{4.74} in  Appendix~ A is $ \lambda_{n,m}= m   \omega_{\text{\rm c}}+  2 \pi m  \, \omega_{\text{\rm c}}\, \frac{a_{n,|m|}}{n^2}+ a_{n,|m|}\, O\left(\frac{1}{|m|}\right) $. It yields 
\beq\label{4.102}
 h^{(3)}( \lambda_{n,m})= \frac{n^4}{2\pi}\, \frac{1}{a_{n,m}}\, \left( 1+ O\left( \frac{1}{m^2}  \right)  \right), \qquad  m \to \infty.
 \ene
Moreover, by  \eqref{4.53} and  \eqref{4.74} there is an $m_0$ such that
$$
 \left( \frac{q \omega_{\text{\rm c}}}{q \omega_{\text{\rm c}}-  \lambda_{n,m}}  \right)^2  \leq 2 \, q^2 , \qquad     q \geq  m+1,  m \geq m_0.
$$
Then, using \eqref{4.53} we obtain for all $ p >0,$
\beq\label{4.103}
\left|h^{(4)}( \lambda_{n,m})\right| \leq  4 \pi \sum_{ m+1 \leq q } q^2\, a_{n,q} \leq C_p \, \frac{1}{m^p},  \qquad m \geq m_0.
\ene
Equation \eqref{4.86} follows from \eqref{4.87},  \eqref{4.88},  \eqref{4.93}, \eqref{4.94}, \eqref{4.101}, \eqref{4.102}, and, \eqref{4.103}.
\end{proof}

Since $\mathbf Y_{n,m} \in \mathcal H$ we can define the associated normalized eigenfunctions as follows. Let us denote,
\beq\label{4.103.1.1}
b_{n,m}:= \sqrt{  \left\|\eta_{n,m}\right\|^2_{L^2(\mathbb R^2)}+ n^2}= \|Y_{n,m} \|_{\mathcal H}.
\ene
The normalized eigenfunctions are given by,
\beq\label{4.103.b}
\mathbf Z_{n,m}:= \frac{1}{b_{n,m}}\, \mathbf Y_{n,m}, \qquad n,m \in \mathbb Z^\ast.
\ene
\textcolor{black}{The normalized eigenfunctions \eqref{4.103.b} are the Bernstein modes \cite{bernstein}.}

Then, we have,
\begin{lemma}\label{lemm4.8.b}  Let $\mathbf H$ be the  \textcolor{black}{magnetized} Vlasov-Amp\`ere operator defined in \eqref{4.4} and, \eqref{4.5}. 
Let $\lambda_{n,m}, n, m \in \mathbb Z^\ast, $  be the roots to equation \eqref{4.49} obtained in   Lemma~\ref{lemm4.4}. Then, each  $\lambda_{n,m}, n,  m \in \mathbb Z^\ast, $  is an  eigenvalue of $\mathbf H$ with eigenfunction  $\mathbf Z_{n,m}.$
\end{lemma}

\begin{proof} The fact that the  $\lambda_{n,m}, n, m \in \mathbb Z^\ast,$ are eigenvalues of $\mathbf H$ with eigenfunction $\mathbf Z_{n,m}$ follows from \eqref{4.83}, \eqref{4.84.b}, and \eqref{4.86}. 
\end{proof}

In preparation for Lemma \ref{lemm4.9} below, we briefly study the asymptotic expansion for large $|m|$  of the normalized eigenfunction.  By \eqref{4.84}, \eqref{4.93}, \eqref{4.101}, and \eqref{4.103}, we have,
\beq\label{4.103.1}
\left\|  \eta_{n,m}  - e^{\frac{-r^2}{4}}\, \frac{m \omega_{\text{\rm c}}}{m \omega_{\text{\rm c}}- \lambda_{n,m}}\, e^{im\varphi}\, J_m\left(\frac{nr}{\omega_{\text{\rm c}}}\right) \right\|_{L^2(\mathbb R^2)} = O\left( \frac{1}{|m|} \right), \qquad m \to \pm \infty.
\ene
Note that \eqref{4.93}, \eqref{4.101}, and \eqref{4.103} were only proven for $m \geq 1,$ and then, they only imply \eqref{4.103.1} for $m \to \infty.$ However, using  both equations in \eqref{4.28} and as $\lambda_{n, -m}= - \lambda_{n,m}$ we prove that \eqref{4.103.1} with $m \to \infty$ implies  \eqref{4.103.1} with $m \to -\infty.$
Then, by  \eqref{4.86},
\beq\label{4.103.2}
\left\|  \frac{1}{b_{n,m}} \eta_{n,m} -  \frac{1}{ b_{n,m}}   e^{\frac{-r^2}{4}}\, \frac{m \omega_{\text{\rm c}}}{m \omega_{\text{\rm c}}-\lambda_{n,m}
}\, e^{im\varphi}\, J_m\left(\frac{nr}{\omega_{\text{\rm c}}}\right)  \right\|_{L^2(\mathbb R^2)}= \sqrt{a_{n,|m|}}  O\left( \frac{1}{m^2} \right), \qquad  m \to \pm \infty.
\ene
Let us denote,
\beq\label{4.103.3}
\eta^{(a)}_{n,m}:=- \frac{1}{\sqrt{2\pi a_{n,|m|}}}\, e^{\frac{-r^2}{4}}\, e^{im \varphi}\, J_m\left(\frac{nr}{\omega_{\text{\rm c}}}\right) , \qquad n, m \in \mathbb Z^\ast.
\ene
By \eqref{4.74} and \eqref{4.86},
\beq\label{4.103.4}
\frac{1}{b_{n,m}} \frac{m\omega_{\text{\rm c}}}{m \omega_{\text{\rm c}}-\lambda_{n,m}}= - \frac{1}{\sqrt{ 2 \pi a_{n,|m|}}} \left( 1+O\left(\frac{1}{|m|}\right)\right), \qquad m \to \pm \infty.
\ene
Then, by \eqref{4.103.3} and, \eqref{4.103.4}
\beq\label{4.103.5}
\left\|  \frac{1}{b_{n,m}}  e^{\frac{-r^2}{4}}\, \frac{m \omega_{\text{\rm c}}}{m \omega_{\text{\rm c}}-\lambda_{n,m}}\, e^{im\varphi}\, J_m\left(\frac{nr}{\omega_{\text{\rm c}}}\right) - \eta^{(a)}_{n,m} \right\|_{L^2(\mathbb R^2)}=  \,O\left(\frac{1}{|m|}\right), \qquad m \to \pm \infty.
\ene
Let us now define the asymptotic function that is the dominant term  for large $m$ of  the normalized eigenfunction $\mathbf Z_{n,m}$
\beq\label{4.103.6}
\mathbf Z^{(a)}_{n,m}:=    \frac{1}{b_{n,m}}   \frac{1}{\sqrt{2\pi}} \, e^{inx}
  \begin{pmatrix} e^{-in \frac{v_2}{\omega_{\text{\rm c}}}} \eta^{(a)}_{n,m} \\-in\end{pmatrix}, \qquad n,m \in \mathbb Z^\ast.
\ene
In the next lemma we show that for large $m$ the eigenfunction  $\mathbf Z_{n,m}$ is   concentrated in $\mathbf Z^{(a)}_{n,m}.$
\begin{lemma}\label{lemm4.9}
Let $a_{n,m}$ be the quantity defined in \eqref{4.48}, let   $\mathbf Z_{n,m}$  be the eigenfunction defined in \eqref{4.103.b}, and let $\mathbf Z^{(a)}_{n,m}$ be the  asymptotic function defined in \eqref{4.103.6}. We have that,
\beq\label{4.103.7}
\left\|\mathbf Z_{n,m}-\mathbf Z^{(a)}_{n,m}\right\|_{\mathcal H} \leq C \frac{1}{|m|}, \qquad m \to \pm \infty, n \in \mathbb Z^\ast.
\ene
\end{lemma}

\begin{proof} The lemma follows  from  \eqref{4.53}, \eqref{4.103.2}, and \eqref{4.103.5}
\end{proof}

\subsection{The completeness of the eigenfunctions of $\mathbf H$}

In this subsection we prove that the eigenfunctions of the \textcolor{black}{magnetized} Vlasov-Amp\`ere operator    $\mathbf H$ are a complete set in $\mathcal H.$ That is to say, that the closure of the set of all finite linear combinations of  eigenfunctions of $\mathbf H$ is equal to $\mathcal H,$ or in other words, that $\mathcal H$ coincides with the span of the set of all the eigenfunctions of $\mathbf H.$ For this purpose we first introduce some notation. By \eqref{4.20}
\beq\label{4.104}
L^2(0,2 \pi)= \oplus_{n \in \mathbb Z}  \text{\rm Span}\left[\frac{e^{inx}}{\sqrt{2 \pi}}\right],
\ene
and, 
\beq\label{4.105}
L^2_0(0,2 \pi)= \oplus_{n \in \mathbb Z^\ast}  \text{\rm Span}\left[\frac{e^{inx}}{\sqrt{2 \pi}}\right].
\ene
Furthermore, by \eqref{4.104} and \eqref{4.105},
\beq\label{4.106}
\mathcal H= \oplus_{n \in \mathbb Z}\, \mathcal H_n,
\ene
where
\beq\label{4.107} 
\mathcal H_0:= L^2(\mathbb R^2)\oplus \{0\},
\ene
and,
\beq\label{4.108}
\mathcal H_n:= \text{\rm Span}\left[\frac{e^{inx}}{\sqrt{2 \pi}}\right] \otimes \left( L^2(\mathbb R^2)\oplus \mathbb C\right), \qquad n \in \mathbb Z^\ast.
\ene
Alternatively,  $\mathcal H_0$ can be written as the Hilbert space of all   vector valued functions of the form
$( u,  0)^T$, $ u \in L^2(\mathbb R^2)$, 
where  the injection of  $ L^2(\mathbb R^2)$ onto the subspace of $\mathcal A$  consists of all the functions in $\mathcal A$ that are independent of $x$.
 In other words, we identify $ f(v)\in L^2(\mathbb R^2)$ with the   same function $ f(v)\in \mathcal A$ that is independent of $x.$ 
Moreover, $\mathcal  H_n$ can be written as the Hilbert space of all vector valued functions of the form,
$$
\frac{e^{inx}}{\sqrt{2 \pi}}\, \begin{pmatrix} u(v)\\ \alpha \end{pmatrix}, \qquad u \in L^2(\mathbb R^2), \alpha \in \mathbb C.
$$
Furthermore, $\mathcal H$ can be written as the Hilbert space  of all vector valued functions of the form
$$
\begin{pmatrix} u_0(v)\\ 0\end{pmatrix}+ \sum_{n \in \mathbb Z^\ast}\, \frac{e^{inx}}{\sqrt{2 \pi}}\, \begin{pmatrix} u_n(v)\\ \alpha_n \end{pmatrix},
$$
 where $ u_n \in L^2(\mathbb R^2), \alpha_n \in \mathbb C, n \in \mathbb Z^\ast,$ and, further,
 $
 \sum_{n \in \mathbb Z^\ast} \| u_n\|^2_{L^2(\mathbb R^2)} < \infty$, 
 and
 $
  \sum_{n \in \mathbb Z^\ast} |\alpha_n|^2 < \infty$. 
 The strategy of the proof that the eigenfunctions of $\mathbf H$ are complete in $\mathcal H$ will be to prove that the eigenfunctions of a given $n$ are complete on the corresponding $\mathcal H_n.$ For this purpose we introduce the following  convenient spaces. A first space is defined as follows,
 \beq\label{4.109}
 \mathbf W_0:= \text{\rm Span}\left[ \left\{ \mathbf M^{(0)}_{0,j}\right\}_{j \in \mathbb N^*}\right] \oplus \text{\rm Span}\left[ \left\{ \mathbf V_{m,j}\right\}_{m \in \mathbb Z^\ast,  j \in \mathbb N^*}\right]
 \subset \mathcal H_0 
 \ene
 where the eigenfunctions $ \mathbf M^{(0)}_{0,j},\, j \in \mathbb N^*,$ are defined in \eqref{4.18} and the eigenfunctions  $\mathbf V_{m,j}, m \in \mathbb Z^\ast,  j \in \mathbb N^*$ are defined in \eqref{4.26}.  Next we introduce the space,
 \beq\label{4.110}
 \mathbf W^{(1)}_n:=  \text{\rm Span}\left[ \left\{ \mathbf W_{n,m,j} \right\}_{ n, m \in \mathbb Z^\star, j \in \mathbb N^*}\right]
  \subset \mathcal H_n, \ \ n\neq 0, 
 \ene
 where the eigenfunctions $\mathbf W_{n,m,j}, n, m \in \mathbb Z^\star, j \in \mathbb N^*$ are defined in \eqref{4.35}.
 We also need the following space,
 \beq\label{4.111}
 \mathbf W^{(2)}_n:=  \text{\rm Span}\left[ \left\{ \mathbf Z_{n,m} \right\}_{ n, m \in \mathbb Z^\ast}\right] 
   \subset \mathcal H_n, \ \ n\neq 0, 
 \ene
 where the eigenfunctions $\mathbf Z_{n,m}, n, m \in \mathbb Z^\ast$ are defined in \eqref{4.103.b}.
Finally, we define the space,
 \beq\label{4.112}
 \mathbf W^{(3)}_n:=  \text{\rm Span}\left[\left\{\mathbf V^{(0)}_{n}\right\}_{ n \in \mathbb Z^\ast} \cup 
  \left\{ \mathbf M^{(0)}_{n,j}\right\}_{ n \in \mathbb Z^\ast, j \in  \mathbb N^*} \right]
    \subset \mathcal H_n\cap \mbox{Ker}[\mathbf H], \ \ n\neq 0, 
  \ene
  where the eigenfunctions $\mathbf V^{(0)}_{n}$ and    $\mathbf M^{(0)}_{n,j}$ are defined in \eqref{4.18}.
  \begin{theorem}\label{theo4.10}  Let $\mathbf H$ be the \textcolor{black}{magnetized} Vlasov-Amp\`ere operator defined in \eqref{4.4} and \eqref{4.5}.
  Then, the eigenfunctions of $\mathbf H$ are a complete set in $\mathcal H.$ Namely,
  \beq\label{4.113}
  \mathcal H_0= \mathbf W_0,
  \ene
  \beq\label{4.114}
  \mathcal H_n= \mathbf W^{(1)}_n\oplus \mathbf W^{(2)}_n\oplus \mathbf W^{(3)}_n, \qquad n \in \mathbb Z^\ast.
  \ene
  Furthermore,
  \beq\label{4.115}
  \mathcal H=  \mathbf W_0\oplus_{n \in \mathbb Z^\ast}\left(   \mathbf W^{(1)}_n\oplus \mathbf W^{(2)}_n\oplus \mathbf W^{(3)}_n  \right).
  \ene
  \end{theorem}
  \begin{proof} Note that    $ \mathbf W_0$ is orthogonal to $ \mathbf W^{(1)}_n,  \mathbf W^{(2)}_n,$ and $\mathbf W^{(3)}_n$ because  $ \mathbf W_0$ is the span of eigenfunctions with $n=0$ and  $ \mathbf W^{(1)}_n,  \mathbf W^{(2)}_n,$ and $\mathbf W^{(3)}_n$ are the span of eigenfunctions with $n$ different from zero.
   Furthermore,  the $ \mathbf W^{(1)}_n,  \mathbf W^{(2)}_n,$ and $\mathbf W^{(3)}_n$ are  orthogonal among themselves because they are the span of eigenfunctions with different eigenvalues. Furthermore the $ \mathbf W^{(1)}_n,  \mathbf W^{(1)}_q,$ with $n\neq q,$ are orthogonal to each other because they are the span of eigenfunctions that contain the factor, respectively, $e^{inx}, e^{iqx}$. Similarly,    $ \mathbf W^{(2)}_n,  \mathbf W^{(2)}_q,  n \neq q $ are orthogonal to each other  and    $ \mathbf W^{(3)}_n,  \mathbf W^{(3)}_q,  n \neq q $ are also orthogonal to each other.   Equation \eqref{4.113} is immediate because the span of $u_{0,m,j}, m \in \mathbb Z, j \in \mathbb N^*$ is equal to $L^2(\mathbb R^2).$ We proceed to prove \eqref{4.114}. We clearly have,
   \beq\label{4.116}
    \mathbf W^{(1)}_n\oplus \mathbf W^{(2)}_n\oplus \mathbf W^{(3)}_n \subset \mathcal H_n, \qquad n \in \mathbb Z^\ast.
    \ene
    Our goal is to prove the opposite embedding, i.e.,
     \beq\label{4.117}
   \mathcal H_n \subset \mathbf W^{(1)}_n\oplus \mathbf W^{(2)}_n\oplus \mathbf W^{(3)}_n, \qquad n \in \mathbb Z^\ast.
    \ene
 Consider the decomposition,
 \beq\label{4.118}
   \mathcal H_n= \mathbf W^{(1)}_n\oplus \left( \mathbf W^{(1)}_n\right)^\perp,
 \ene
 where 
 $
  \left( \mathbf W^{(1)}_n\right)^\perp
 $
  denotes the orthogonal complement of $\mathbf W^{(1)}_n$ in $\mathcal H_n$. 
 Recall that
  \beq\label{4.118.b}
  \mathbf W^{(2)}_n\oplus \mathbf W^{(3)}_n  \subset    \left( \mathbf W^{(1)}_n\right)^\perp    \qquad n \in \mathbb Z^\ast.
 \ene
  Our strategy to prove \eqref{4.117} will be to establish,
  \beq\label{4.119}
   \left( \mathbf W^{(1)}_n\right)^\perp \subset  \mathbf W^{(2)}_n\oplus \mathbf W^{(3)}_n,  \qquad n \in \mathbb Z^\ast.
   \ene
   It follows from the definition of  $\mathbf W^{(2)}_n$ in \eqref{4.111} and of $ \mathbf W^{(3)}_n$ in \eqref{4.112} that the following set of eigenfunctions is a basis of
   $ \mathbf W^{(2)}_n\oplus \mathbf W^{(3)}_n,$
   \beq\label{4.120}
   \left\{\begin{array}{l}  \mathbf Z_{n,m},  n, m \in \mathbb Z^\ast,\\
   \mathbf M^{(0)}_{n,j},  n \in \mathbb Z^\ast, j \in  \mathbb N^*,
  \\
   \mathbf V^{(0)}_{n}, n \in \mathbb Z^\ast.
   \end{array}\right.
   \ene
   Furthermore, it is a consequence of  the definition of $ \mathbf W^{(1)}_n$ in \eqref{4.110} and of the definition 
   of $ \mathbf Z^{(a)}_{n,m}$ in (\ref{4.103.6}) that the following set of functions is an orthonormal  basis of  $\left( \mathbf W^{(1)}_n\right)^\perp$
    \beq\label{4.121}
   \left\{\begin{array}{l}  \mathbf Z^{(a)}_{n,m},  n, m \in \mathbb Z^\ast,\\
   \mathbf M^{(0)}_{n,j},  n \in \mathbb Z^\ast, j \in  \mathbb N^*,\\
  \mathbf Q,
   \end{array}\right.,
   \ene
 where the   asymptotic functions  $\mathbf Z^{(a)}_{n,m}, n,m \in \mathbb Z^\ast$ are defined in \eqref{4.103.6}, the eigenfunctions $ \mathbf  M^{(0)}_{n,j},  n \in \mathbb Z^\ast, j \in  \mathbb N^*$  are defined in \eqref{4.18}, and $\mathbf Q$ is given by,
  \beq\label{4.122}
  \mathbf Q:=\begin{pmatrix} 0\\1\end{pmatrix}.
  \ene
  Any $ \mathbf X \in  \left( \mathbf W^{(1)}_n\right)^\perp$ can be uniquely written as,
  \beq\label{4.123}
  \mathbf X= \sum_{m \in \mathbb Z^\ast}\, \left(\mathbf X, \mathbf Z^{(a)}_{n,m} \right)_{\mathcal H}\,  \mathbf Z^{(a)}_{n,m}+ \sum_{j \in \mathbb N^*}\, \left(\mathbf X,  \mathbf M^{(0)}_{n,j}\right) _{\mathcal H}\,  \mathbf M^{(0)}_{n,j}+  (\mathbf X,\mathbf Q )_{\mathcal H}\, \mathbf Q.
  \ene
  We define the following operator from $ \left( \mathbf W^{(1)}_n\right)^\perp$ into  $\left( \mathbf W^{(1)}_n\right)^\perp,$
  \beq\label{4.124}
  \mathbf \Lambda \mathbf X:= \sum_{m \in \mathbb Z^\ast}\, \left(\mathbf X, \mathbf Z^{(a)}_{n,m} \right)_{\mathcal H}\,  \mathbf Z_{n,m}+ \sum_{j \in \mathbb N^*}\, \left(\mathbf X,  \mathbf M^{(0)}_{n,j}\right) _{\mathcal H}\,  \mathbf M^{(0)}_{n,j}+  (\mathbf X,  \mathbf Q)_{\mathcal H}\, \mathbf V^{(0)}_{n}.
  \ene
We will prove that \eqref{4.119} holds by showing that $\mathbf \Lambda$ is onto,  $ \left( \mathbf W^{(1)}_n\right)^\perp.$
We write $\mathbf \Lambda$ as follows,
\beq\label{4.125}
\mathbf \Lambda = I + \mathbf K,
\ene
where $\mathbf K$ is the operator,
\beq\label{4.126}
   \mathbf K \mathbf X:=  \sum_{m \in \mathbb Z^\ast}\, \left(\mathbf X, \mathbf Z^{(a)}_{n,m} \right)_{\mathcal H}\, \left( \mathbf Z_{n,m}- \mathbf Z^{(a)}_{n,m}\right)
   + (\mathbf X, \mathbf Q)_{\mathcal H}\,   \left(\mathbf V^{(0)}_{n} -  \mathbf Q\right).
   \ene
   We will prove that $\mathbf K$ is Hilbert-Schmidt. For information about Hilbert-Schmidt operators see Section 6 of Chapter VI of \cite{rs1}.  For this purpose  we have to prove that $\mathbf K^\ast\,\mathbf  K$ is trace class. Since  the functions in \eqref{4.121} are an orthonormal basis of     $\left( \mathbf W^{(1)}_n\right)^\perp,$ we can verify the trace class criterion under  the form,
   \beq\label{4.127}
  \sum_{m \in \mathbb Z^\ast}\, \left( \mathbf K   \mathbf Z^{(a)}_{n,m},   \mathbf K   \mathbf Z^{(a)}_{n,m}\right) _{\mathcal H}+
  \sum_{j \in \mathbb N^*}\, \left( \mathbf K \mathbf M^{(0)}_{n,j},   \mathbf K  \mathbf M^{(0)}_{n,j} \right)_{\mathcal H}+
  \left( \mathbf K \mathbf Q,   \mathbf K \mathbf Q \right) _{\mathcal H} < \infty.
  \ene
However, by \eqref{4.126}
$
 \sum_{m \in \mathbb Z^\ast}\, \left( \mathbf K   \mathbf Z^{(a)}_{n,m},   \mathbf K   \mathbf Z^{(a)}_{n,m}\right) _{\mathcal H}= 
  \sum_{m \in \mathbb Z^\ast}\, \left\| (\mathbf Z_{n,m}-\mathbf Z^{(a)}_{n,m})\right|^2 _{\mathcal H} < \infty$, 
  where we used, \eqref{4.103.7}. Moreover,
  $
   \sum_{j \in \mathbb N^*}\, \left( \mathbf K \mathbf M^{(0)}_{n,j},   \mathbf K  \mathbf M^{(0)}_{n,j} \right)_{\mathcal H}=0$, 
  and, clearly,
  $
   \left( \mathbf K \mathbf Q,   \mathbf K \mathbf Q \right) _{\mathcal H} < \infty$. 
Hence, $\mathbf K$ is Hilbert-Schmidt, and then, it is compact. It follows from the Fredholm alternative, see the Corollary in page  203 of \cite{rs1}, that to prove that $\mathbf \Lambda$ is onto it is enough to prove that it is invertible. Suppose that $\mathbf X \in $$\left( \mathbf W^{(1)}_n\right)^\perp$ satisfies
$
\mathbf \Lambda \mathbf X=0$. 
Then, by \eqref{4.124}
 \beq\label{4.127.b}
 \sum_{m \in \mathbb Z^\ast}\, \left(\mathbf X, \mathbf Z^{(a)}_{n,m} \right)_{\mathcal H}\,  \mathbf Z_{n,m}+ \sum_{j \in \mathbb N^*}\, \left(\mathbf X,  \mathbf M^{(0)}_{n,j}\right) _{\mathcal H}\,  \mathbf M^{(0)}_{n,j}+  (\mathbf X,  \mathbf Q)_{\mathcal H}\, \mathbf V^{(0)}_{n}=0.
  \ene
 However, as the eigenfunctions $ \mathbf Z_{n,m}$ are orthogonal to the  $ \mathbf M^{(0)}_{n,j}$and to  $ \mathbf V^{(0)}_{n},$ we have,
  \beq\label{4.128}
   \sum_{m \in \mathbb Z^\ast}\, \left(\mathbf X, \mathbf Z^{(a)}_{n,m} \right)_{\mathcal H}\,  \mathbf Z_{n,m}=0,
  \ene
  and, 
  \beq\label{4.129}
   \sum_{j \in \mathbb N^*}\, \left(\mathbf X,  \mathbf M^{(0)}_{n,j}\right) _{\mathcal H}\,  \mathbf M^{(0)}_{n,j}+  (\mathbf X,  \mathbf Q)_{\mathcal H}\, \mathbf V^{(0)}_{n}=0.
\ene
Since the eigenfunctions  $ \mathbf Z_{n,m}$ are mutually orthogonal, it follows from \eqref{4.128} that 
$
\left(\mathbf X, \mathbf Z^{(a)}_{n,m} \right)_{\mathcal H}=0$, $ m \in \mathbb Z^\ast$. 
Moreover, by Lemma ~\ref{lemm4.0} the eigenfunctions $\mathbf M^{(0)}_{n,j}, j \in \mathbb N^*$ and $, \mathbf V^{(0)}_{n}$ are linearly independent, and, then \eqref{4.129} implies
$
\left(\mathbf X,  \mathbf M^{(0)}_{n,j}\right) _{\mathcal H}=0, \qquad j \in \mathbb N^*$, 
and
$
 (\mathbf X,  \mathbf Q)_{\mathcal H}=0$. 
Finally, as the set   \eqref{4.121} is an orthonormal basis of   $\left( \mathbf W^{(1)}_n\right)^\perp$ we have that $\mathbf X=0.$ Then, $\mathbf \Lambda$ is onto
$\left( \mathbf W^{(1)}_n\right)^\perp$ and \eqref{4.119} holds. Since also \eqref{4.118.b} is satisfied we obtain 
 $
  \mathbf W^{(2)}_n\oplus \mathbf W^{(3)}_n = \left( \mathbf W^{(1)}_n\right)^\perp$, $ n \in \mathbb Z^\ast.
 $
This completes the proof of the theorem
\end{proof}

\begin{theorem}
Let $\mathbf H$ be the \textcolor{black}{magnetized} Vlasov-Amp\`ere operator defined in \eqref{4.4} and, \eqref{4.5}. Then, $\mathbf H$ is selfadjoint and it has pure point spectrum. The eigenvalues of $\mathbf H$ are given by.
\begin{enumerate}
\item
The infinite multiplicity eigenvalues, $  \lambda^{(0)}_m:= m \omega_{\text{\rm c}}, m \in \mathbb Z.$
\item
The simple eigenvalues $\lambda_{n,m}, n, m \in \mathbb Z^\ast, $  given by the roots to equation \eqref{4.49} obtained in   Lemma~\ref{lemm4.4}. 
\end{enumerate}
\end{theorem}

\begin{proof} We already proven that $\mathbf H$ is selfadjoint below \,\eqref{4.8}. The spectrum of $\mathbf H$ is pure point because it has a  complete set of eigenfunctions, as we proven in Theorem~ \ref{theo4.10}. 
The fact that the eigenvalues of $\mathbf H$ are equal to the $\lambda^{(0)}_m, m \in \mathbb Z,$ and the $\lambda_{n,m},n, m \in \mathbb Z^\ast$ follows from Lemmata~\ref{lemm4.0}, \ref{lemm4.1}, \ref{lemm4.2} and \ref{lemm4.8.b}. The  $\lambda^{(0)}_m, m \in \mathbb Z$ have infinite multiplicity because by Lemmata~\ref{lemm4.0}, \ref{lemm4.1}, and \ref{lemm4.2}  each $\lambda^{(0)}_m $ has a countable set  of orthogonal eigenfunctions. Let us prove that the eigenvalues $\lambda_{n,m}$ are simple. Suppose that for some $n,m \in \mathbb Z^\ast$ the eigenvalue  $\lambda_{n,m}$ has  multiplicity bigger than one. Then, there is an eigenfunction, $\mathbf P,$ such that
$
\mathbf H P = \lambda_{n,m}\, \mathbf P$,  
and with $\mathbf P$ orthogonal to $\mathbf Z_{n,m}.$ However since by  Lemma~\ref{lemm4.4}  $\lambda_{n_1,m_1} =\lambda_{n_2,m_2}$ if and only if $n_1=n_2,$  and $m_1=m_2,$  it follows that $\mathbf P$ is orthogonal to the 
\textcolor{black}{right}
 hand side of \eqref{4.115}, but hence,  $\mathbf P$ is orthogonal to $\mathcal H,$ and then $\mathbf P=0.$ This completes the proof  that the $\lambda_{n,m}$ are simple eigenvalues.
\end{proof}

\subsection{Orthonormal basis  for the kernel of $\mathbf H$}
In Subsection ~\ref{kernel} we constructed a linear independent basis for the kernel of the  \textcolor{black}{magnetized} Vlasov-Amp\`ere operator $\mathbf H.$ 
In this subsection we prove that, for an appropriate choice of the orthonormal basis of $L^2(\mathbb R^+, r dr)$ that appears  in the definition of the eigenfunctions $\mathbf M^{(0)}_{n,j}$ in \eqref{4.18},
 we can construct an  orthonormal basis  for the kernel of $\mathbf H.$ The choice of the orthonormal basis is $n$ dependent. For $ n \in \mathbb Z^\ast,$  let  $\tau_j^{(n)}, j=1, \dots$ \modifsF{be any orthonormal basis of $L^2(\mathbb R^+, r dr)$ where the first basis function is}
\beq\label{4.130}
\tau_1^{(n)}(r):= \frac{1}{\sqrt{a_{n,0}}}\, e^{\frac{-r^2}{4}}\, J_0\left(\frac{nr}{\omega_{\text{\rm{c}}}}\right), \qquad n \in \mathbb Z^\ast,
\ene
with $a_{n,0}$ defined in \eqref{4.48}. Note that this implies that the  $\tau_j^{(n)}, j=2, \dots$ is an orthonormal basis of the subspace $V_{n,0}$ that we defined in \eqref{4.33}. 
Moreover, in the definition of the  $u_{n,0,j}$ in \eqref{3.4} let us use this basis. In particular it yields
\beq\label{4.131}
u_{n,0,1}:=  \frac{e^{i n (x-\frac{v_2}{   \omega_{\text{\rm {c}}}})}} {\sqrt{2 \pi}}  
  \,\frac{1}{\sqrt{2\pi}}\,  \frac{1}{\sqrt{a_{n,0}}}\, e^{\frac{-r^2}{4}}\, 
J_0\left(\frac{nr}{\omega_{\text{\rm{c}}}}\right), \qquad n \in \mathbb Z^\ast.
\ene 
\modifsF{
The eigenfunctions $\mathbf M^{(0)}_{n,j}= \begin{pmatrix} u_{n,0,j}\\0 \end{pmatrix},  n \in \mathbb Z^\ast, j\in \mathbb N^*$, of $\mathbf H$ precised with \eqref{4.130} are now a particular case of the ones defined in \eqref{4.18}.
 However, we keep the same notation for $\mathbf M^{(0)}_{n,j}$ for a sake of readability}.

 For the other eigenfunctions we can use different orthonormal basis of $L^2(\mathbb  R^+,r dr),$  if we find it convenient. It follows from simple calculations that the eigenfunctions $\mathbf V^{(0)}_n, n \in \mathbb Z^\ast,$ defined in \eqref{4.18} are mutually orthogonal and that the eigenfunctions $\mathbf  M^{(0)}_{n,j},  (n,j) \in \mathbb Z^\ast \times \mathbb N^*,$ are also mutually orthogonal. Moreover, since the functions $e^{inx}, n \in \mathbb Z^\ast$ are orthogonal in $L^2(0,2\pi)$ to the function equal to one,  the eigenfunctions  $\mathbf V^{(0)}_n, n \in \mathbb Z^\ast,$ and $\mathbf M^{(0)}_{0,j}, j\in \mathbb N^*$ are orthogonal.   Let us compute the scalar product of the  $\mathbf V^{(0)}_n, n \in \mathbb Z^\ast,$ and the $\mathbf  M^{(0)}_{n,j},  n \in \mathbb Z^\ast , j=1,\dots.$
\beq\label{4.133}
\left( \mathbf V^{(0)}_n,   \mathbf  M^{(0)}_{m,j}\right)_{\mathcal H}= \delta_{n,m}\frac{1}{\sqrt{2\pi+n^2}}\,\left( e^{\frac{-v^2}{4}},    \frac{ e^{- i n \frac{v_2}{\omega_{\text{\rm c}}}
}}{\sqrt{2 \pi}}    \, \tau_j(r)   \right)_{\mathcal A}, \qquad n \in \mathbb Z^\ast,  m \in\mathbb Z^\ast, j \in \mathbb N^*.
\ene
Moreover, by \textcolor{black}{the Jacobi-Anger} formula \eqref{4.29}, with  $z= \frac{-nr}{\omega_{\text{\rm{c}}}},$ 
$$
\left( e^{\frac{-v^2}{4}},    \frac{ e^{- i n \frac{v_2}{\omega_{\text{\rm c}}}
}}{\sqrt{2 \pi}}    \, \tau_j(r)   \right)_{\mathcal A}=  \left( e^{\frac{-v^2}{4}},  \left(\sum_{m \in \mathbb Z} e^{im\varphi}\, J_m\left( \frac{-nr}{\omega_{\text{\rm{c}}}}\right)\right)
\frac{1} {\sqrt{2 \pi}}    \, \tau_j(r)   \right)_{\mathcal A},  \qquad n \in \mathbb Z^\ast,  j \in \mathbb N^*.
$$ 
Hence, by \eqref{4.130} and the second equation in \eqref{4.28}
\beq\label{4.134}
\left( e^{\frac{-v^2}{4}},    \frac{ e^{- i n \frac{v_2}{\omega_{\text{\rm c}}}
}}{\sqrt{2 \pi}}    \, \tau_j(r)   \right)_{\mathcal A}= \delta_{j,1} \sqrt{2\pi}\, \sqrt{a_{n,0}}, \qquad n \in \mathbb Z^\ast. 
\ene
By \eqref{4.133} and \eqref{4.134},
\beq\label{4.135}
\left( \mathbf V^{(0)}_n,   \mathbf  M^{(0)}_{m,j}\right)_{\mathcal H}=  \delta_{n,m} \, \delta_{j,1}\,   \frac{\sqrt{2 \pi a_{n,0}}}{\sqrt{2\pi+n^2}}, \qquad n \in \mathbb Z^\ast, m \in \mathbb Z^\ast, j \in \mathbb N^*.
\ene
This proves that the  $\mathbf V^{(0)}_n, n \in \mathbb Z^\ast,$ and the $\mathbf  M^{(0)}_{n,j},  n \in \mathbb Z^\ast , j=2,\dots,$ are orthogonal to each other, and  also that
 $\mathbf V^{(0)}_n,$ and $  \mathbf  M^{(0)}_{n,1}, n \in \mathbb Z^\ast,$ are not orthogonal. We apply the Gramm-Schmidt orthonormalization process to   $\mathbf V^{(0)}_n,$ and $  \mathbf  M^{(0)}_{n,1}, n \in \mathbb Z^\ast,$  and we define the eigenfunctions,
 \beq\label{4.136}
 \mathbf E^{(0)}_{n}:= \mathbf M^{(0)}_{n,1}- \left(  \mathbf M^{(0)}_{n,1},  \mathbf V^{(0)}_n\right)_{\mathcal H}\,      \mathbf V^{(0)}_n, \qquad n \in \mathbb Z^\ast,
 \ene
 and the normalized eigenfunctions,
 \beq\label{4.137}
 \mathbf F^{(0)}_{n} := \frac{\mathbf E^{(0)}_{n}}{\|\mathbf E^{(0)}_{n}\|_{\mathcal H}}, n \in \mathbb Z^\ast.
 \ene
 By \eqref{4.135},  \eqref{4.136}, and \eqref{4.137},
 \beq\label{4.138}
  \mathbf F^{(0)}_{n} = \frac{2\pi+n^2 }{2\pi(1-a_{n,0})+ n^2}\, \left(    \mathbf M^{(0)}_{n,1}- \,   \frac{\sqrt{2 \pi a_{n,0}}}{\sqrt{2\pi+n^2}}\,  \mathbf V^{(0)}_n\right), \qquad n \in \mathbb Z^\ast.
  \ene
  Note that by \eqref{4.46} and \eqref{4.48} $ a_{n,0} < 1,$ and then $ 1- a_{n,0} >0.$
  
  Using the results above we prove the following theorem. 
  \begin{theorem}\label{ortokernel}
Let $\mathbf H$ be the \textcolor{black}{magnetized} Vlasov-Amp\`ere operator defined in \eqref{4.4} and  \eqref{4.5}. Then, the following set of eigenfunctions of $\mathbf H$ with eigenvalue zero,
\beq\label{4.139}
\left\{\mathbf V^{(0)}_{n}, n \in \mathbb Z^\ast\right\} \ds\cup \left\{  \mathbf M^{(0)}_{0,j}, j\in \mathbb N^*  \right\}\cup
\left\{\mathbf M^{(0)}_{n,j}, n \in \mathbb Z^\ast, j =2,\dots\right\}  \cup \left\{  \mathbf F^{(0)}_n, n \in \mathbb Z^\ast  \right\},
\ene
is \textcolor{black}{an} orthonormal  basis of $\text{\rm Ker}[\mathbf H].$  
The eigenfunctions $\mathbf V^{(0)}_{n},$ and  $\mathbf M^{(0)}_{0,j},$ are defined in \eqref{4.18}, and the eigenfunctions, $\mathbf M^{(0)}_{n,j},$ and $\mathbf F^{(0)}_n,$  are defined, respectively, in  \modifsF{\eqref{4.18} with \eqref{4.131}}, and \eqref{4.137}. 
\end{theorem}  
 
\begin{proof} The lemma follows from Lemma ~\ref{lemm4.0}
\end{proof}

\subsection{Orthonormal basis with eigenfunctions of $\mathbf H$}
\label{obei}
In this subsection we show how to assemble  a orthonormal basis for $\mathcal H$ with eigenfunctions of $\mathbf H,$ using the eigenfunctions that we have already computed.   We first obtain a orthonormal basis for $\text{\rm{Ker}}[ \mathbf H]^\perp,$ with the eigenfunctions of $\mathbf H$ with eigenvalue different from zero.

\begin{theorem} \label{bno}
Let $\mathbf H$ be the \textcolor{black}{magnetized} Vlasov-Amp\`ere operator defined in \eqref{4.4} and, \eqref{4.5}. Then, the following set of eigenfunctions of $\mathbf H$ with eigenvalue different from zero,
\beq\label{4.139.b}
\left\{\mathbf V_{m,j}, m \in \mathbb Z^\ast,  j \in \mathbb N^*\right\}\cup \left\{ \mathbf W_{n,m,j},  n, m \in \mathbb Z^\ast, j \in \mathbb N^* \right\}\cup \left\{  \mathbf Z_{n,m}, n, m \in \mathbb Z^\ast\right\},
\ene
is a orthonormal basis of  $\text{\rm{Ker}}[ \mathbf H]^\perp.$ Moreover, the eigenfunctions $\mathbf V_{m,j}, \mathbf W_{n,m,j},$ and  
$\mathbf Z_{n,m}$ are defined, respectively in \eqref{4.26}, \eqref{4.35}, and \eqref{4.103.b}.

\end{theorem}

\begin{proof}
Equation \eqref{4.115} can be written as follows,
\beq\label{4.140}
\mathcal H= \left[\text{\rm {Span}}\left[ \left\{ \mathbf M^{(0)}_{0,j}\right\}_{j \in \mathbb N^*}\right] \oplus_{n \in \mathbb Z^\ast} \mathbf W^{(3)}_n  \right] 
\oplus\left[   \text{\rm Span}\left[ \left\{ \mathbf V_{m,j}\right\}_{m \in \mathbb Z^\ast,  j \in \mathbb N^*}\right]\oplus
   \mathbf W^{(1)}_n\oplus \mathbf W^{(2)}_n \right].  
\ene
Moreover, by Lemma~\ref{lemm4.0}
\beq\label{4.141}
\text{\rm{Ker}}[ \mathbf H]= \text{\rm {Span}}\left[ \left\{ \mathbf M^{(0)}_{0,j}\right\}_{j \in \mathbb N^*}\right] \oplus_{n \in \mathbb Z^\ast} \mathbf W^{(3)}_n.
\ene
Further, as
$
\mathcal H = \text{\rm{Ker}}[ \mathbf H]\oplus \text{\rm{Ker}}[ \mathbf H]^\perp$, 
it follows  from \eqref{4.140}, \eqref{4.141}
\beq\label{4.142}
 \text{\rm{Ker}}[ \mathbf H]^\perp=   \text{\rm Span}\left[ \left\{ \mathbf V_{m,j}\right\}_{m \in \mathbb Z^\ast,  j \in \mathbb N^*}\right]\oplus
   \mathbf W^{(1)}_n\oplus \mathbf W^{(2)}_n.  
\ene
Finally, using the definitions of $ \mathbf W^{(1)}_n$ in \eqref{4.110} and of $\mathbf W^{(2)}_n$ in \eqref{4.111} we obtain that the set \eqref{4.139.b} is an orthonormal basis of $\text{\rm{Ker}}[ \mathbf H]^\perp.$ 
\end{proof}

In the following theorem we present a orthonormal basis  for $\mathcal H$ with eigenfunctions of $\mathbf H.$

\begin{theorem} \label{bnp}
Let $\mathbf H$ be the \textcolor{black}{magnetized} Vlasov-Amp\`ere operator defined in \eqref{4.4}, and \eqref{4.5}. Then, the following set of eigenfunctions of $\mathbf H,$
\beq\label{4.143}
\begin{array}{l}
\left\{\mathbf V^{(0)}_{n}, n \in \mathbb Z^\ast\right\} \ds\cup \left\{ \mathbf M^{(0)}_{0,j}, j\in \mathbb N^*  \right\}\cup
\left\{\mathbf M^{(0)}_{n,j}, n \in \mathbb Z^\ast, j =2,\dots\right\}  \cup \left\{\mathbf F^{(0)}_n , n \in \mathbb Z^\ast  \right\} \cup\\
\left\{\mathbf V_{m,j},m  \in \mathbb Z^\ast,  j \in \mathbb N^*\right\}\cup \left\{ \mathbf W_{n,m,j},  n, m \in \mathbb Z^\ast, j \in \mathbb N^* \right\}\cup \left\{  \mathbf Z_{n,m}, n, m \in \mathbb Z^\ast\right\},
\end{array}
\ene
is a orthonormal basis of $\mathcal H.$ The eigenfunctions,$ \mathbf V^{(0)}_{n},$  and $\mathbf M^{(0)}_{0,j}$ are defined in \eqref{4.18}. The eigenfunctions, $ \mathbf M^{(0)}_{n},$ and  $\mathbf F^{(0)}_n$ are defined, respectively in \modifsF{\eqref{4.18} with \eqref{4.131}}, and \eqref{4.137}. Moreover, the eigenfunctions $\mathbf V_{m,j}, \mathbf W_{n,m,j},$ and  
$\mathbf Z_{n,m}$ are defined, respectively in \eqref{4.26}, \eqref{4.35}, and \eqref{4.103.b}.
 \end{theorem}

\begin{proof} The result follows from Theorems~\ref{ortokernel}, and ~\ref{bno}.
\end{proof}

\section{The general solution to the \textcolor{black}{magnetized} Vlasov-Amp\`ere  system, and the Bernstein-Landau paradox }
\label{paradox}\sss
In this section we give an explicit formula for the general solution of the Vlasov-Amp\`ere system  with the help of the orthonormal basis of $\mathcal H$ with eigenfunctions of $\mathbf H.$ Let us take a general initial state,
$$
\mathbf G_0= \begin{pmatrix} u \\ F\end{pmatrix}\in \mathcal H.
$$
Then, by Theorem~\ref{bnp},  the general solution to the \textcolor{black}{magnetized} Vlasov-Amp\`ere system with initial value at $t=0$ equal to $\mathbf G_0$ is given by,
\beq\label{4.144}
\mathbf G(t):= e^{-i t \mathbf H}\, \mathbf G_0,
\ene
and, furthermore,
\beq\label{4.145}
\mathbf G(t)= \mathbf G_1+\mathbf G_2(t),
\ene
where the static parts $\mathbf G_1$ is time independent, and the dynamical part $\mathbf G_2(t)$ is oscillatory in time. They are given by,
\beq\label{4.146}
\begin{aligned}
\mathbf G_1=
&\sum_{ n \in \mathbb Z^\ast} \left( \mathbf G_0,\mathbf V^{(0)}_{n}\right)_{\mathcal H} \, \mathbf V^{(0)}_{n} +\sum_{ j \in \mathbb N^* } \left( \mathbf G_0, \mathbf M^{(0)}_{0,j}\right)_{\mathcal H} \,
\mathbf M^{(0)}_{0,j} +\sum_{  n \in \mathbb Z^\ast, j \geq 2} \left( \mathbf G_0, \mathbf M^{(0)}_{n,j}\right)_{\mathcal H} \, \mathbf M^{(0)}_{n,j} \\
&+\sum_{  n \in \mathbb Z^\ast} \left( \mathbf G_0,\mathbf F^{(0)}_n\right)_{\mathcal H} \, \mathbf F^{(0)}_n,
\end{aligned} 
 \ene
 and
 \beq\label{4.146.b}
\begin{aligned}
 \mathbf G_2(t)=
&\sum_{ m \in \mathbb Z^\ast,  j \in \mathbb N^* } \, e^{-it \lambda^{(0)}_m}\, \left( \mathbf G_0,\mathbf V_{m,j},\right)_{\mathcal H} \, \mathbf V_{m,j}+\sum_{  n, m \in \mathbb Z^\ast, j \in \mathbb N^* } \, e^{-it \lambda^{(0)}_m}\, \left(\mathbf G_0,\mathbf W_{n,m,j} \right)_{\mathcal H} \, \mathbf W_{n,m,j}
 \\ 
&+\sum_{ n, m \in \mathbb Z^\ast} \, e^{-it \lambda_{n,m}}\,
  \left(\mathbf G_0,\mathbf Z_{n,m} \right)_{\mathcal H} \, \mathbf Z_{n,m}.
\end{aligned}
 \ene
We still have to impose the Gauss law \eqref{gauss}, \eqref{density}, or equivalently \eqref{4.147}, to our general solution to the \textcolor{black}{magnetized} Vlasov-Amp\`ere  system \modifsF{\eqref{vlasamp}}.  \modifsF{For the eigenfunction $ \mathbf M^{(0)}_{n,j}, n \in \mathbb Z^\ast, j \geq 2,$  the Gauss law  \eqref{4.147} is equivalent to
$$
\left(\mathbf M^{(0)}_{n,j}, \mathbf V^{(0)}_n\right)_{\mathcal H}=0,
$$
} that is valid by the orthogonality of the $ \mathbf M^{(0)}_{n,j}$ and the $ \mathbf V^{(0)}_{n}.$  We prove in the same way that the Gauss law \eqref{4.147} holds for the eigenfunctions   
$\mathbf  F^{(0)}_n,  \mathbf W_{n,m,j},$ and   $\mathbf Z_{n,m}.$  We prove that $\mathbf V_{m,j}$ satisfies the Gauss law by direct computation. It remains to consider the  eigenfunctions $ \mathbf M^{(0)}_{0,j}, j \in \mathbb N^*,$ defined in \eqref{4.18}.
For the  $ \mathbf M^{(0)}_{0,j},$ the Gauss law \eqref{4.147} reads,
\beq\label{4.149}
\int_0^\infty  \, e^{\frac{-v^2}{4}}\, \tau_j\, dv=0, \qquad j\in \mathbb N^*.
\ene
 We can make sure that \eqref{4.149} holds for all but one $j$ by choosing the orthonormal basis in  $L^2(\mathbb R^+, r dr)$ that we use in the definition of the   $ \mathbf M^{(0)}_{0,j}, j \in \mathbb N^*,$  as follows. \modifsF{As we proceed in \eqref{4.130}-\eqref{4.131} for  $n\in \mathbb Z^\ast$, we specify the choice of the orthonormal basis $(\tau_j)_{j\in \mathbb N^*}$ in \eqref{3.4} and \eqref{4.18} for $n=0$.
We take a orthonormal basis, $ \tau_j^{(0)}, j\in \mathbb N^*,$ in  $L^2(\mathbb R^+, rdr),$   such that,
 \beq\label{4.150}
 \tau_1^{(0)}(r):= e^{\frac{-v^2}{4}}.
 \ene
}
With this choice of the  $ \tau_j^{(0)}, j\in \mathbb N^*,$ the Gauss law \eqref{4.147} holds for   $ \mathbf M^{(0)}_{0,j}, j=2,\dots.$ Hence, with this choice, the  general solution of the \textcolor{black}{magnetized} Vlasov-Amp\`ere system given in \eqref{4.144} and that satisfies the Gauss law \eqref{4.147} can be written as in \eqref{4.145} with the dynamical part $\mathbf G_2(t)$ as in \eqref{4.146.b}, but with the static part $\mathbf G_1$ given by

\beq\label{4.151}\mathbf G_1=\sum_{ j \geq 2 } \left( \mathbf G_0, \mathbf M^{(0)}_{0,j}\right) \,
\mathbf M^{(0)}_{0,j} +\sum_{  n \in \mathbb Z^\ast, j \geq 2} \left(\mathbf G_0, \mathbf M^{(0)}_{n,j}\right) \, \mathbf M^{(0)}_{n,j} +
\sum_{  n \in \mathbb Z^\ast} \left( \mathbf G_0,\mathbf F^{(0)}_n\right) \, \mathbf F^{(0)}_n.
 \ene
This exhibits   the Landau-Bernstein paradox. Namely, the general solution contains a time independent part and a part that is oscillatory time. There is no part of the solution   
that tends to zero as $t \to \pm \infty,$ that is to say, there is no Landau damping in the presence of the magnetic field.
\begin{remark}\label{remker}  {\rm This remark concerns the space $\mathcal H_{\rm G}$  for the Gauss law and its orthogonal complement.   

 Let us denote, 
$$
\mathcal H_{\rm G}:= {\rm Span}\left [ \left\{ \mathbf  V^{(0)}_n, n \in \mathbb Z^\ast \right\}\right.
\left.\cup\, \mathbf M^{(0)}_{0,1} \right ],
$$
where the eigenfunctions $ \mathbf  V^{(0)}_n$ are defined in \eqref{4.18} and the eigenfunction $\mathbf M^{(0)}_{0,1}$ is defined in \eqref{4.18}, \eqref{4.150}. 
Note that it follows from the results above that the condition that each one of the eigenfunctions that appear in  \eqref{4.146}, and  \eqref{4.146.b} satisfies the Gauss law is equivalent to ask that the eigenfunction is orthogonal to $\mathcal H_{\rm G}.$ Then, it follows from  \eqref{4.144}, \eqref{4.145}, \eqref{4.146}, and  \eqref{4.146.b},   that  general solution to the  \textcolor{black}{magnetized} Vlasov-Amp\`ere system given in \eqref{4.144} satisfies the Gauss law \eqref{4.147} if and only if $\mathbf G_0\in \mathcal H_{\rm G }^\perp.$ 

The Hilbert space $\mathcal H_{\rm G}$  is a closed subspace of  the kernel of $\mathbf H.$  So, the Gauss law is equivalent to have the initial state in the orthogonal complement to a closed subspace of the kernel of $\mathbf H.$ Actually, it is usually the case that when the Maxwell equations are formulated  as a selfadjoint Schr\" odinger equation in the Hilbert space of electromagnetic fields with finite energy, the Gauss law is equivalent  to have the initial data in the orthogonal complement of the kernel of the Maxwell operator. See for example \cite{rw}.
Let us further elaborate \textcolor{black}{on} the condition   $\mathbf G_0\in \mathcal H_{\rm G }^\perp.$ We introduce the  space of test functions
$
\mathcal D_\mathbb T:= \{  \varphi \in C^\infty[0,2\pi]: \frac{d^l}{dx^l}\varphi(0)= \frac{d^l}{dx^l}\varphi(2\pi), l=0,\dots\}$. 
Let us expand $\varphi \in \mathcal D_\mathbb T$ in Fourier series
\beq\label{ww.1}
\varphi(x)= \sum_{n \in \mathbb Z} \frac{1}{\sqrt{2\pi}}\, e^{inx} \varphi_n, \qquad \mbox{\modifsF{where}}
\qquad \varphi_n:= \frac{1}{\sqrt{2\pi}}\int_0^{2\pi}\, \varphi(x)\, e^{-in x}\, dx,n\in \mathbb Z.
\ene
Integrating by parts we prove that 
\beq\label{ww.2}
\left|  \varphi_n  \right| \leq  \frac{C_l}{|n|^l}, \quad l\in \mathbb N, n \in \mathbb Z^\ast.
\ene
By a simple calculation, and using \eqref{ww.1} and \eqref{ww.2} we prove that,
\modifsF{
\beq\label{ww.3}
\begin{pmatrix} \varphi(x)\, e^{\frac{-v^2}{4}}\\ - \frac{d}{dx} \varphi(x) \end{pmatrix} 
= \sum_{n \in \mathbb Z^\ast} \varphi_n \sqrt{2\pi+n^2} \,  V^{(0)}_n + \sqrt{2\pi}\, \varphi_0\, \mathbf M^{(0)}_{0,1}
\in \mathcal H_{\rm G}.
\ene
}
Suppose that
\beq\label{ww.4}
\begin{pmatrix} u(x,v)\\  F(x) \end{pmatrix} \in \mathcal H_{\rm G}^\perp.
\ene
Then, by \eqref{ww.3}
\beq\label{ww.5}
\left(  \begin{pmatrix} u(x,v)\\  F(x) \end{pmatrix},    \begin{pmatrix} \varphi(x)\, e^{\frac{-v^2}{4}}\\ - \frac{d}{dx} \varphi(x) \end{pmatrix}    \right)_\mathcal H= 
\int_0^{2\pi}\, \rho(x)\, \varphi(x)\, dx- \int_0^{2\pi} F(x) \frac{d}{dx}\varphi(x)\, dx=0, \qquad \varphi \in \mathcal D_\mathbb T,
\ene
where $\rho(x)$ is defined in \eqref{density}.   By \eqref{ww.5} we see $(u,F)^T$ satisfies the Gauss law \eqref{gauss}, \eqref{density}, or equivalently \eqref{4.147}, in weak sense, where the weak derivatives are defined with respect  to the test space $\mathcal D_\mathbb T.$ Conversely, if $(u,F)^T$ satisfies \eqref{ww.5} for all $\varphi \in \mathcal D_\mathbb T,$ we prove in a similar way that 
\eqref{ww.4} holds taking $\varphi(x)= e^{inx}, n \in \mathbb Z.$ 
}
\end{remark}

\begin{remark}{\rm
Observe that the  general solution of the \textcolor{black}{magnetized} Vlasov-Amp\`ere system, 
\textcolor{black}{$\mathbf G(t)= (u(t,x,v), F(t,x))^T  $} given in \eqref{4.144} and that satisfies the Gauss law \eqref{4.147} fulfills the condition that the total charge fluctuation is
equal to zero,
\beq\label{4.152.b}
\int_{[0,2\pi]\times \mathbb R^2}\, u(t,x,v)e^{\frac{-v^2}{4}}\, dx\, dv=0.
\ene
This true because each one of the the eigenfunctions that appear in the expansion \eqref{4.145}, with $G_2(t)$ as in \eqref{4.146.b} and $G_1(t)$ as in \eqref{4.151} satisfy this condition.}
\end{remark}
\qed

\textcolor{black}{
Let us now consider the expansion of the charge density fluctuation of the perturbation to the Maxwellian equilibrium state, $\rho(t,x),$ that we defined in \eqref{density}. We  compute the expansion of $\rho(t,x)$ multiplying the first component of the left- and right- hand sides of \eqref{4.145},  by $e^{-v^2/4}$,   integrating  both sides of the resulting equation  over $ v \in \mathbb R^2$ , and using  \eqref{4.146.b}, and \eqref{4.151}. 
For this purpose, note that for a function in $ (u,0)^T \in  \mathcal H$ with electric field zero the Gauss law \eqref{gauss}, \eqref{density} implies that the charge  density fluctuation of the function is zero. In particular the charge  density fluctuation of the eigenfunctions $\mathbf M^{(0)}_{0,j}, j \geq2,  \mathbf M^{(0)}_{n,j},  n \in \mathbb Z^\ast, j \geq 2, \mathbf V_{m,j}, m \in \mathbb Z^\ast,  j \in \mathbb N^*, \mathbf W_{n,m,j},  n, m \in \mathbb Z^\ast, j \in \mathbb N^*$ is equal to zero. 
Then, if we apply the expansion \eqref{4.145}, with $G_1(t)$ given in \eqref{4.151} and $G_2(t)$  given in \eqref{4.146.b}
to the charge  density fluctuation of the general solution to the \textcolor{black}{magnetized} Vlasov-Amp\`ere system \eqref{4.144} that satisfies the Gauss law, only the terms with  , $\mathbf F^{(0)}_n,  n \in \mathbb Z^\ast,$ and  $\mathbf Z_{n,m}, n, m \in \mathbb Z^\ast$   survive.  and we obtain,
\beq\label{4.152.0}
\rho(t,x)= \rho_{\text{\rm stat}}(x)+ \rho_{\text{\rm din}}(t,x),
\ene
where,
\beq\label{4.152.1}
\rho_{\text{\rm stat}}:=-  \sum_{n \in \mathbb Z^\ast}\, \left( \mathbf G_0, \mathbf F^{(0}_n\right)\, \int_{\mathbb R^2}\, \mathbf F^{(0,1)}(x,v) \, e^{\frac{-v^2}{4}}\, dv,
\ene
is the static part of the charge  density fluctuation,  and where  $\mathbf F^{(0,1)}(x,v) $ is the first component  of $\mathbf F^{(0)}_n.$ Moreover,
\beq\label{4.152}
\rho_{\text{\rm din}}(t,x)=\sum_{n,m \in \mathbb Z^\ast}\,  e^{-it \lambda_{n,m}} \,  \left(\mathbf G_0,\mathbf Z_{n,m} \right)_{\mathcal H} \, \rho_{n,m}(x),
\ene
is the time dependent part of the charge density fluctuation. Here, $ \rho_{n,m}(x)$ is the charge density fluctuation of the eigenfunction $\mathbf Z_{n,m}$ that is given by
\beq\label{4.153}
\rho_{n,m}(x)= - \frac{1}{b_{n,m} \sqrt{2 \pi} }\, e^{inx}\,  \int_{\mathbb R^2}\, e^{\frac{-r^2}{4}} \, e^{-in \frac{v_2}{\omega_{\rm c}}}\, \eta_{n,m}(v)\, dv, n,m \in \mathbb Z^\ast,
\ene
where we used \eqref{4.103.b}.}
\textcolor{black}{
 The right-hand side of  \eqref{4.152} is the expansion of the charge density fluctuation in the Bernstein modes, \cite{bernstein}, \cite{bedro}. Note however, that for general initial data there is also the static part of the charge density fluctuation \eqref{4.152.1}, that is not reported in \cite{bernstein}, \cite{bedro}. This means that the Bernstein modes $ \mathbf Z_{n,m},  n, m \in \mathbb Z^\ast $ are not complete, and that to expand the charge density fluctuation, $\rho(t,x)$  with the general initial data, $\mathbf G_0$ that has finite energy and that satisfies the Gauss law, one has to   add the contribution of the static part $\rho_{\text{\rm stat}}(x)$ given by the modes, $\mathbf F^{(0)}_n, n \in \mathbb Z^\ast.$ It appears that this fact has not been observed before.  }

\textcolor{black}{}
 In the following theorem we prove that the expansion \eqref{4.152.0}, \eqref{4.152.1}, \eqref{4.152}
 of the  charge density fluctuation converges for   initial data in $\mathcal H.$  
\begin{theorem} \label{expansion}
Let $\rho(t,x)$ be the charge density fluctuation defined in \eqref{density}. Then, for any initial state, $\mathbf G_0 \in \mathcal H,$ that satisfies the Gauss law, the expansion, \eqref{4.152.0}, \eqref{4.152.1}, \eqref{4.152} converges strongly in the norm of $L^2(0,2\pi).$
\end{theorem}
\begin{proof}

\textcolor{black}{}
 We denote by $G(t,x,v)$ the  following quantity,
 \beq\label{4.154.0}
\begin{aligned}
G:= &\sum_{  n \in \mathbb Z^\ast, } \left( \mathbf G_0,\mathbf F^{(0)}_n\right) \, \mathbf F^{(0,1)}_n
+ \sum_{ n, m \in \mathbb Z^\ast} \, e^{-it \lambda_{n,m}}\,
  \left(\mathbf G_0,\mathbf Z_{n,m} \right)_{\mathcal H} \, \mathbf Z_{n,m}^{(1)},
 \end{aligned}
\ene
where  $  \mathbf Z_{n,m}^{(1)}$  is the first component of the eigenfunction  $\mathbf Z_{n,m}.$
Then,
\beq\label{4.154.1}
\rho(t,x)= - \int_{ \mathbb R^2}\, G(t,x,v)\, e^{\frac{-v^2}{4}} dv.
\ene
  Hence, since $G(t,x,v) \in \mathcal A,$ it follows from Fubini's theorem that  for a.e. $x\in (0,2\pi),$ $G(t,x,\cdot) \in L^2(\mathbb R^2),$ and as also $e^{\frac{-v^2}{4}}  \in L^2(\mathbb R^2),$  the integral in the right-hand side of \eqref{4.154.1} exists, and then,   the charge density fluctuation $\rho(t,x)$   is well defined. Furthermore, by the Cauchy-Schwarz inequality
$\rho(t,x) \in L^2(0, 2\pi)$. 
We denote,
\begin{align}\label{4.154}  
\rho_N(t,x):=  - \sum_{n \in \mathbb Z^\ast, |n| \leq N }\, \left( \mathbf G_0, \mathbf F^{(0}_n\right)_{\mathcal H}\, \int_{\mathbb R^2}\, \mathbf F^{(0,1)}(x,v) \, e^{\frac{-v^2}{4}}\, dv + \nonumber\\
\sum_{n,m \in \mathbb Z^\ast, |n|+|m|\leq N}\,  e^{-it \lambda_{n,m}} \,  \left(\mathbf G_0,\mathbf Z_{n,m} \right)_{\mathcal H} \, \rho_{n,m}(t,x).
\end{align}
We will prove that $\rho_N(t,x)$ converges to $\rho(t,x)$ in norm in $L^2(0,2\pi),$ i.e. that the series, \eqref{4.152.0}, \eqref{4.152.1} and  \eqref{4.152} converges strongly in $L^2(\mathbb R^2).$ 
 We designate,
\beq\label{4.155}
\begin{aligned}
G_{N}:= &\sum_{  n \in \mathbb Z^\ast, |n|\leq N} \left( \mathbf G_0,\mathbf F^{(0)}_n\right) \, \mathbf F^{(0,1)}_n
+ \sum_{ n, m \in \mathbb Z^\ast, |n|+|m|\leq N} \, e^{-it \lambda_{n,m}}\,
  \left(\mathbf G_0,\mathbf Z_{n,m} \right)_{\mathcal H} \, \mathbf Z_{n,m}^{(1)}.
 \end{aligned}
\ene
We have that
\beq\label{4.157}
\lim_{N \to \infty}\| G-G_N \|_{\mathcal A}=0.
\ene
Furthermore,
\beq\label{4.158}
-\int_{\mathbb R^2} \, G_N(t,x,v)\, e^{\frac{-v^2}{4}}dv= \rho_N(t,x).
\ene
Hence,
\beq\label{4.159}
\rho(t,x)- \rho_N(t,x)= - \int_{\mathbb R^2}\, (G(t,x,v)- G_N(t,x,v)) \, e^{\frac{-v^2}{4}}\,dv.
\ene
Finally, by \eqref{4.157}, \eqref{4.159}, and  the Cauchy- Schwarz inequality,
\beq\label{4.160}
\begin{aligned}
\int_0^{2\pi}\, |\rho(t,x)- \rho_N(t,x)|^2\,dx 
& = \int_0^{2\pi} 
\left| \int_{\mathbb R^2}\, (G(t,x,v)-  G_N(t,x,v))  
e^{\frac{-v^2}{4}}\, dv\right|^2\, dx  \\
 & \, \leq 2 \pi \, \int_0^{2\pi}\,  \int_{\mathbb R^2}\,   \left| (G(t,x,v) - G_N(t,x,v) \right|^2\, dx\, dv=
 2\pi \, \| G-G_N \|_{\mathcal A}^2\to 0,\, \text{\rm as }\,  N \to   \infty.
\end{aligned}
\ene
This completes the proof that the expansion \eqref{4.152.0}, \eqref{4.152.1},  \eqref{4.152}  converges strongly in the norm of $L^2(0,2\pi).$
\end{proof}

\begin{remark}{\rm
The eigenfunctions $\mathbf M^{(0)}_{0,j}, j \geq2,  \mathbf M^{(0)}_{n,j},  n \in \mathbb Z^\ast, j \geq 2,  \mathbf V_{m,j}, m \in \mathbb Z^\ast,  j \in \mathbb N^*, \mathbf W_{n,m,j},  n, m \in \mathbb Z^\ast, j \in \mathbb N^*,$ do not appear in the expansion \eqref{4.152.0}, \eqref{4.152.1}, \eqref{4.152}  of the charge density fluctuation. 
Still, as we mentioned in the introduction, 
these eigenfunctions are physically interesting because they show that there are plasma oscillations such that at each point the charge density fluctuation is zero and the electric field is also zero. \textcolor{black}{Some of them} are time independent. Note that since our eigenfunctions are orthonormal, these special plasma oscillation actually exist on their own, without the excitation  of the other modes. It appears that this fact has not been  observed previously in the literature.}
\end{remark}

\section{Operator theoretical proof of the Bernstein-Landau paradox}
\label{optheor}
\sss
We first study  the operator $\mathbf H_0$ that appears in the  formula for $\mathbf H$ that we gave in (\ref{4.6},\ref{4.7}, \ref{4.8}). Let us recall the representation of $\mathcal H$ as  the direct sum of the  $\mathcal H_n$ given in \eqref{4.106}. Using  Proposition~\ref{prop3.1} we see that the functions $(u_n,\alpha_n)^T$ in $\mathcal H_n$ can be written as
\modifsF{
\beq\label{o.1}
\begin{pmatrix} u_n(x,v) \\ \alpha_n\end{pmatrix}= 
\begin{pmatrix}  
\displaystyle \sum_{m \in \mathbb Z, j \in \mathbb N^*}\, u_{n,m,j}(x,v) \, (u_n, u_{n,m,j})_{\mathcal A} \\ 
\alpha_n
\end{pmatrix},
\ene
}
where for $ n=0, \alpha_n=0.$ Then, by Proposition ~\ref{prop3.1} 
\beq\label{o.2}
\mathbf H_0 \begin{pmatrix} u_n(x,v) \\ \alpha_n\end{pmatrix}= \mathbf H_{0,n}  \begin{pmatrix} u_n(x,v) \\ \alpha_n\end{pmatrix},
\ene
where by $\mathbf H_{0,n}$ we denote the operator in $\mathcal H_n$ given by,
$$
\mathbf H_{0,n}\,  \begin{pmatrix} u_n(x,v) \\ \alpha_n\end{pmatrix}:=  \sum_{m \in \mathbb Z, j \in \mathbb N^*}\, \begin{pmatrix} \lambda^{(0)}_m  u_{n,m,j}(x,v) \, (u_n, u_{n,m,j})_{\mathcal A} \\ 0
\end{pmatrix}, 
$$ 
with domain
$
D[\mathbf H_{0,n}]:=\ \{  ( u_n,\alpha_n) ^T :   \sum_{m \in \mathbb Z, j \in N} (\lambda^{(0)}_m)^2  |(u_n, u_{n,m,j})_{\mathcal A}|^2 <
 \infty$. 
Observe that $\mathbf H_{0,n}$ is the restriction of $\mathbf H_0$ to $\mathcal H_n,$ and that,
\beq\label{o.3}
\mathbf H_0= \oplus_{n \in \mathbb Z}\, \mathbf H_{0,n}.
\ene
Further,   the spectrum of $\mathbf H_{0,n}$ is pure point and it consists of the infinite multiplicity eigenvalue $\lambda^{(0)}_m, m \in \mathbb Z.$ Then, also the spectrum of $\mathbf H_0$ is pure point and it consists of the infinite multiplicity eigenvalues $\lambda^{(0)}_m, m \in \mathbb Z.$ Recall that the discrete spectrum of a selfadjoint operator consists of the isolated eigenvalues of finite multiplicity, and that the essential spectrum is the complement in the spectrum of the discrete spectrum. So, we have reached the conclusion that the spectrum of $\mathbf H_0$ coincides with the essential spectrum and it is given by the infinite multiplicity eigenvalues $\lambda^{(0)}_m, m \in \mathbb Z.$ 
Let us now consider the operator $\mathbf V$ that appears in \eqref{4.8}.  For  $e^{in x}\,( \tau(v),\alpha_n)^T\in \mathcal H_n,$
$$
\mathbf V e^{in x} \begin{pmatrix}  \tau(v)\\\alpha_n \end{pmatrix}= e^{inx} \, \begin{pmatrix} -i v_1 e^{\frac{-v^2}{4}}\, \alpha_n\\
i I^\ast \int_{\mathbb R^2}\, v_1\,  e^{\frac{-v^2}{4}}\, \tau(v)\, dv \end{pmatrix}.
$$
Then,   $\mathbf V$ sends $\mathcal H_n$  into  $\mathcal H_n,$ and that it acts in the same way in all the  $\mathcal H_n.$ Let us denote by $\mathbf V_n$ the restriction of $ \mathbf V$ to $\mathcal H_n.$ Then, we have,
\beq\label{0.4}
\mathbf V= \oplus_{n \in \mathbb Z}\, \mathbf V_n.
\ene
furthermore, by \eqref{4.6}, \eqref{o.3}, and \eqref{0.4},
\beq\label{o.5}
\mathbf H= \oplus \mathbf H_n,
\ene
where
$
\mathbf H_n=\mathbf H_{0,n}+ \mathbf V_n$. 
Further,  it follows from \eqref{0.4} that $\mathbf V_n$ is a rank two operator, hence, it is compact. Then, it is a consequence of the Weyl theorem for the invariance of the essential spectrum, see Theorem 3, in page 207 of \cite{bs}, that the essential spectrum of $\mathbf H_n, n \in \mathbb Z$ is  given by the infinite multiplicity eigenvalues $\lambda^{(0)}_m, m \in \mathbb Z.$  Hence, by \eqref{o.5} the essential spectrum of $\mathbf H$ is  given by the infinite multiplicity eigenvalues $\lambda^{(0)}_m, m \in \mathbb Z. $ However, since the complement of the essential spectrum is discrete, we have that the spectrum of $\mathbf H$ consists of the  infinite multiplicity eigenvalues $\lambda^{(0)}_m, m \in \mathbb Z,$ and of a set of 
isolated eigenvalues of finite multiplicity that can only accumulate at the essential spectrum and at $\pm \infty.$ We know from the results of Section ~\ref{secvlasamp} that these eigenvalues are the $\lambda_{n,m}, n, m\in \mathbb Z^\ast,$ and that they are of multiplicity one. However, the operator  theoretical argument does not tell us that.
However, it tells us that the spectrum of $\mathbf H$ is pure point and that $\mathbf H$ has a complete orthonormal set of eigenfunctions. This implies that the Bernstein -Landau paradox exists.  Let us elaborate on this point. As we mentioned in the introduction, it was shown by \cite{despres1}, \cite{despres2} that the Landau damping  can be 
characterized as  the fact that when the magnetic field is zero $e^{-it \mathbf H}$ goes weakly to zero as $t \to \pm \infty.$ Let us prove that when the magnetic field is non zero this is not true. We prove this fact using only the operator theoretical results of this section, i.e. without using the detailed calculations of Section~ \ref{secvlasamp}.
\modifsF{Let us denote by $\gamma_j$, $j= 1,\dots$, 
the eigenvalues of $\mathbf H,$ repeated according to their multiplicity, and let $\mathbf X_j$, $j=1,\dots$ be a  complete set of orthonormal  eigenfunctions, where the eigenfunction $\mathbf X_j,$ is associated with the eigenvalue, $\gamma_j, j=1,\dots$}. We know explicitly from Section ~\ref{secvlasamp}  the eigenvalues and a orthonormal basis of eigenvectors, but we do not need this information here.  Suppose that     $e^{-it \mathbf H}$ goes weakly to zero as $t \to \pm \infty.$ Then, for any $\mathbf X, \mathbf Y \in \mathcal H,$
\beq\label{o.6}
\lim_{t \to \pm \infty}\, \left(  e^{-it \mathbf H}\mathbf X, \mathbf Y \right)_{\mathcal H}=0.
\ene
Let us prove that there is no non trivial $\mathbf X \in \mathcal H$ such that \eqref{o.6} holds for  all $\mathbf Y \in \mathcal H.$ We have that,
$$
\left(  e^{-it \mathbf H}\mathbf X, \mathbf Y \right)_{\mathcal H}= \sum_{l=1}^\infty e^{-it \gamma_l} (\mathbf X, \mathbf X_l)_\mathcal H\, (\mathbf X_l,\mathbf Y)_\mathcal H.
$$
However, let us take $\mathbf Y= \mathbf X_j, j= 1,\dots.$ Then,
$ \displaystyle
\lim_{t \to \pm \infty}\, \left(  e^{-it \mathbf H}\mathbf X, \mathbf Y_j \right)_{\mathcal H}=   \lim_{t \to \pm \infty}\,    e^{-it \gamma_j} \,(\mathbf X, \mathbf X_j)_\mathcal H$,
$ j=1,\dots$, 
is a non-zero constant if $\gamma_j=0,$ and it is oscillatory if $\gamma_j\neq 0,$ unless $(\mathbf X, \mathbf X_j)_\mathcal H=0,  j=1,\dots.$ However, if 
 $(\mathbf X, \mathbf X_j)_\mathcal H=0, j=1,\dots, $ then, $\mathbf X=0.$ It follows that \eqref{o.6} only holds for $ \mathbf X=0.$ 
%

\section{Numerical results}
\label{numerical} \sss
The objective of this section is to illustrate the numerical behavior of the eigenfunctions constructed previously. More precisely, we will construct a numerical scheme that approximates the solution of the  \textcolor{black}{magnetized} Vlasov-Amp\`ere  system initialized with an eigenfunction and compare this numerical solution with the theoretical dynamics of the system. The numerical results below show that the difference between the theoretical and numerical solutions is small, confirming the theoretical analysis.
Furthermore, we will use the eigenfunctions to initialize a code solving the non-linear \textcolor{black}{magnetized} Vlasov-Poisson system showing how we can approximate the solution of the non-linear system with our linear theory.
Finally, using the same non-linear code, we will illustrate the Bernstein-Landau paradox, as in the spirit of \cite{B2010,VVM}, by initializing with a standard test function traditionally used to highlight Landau damping and show how the damping is lost when we add a constant magnetic field.

\subsection{Computing the eigenvalues}
As in \eqref{4.83},  we consider an eigenfunction  
 \beq\left(\begin{array}{c}\label{eee.1}
 	w_{n,m} \\
 	F_n
 \end{array}
 \right),
 \ene
  of the operator $\mathbf H$ 
 associated to the Fourier mode $n \ne 0$ and the eigenvalue $\lambda_{n,-m}= - \lambda_{n,m}$ where $w_{n,m}$ and $F_n$ are given by
 
 \begin{equation}\label{nn.1}
 w_{n,m}=e^{in (x -\frac{v_2}{\omega_c})}e^{-\frac{r^2}4}  \sum_{p\in \mathbb Z^*} \frac{p \omega_c}{p\omega_{\rm c} +\lambda_{n,m}}e^{ i p \varphi} J_p\left( \frac{n r}{\omega_c}\right) \mbox{and} \, F_n=-ine^{inx}.
 \end{equation}
 Furthermore,  $\lambda_{n,m}$ is one of the roots of a secular equation \eqref{4.45},
\modifsF{which could be written as
 $$
 \alpha(\lambda)=0,
 $$
 where the {\it secular function} $\alpha(\lambda)$ is given by
 \begin{equation} \label{sec}
 \alpha(\lambda)=-1-\frac {2\pi }{n^2}
 \sum_{m \in \mathbb Z^*} 
 \frac{m\omega_c }{m\omega_c +\lambda} a_{n,m}.
 \end{equation}
In  \eqref{sec} $a_{n,m}$ is defined by \eqref{4.48}. The secular function $\alpha(\lambda)$ is a convergent series with poles at the multiples of the cyclotron frequency $\omega_{\rm c}$.
Note that the function $\alpha$ in \eqref{sec} and the function $g$ in \eqref{4.49} are linked by the relation
\begin{equation} \label{linkalphag}
\alpha(\lambda) = - 1 - \frac{1}{n^2} g(\lambda).
\end{equation}
}

   \begin{figure}[H] 
 	\centering
 	\includegraphics[width=0.8\linewidth]{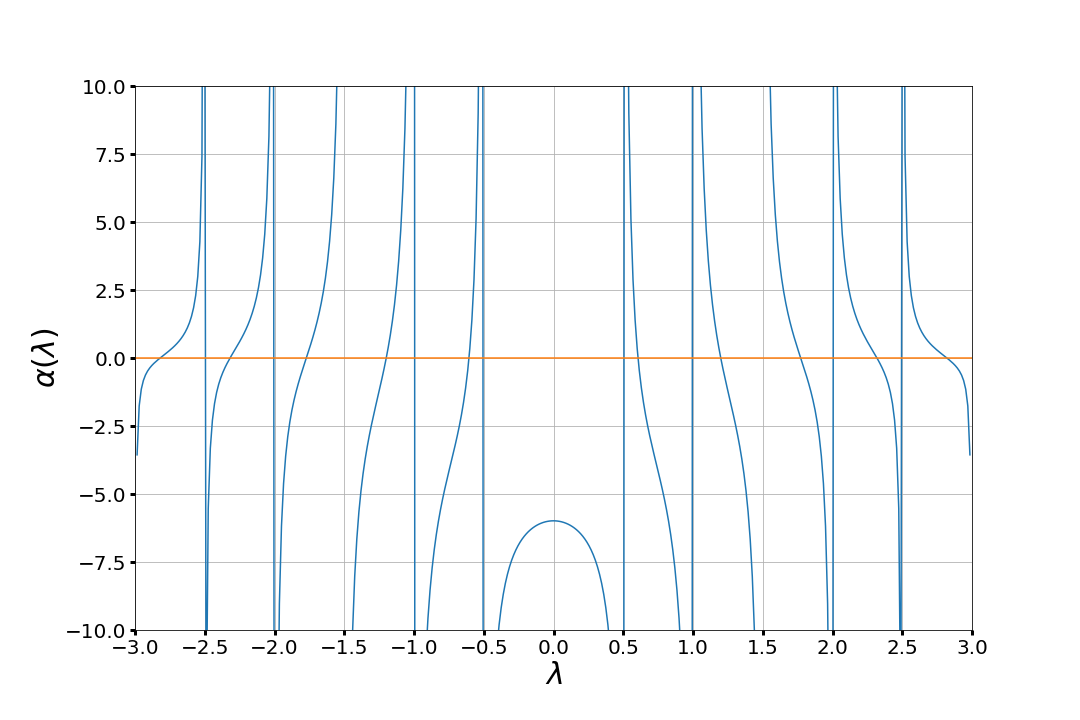}
 	\caption{Secular function for $\omega_c=0.5$ and $n=1$}
	\label{fig:tt1}
 \end{figure}

 The plot in Figure \ref{fig:tt1}   illustrates the properties of $\alpha$ \modifsF{(deduced from Lemma \ref{lemm4.3} and relation \eqref{linkalphag})}, most notably that there is unique root (hence an eigenvalue for $\mathbf H$) in $(m\omega_ {\rm c}, (m+1)\omega_{\rm c})$ for $m\geq 1,$ and $((m-1)\omega_{\rm c}, m\omega_{\rm c})$ for $m\leq -1$.
 With a standard numerical method (dichotomy or Newton), we can determine the roots of $\alpha$. For example, with $(n,m)=(1,2)$, we find $\lambda_{1,2} \approx 1.19928.$ This eigenvalue $\lambda_{1,2}$ will be used in all the following numerical tests.
 
\subsection{Solving the linear \textcolor{black}{magnetized} Vlasov-Amp\`ere system with a Semi-Lagrangian scheme with splitting}

To approximate the linear system (\ref{vlasamp}) or (\ref{4.1}-\ref{4.2}), we use a semi-Lagrangian scheme \cite{CK76,SRBG}, which is a classical method to approximate transport equations of the form $\partial_t f+E(x,t)\partial_x f=0$, coupled with a splitting procedure. A splitting procedure corresponds to approximating the solution of $\partial_t f+\mathcal{(A+B)} f=0$ by solving $\partial_t f+\mathcal{A} f=0$ and $\partial_t f+\mathcal{B} f=0$ one after the other.

Hence, the  \textcolor{black}{magnetized} Vlasov-Amp\`ere system is split so as to only solve transport equations with constant advection terms.
$$
\partial_t\left(
\begin{array}{c}
u \\
F
\end{array}
\right)+\mathcal{(A+B+C+D)}\left(\begin{array}{c}
u \\
F
\end{array}
\right)=0,
$$
with 
$$
\begin{aligned}
	\mathcal{A}=\left(\begin{array}{c}
	v_1\partial_x \\
	0
	\end{array}
	\right), &\ \   \mathcal{B}=\left(\begin{array}{c}
	Fv_1e^{- \frac{v_1^2+v_2^2}{4} }\\
	1^* \int u  e^{- \frac{v_1^2+v_2^2}{4} } v_1dv_1dv_2
	\end{array}
	\right), 
	    &\mathcal{C}=\left(\begin{array}{c}
	-\omega_c v_2\partial_{v_1} \\
	0
	\end{array}
	\right), & \ \   \mathcal{D}=\left(\begin{array}{c}
	\omega_c v_1\partial_{v_2} \\
	0
	\end{array}
	\right).
\end{aligned}
$$    
The algorithm used to solve the linearized \textcolor{black}{magnetized} Vlasov-Amp\`ere system can thus be summarized as follows
\begin{enumerate}
    \item \textbf{Initialization}
    $\mathbf U_{ini}=\left(
\begin{array}{c}
w_{n,m} \\
F_n
\end{array}
\right)$ given in \eqref{eee.1}.
    
    \item \textbf{Going from $t_n$ to $t_{n+1}$} 
    
    Assume we know $\mathbf U_n$, the approximation of 
    $
   \mathbf  U= \begin{pmatrix} u \\ F\end{pmatrix}$
     at time $t_n$.
    \begin{itemize}
        \item We compute $\mathbf U^*$ by solving $\partial_t\mathbf U + \mathcal{A}\mathbf U=0$ with a semi-Lagrangian scheme during one time step $\Delta t$ with initial condition $\mathbf U^n$.
        
        \item We compute $\hat{\mathbf U}$ by solving $\partial_t\mathbf  U + \mathcal{B}\mathbf U=0$ with a Runge-Kutta 2 scheme during one time step $\Delta t$ with initial condition $\mathbf U^*$.
    
        \item We compute $\mathbf U^{**}$ by solving $\partial_t \mathbf U + \mathcal{C}\mathbf U=0$ with a semi-Lagrangian scheme during one time step $\Delta t$ with initial condition $\hat{\mathbf U}$.
        
        \item We compute $\mathbf U^{n+1}$ by solving $\partial_t\,\mathbf  U + \mathcal{D}\mathbf U=0$ with a semi-Lagrangian scheme during one time step $\Delta t$ with initial condition $\mathbf U^{**}$.
    \end{itemize}
\end{enumerate}

\subsection{Results for the \textcolor{black}{magnetized} Vlasov-Amp\`ere system}
\label{subsec8.3}
The solution of the  \textcolor{black}{magnetized} Vlasov-Amp\`ere system initialized with an eigenfunctions $
\mathbf U_{\rm ini}= \left(\begin{array}{c}
 	w_{n,m} \\
 	F_n
 \end{array}
 \right)$ 
as in \eqref{eee.1}
 is simply given by
\modifsF{
\begin{equation}
\label{evolt}
	\mathbf U(t)  = e^{ i \lambda_{n,m} t} \, \mathbf U_{\rm ini}.
\end{equation}
}
Recall that \eqref{eee.1} is an eigenfunction of $\mathbf H$ with eigenvalue $\lambda_{n,-m}= - \lambda_{n,m}.$
In the following results, we have taken $(n,m)=(1,2)$,   $\omega_{c}=0.5$, $N_x=33$ (number of points of discretization in position), $N_{v_1}=N_{v_2}=63$ (number of points of discretization in both velocity variables), $L_{v_1}=L_{v_2}=10$ (numerical truncation in both velocity variables) and, most importantly, $T_{\rm f}=\frac{\pi}{ \lambda_{1,2}}$. This means that 
$	\mathbf U(T_{\rm f})=\exp\left(i\frac{\pi}{2}\right) \mathbf U_{\rm ini}=i \mathbf U_{\rm ini}$, 
and then,  the solution of the system at $t=T_{\rm f}$ corresponds to the initial condition where the real and imaginary parts have been exchanged (up to a sign).

	\begin{figure}[H]	
	\centering
	\includegraphics[width=0.95\linewidth,height=9cm]{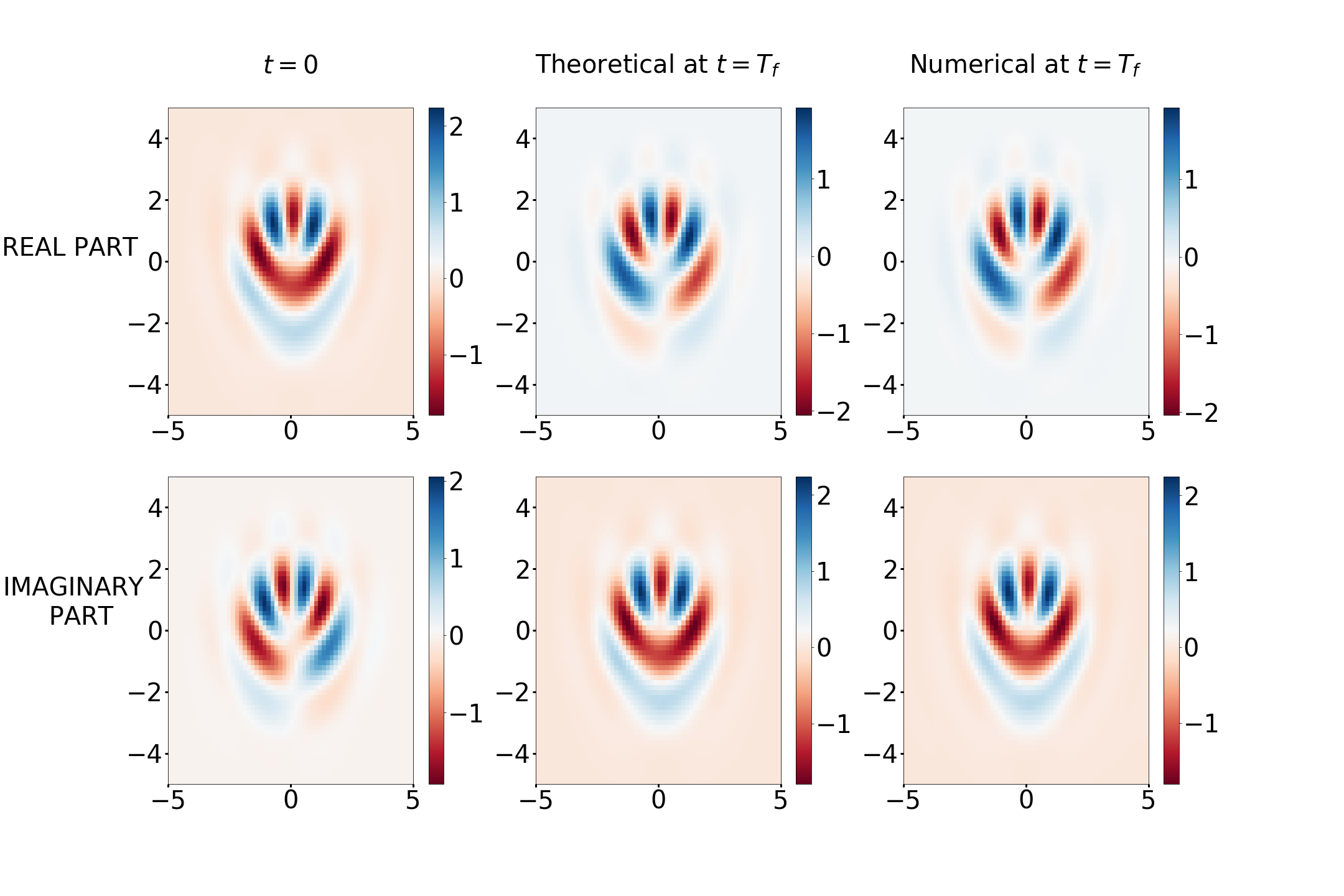}   
	\caption{Real and imaginary parts of \modifsF{the first component of $\mathbf U(t)$ given by \eqref{evolt}  } in $v_1-v_2$ plane for $x=0$.}
	\end{figure}
	
	\begin{figure}[H]
	\centering
	\includegraphics[width=0.95\linewidth,height=9cm]{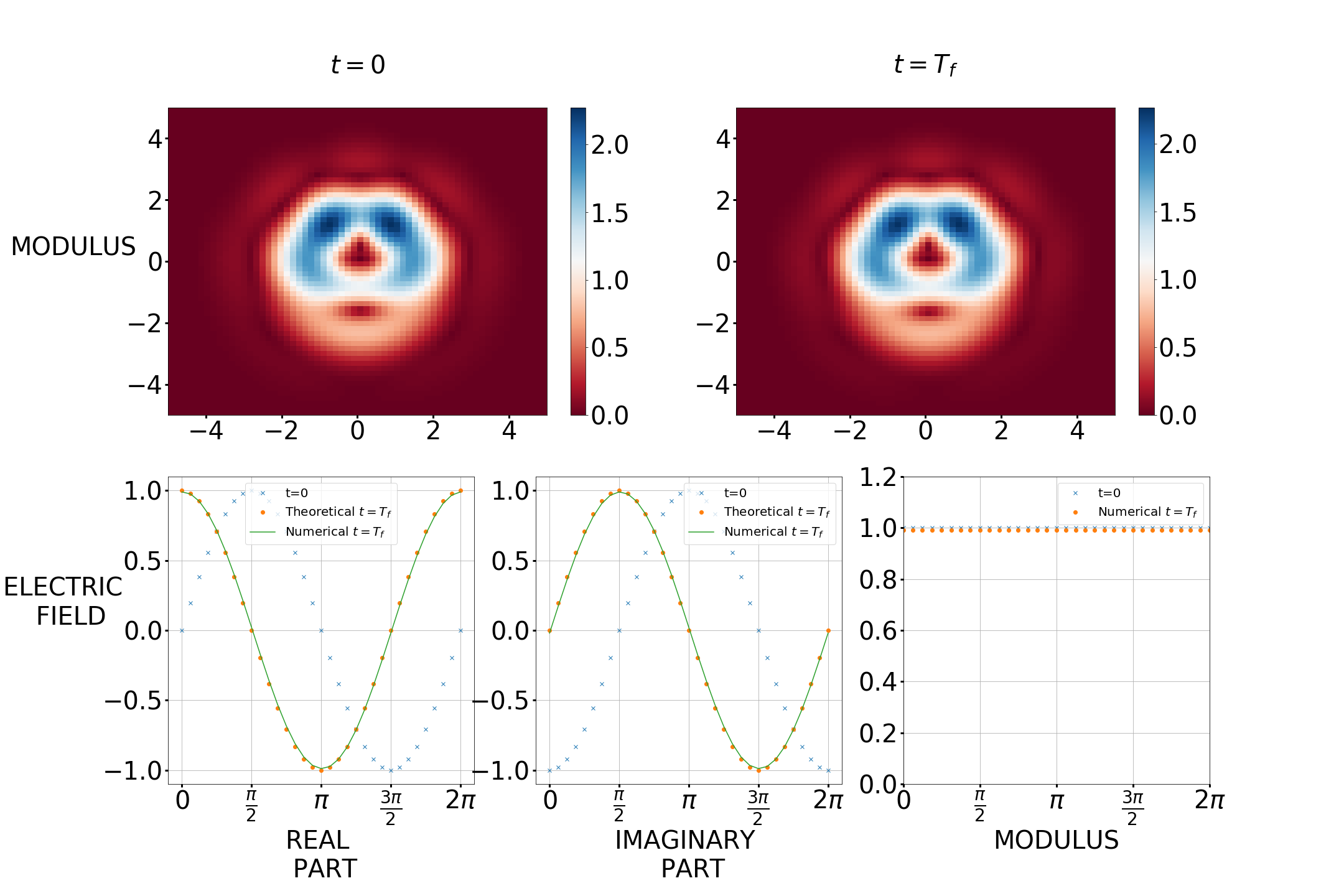}   
	\caption{Modulus of  \modifsF{the first component of $\mathbf U(t)$ given by \eqref{evolt}  }  in $v_1-v_2$ plane for $x=0$, and real and imaginary parts of $F$.}
	\end{figure}

The figures show that the solution of the system behaves according to the theory.


\subsection{Results for the non-linear \textcolor{black}{magnetized} Vlasov-Poisson system}

We  now look at how the solution of the non-linear \textcolor{black}{magnetized} Vlasov-Poisson system \eqref{bd:eq:4}   behaves when initialized with an eigenfunction of the Hamiltonian $\mathbf H$ of the  \textcolor{black}{ magnetized Vlasov-Amp\`ere} system. The idea is that for a certain time, the solution for the non-linear \textcolor{black}{magnetized} Vlasov-Poisson  system follows the same dynamics as the solution for the linearized  \textcolor{black}{magnetized} Vlasov-Poisson system. We consider the \textcolor{black}{magnetized} Vlasov-Poisson system because it is more convenient for numerical purposes.  Recall that the  linearized \textcolor{black}{magnetized} Vlasov-Poisson and the  \textcolor{black}{magnetized} Vlasov-Amp\`ere systems are equivalent. Furthermore, the articles  \cite{bernstein,B2010,suku,VVM} have studied the Bernstein-Landau paradox using the \textcolor{black}{magnetized}  Vlasov-Poisson system. We use almost the same numerical scheme as in the previous subsection to approximate the solution of the system.

The Vlasov equation, namely the first equation in \eqref{bd:eq:4}, in the non-linear  \textcolor{black}{magnetized} Vlasov-Poisson system is split so as to only solve transport equations with constant advection terms,
$$  
\partial_t
u +\mathcal{(A+B+C)}u=0,
$$
with 
$\mathcal{A}=v_1\partial_x$, $ \mathcal{B}=-\left(E+\omega_c v_2\right)\partial_{v_1}$ and
$\mathcal{C}=\omega_c v_1\partial_{v_2}$.
To update the electric field, the strategy adopted is the same as in \cite{CK76} where the Poisson equation is solved at each time step.
On this numerical computation   we consider real valued solutions $f,E.$  

Let us denote by $u$  the perturbation of the charge density function, $f,$  and by $F$ be the perturbation of the electric field, $E.$   The functions $u,F$ solve the linearized   \textcolor{black}{magnetized}   Vlasov-Poisson system \eqref{bd:eq:7}. Recall that we proven  in Section~\ref{secvlasamp} that the linearized \textcolor{black}{magnetized} Vlason-Poisson and   \textcolor{black}{magnetized} Vlasov-Amp\`ere   systems are equivalent. Then, we can use  the real part of  \ref{evolt}  to  write the expression of $u, F$ when initializing with, $u_{\rm ini}, F_{\rm ini},$ with   $u_{\rm ini}=\operatorname{Re} (w_{n,m}),$  $F_{\rm ini}=\operatorname{Re} (F_{n}).$ Recall that  $w_{n,m}$, and  $F_{n}$ are defined in \eqref{nn.1}. Then, we have,
\modifsF{
\begin{equation}\label{evoltnonlin}
\begin{pmatrix}
u(t) \\ F(t)
\end{pmatrix}
=
\operatorname{Re} (\mathbf U(t))
= \begin{pmatrix}
\cos(\lambda_m t)\operatorname{Re} (w_{n,m})-\sin(\lambda_m t)\operatorname{Im} (w_{n,m})\\
\cos(\lambda_m t)\operatorname{Re} (F_{n})-\sin(\lambda_m t)\operatorname{Im} (F_{n})
\end{pmatrix}
\end{equation}
where $\mathbf U(t)$ is given by \eqref{evolt}.
}
The objective of this subsection is to show that we can approximate the solution of the non-linear system using \eqref{evoltnonlin}, which means that the solutions of both linear and non-linear systems are close to each other for a certain time.

The algorithm used to solve the non-linear \textcolor{black}{magnetized} Vlasov-Poisson system can be summarized as follows:

\begin{enumerate}
    \item \textbf{Initialization}
        $f_{ini}= f_0+\varepsilon \sqrt{f_0}
\operatorname{Re}(w_{n,m})$ and $
E_{\rm ini}=\varepsilon \operatorname{Re}(F_n)$  are given, where  $\varepsilon$ is a scalar which controls the amplitude of the perturbation. We take $\varepsilon=0.1$.

    \item \textbf{Going from $t_n$ to $t_{n+1}$} 
    
    Assume we know $f_n$ and $E_n$, the approximations of $f$ and $E$ at time $t_n$.
    \begin{itemize}
        \item We compute $f^*$ by solving $\partial_t f  + v_1\partial_x f=0$ with a semi-Lagrangian scheme during one time step $\Delta t$ with initial condition $f_n$.
        
        \item We compute $E_{n+1}$ by solving the Poisson equation with $f^*$.
    
        \item We compute $\hat{f}$ by solving $\partial_t f -(E_{n+1}+\omega_cv_2)\partial_{v_1}f=0$ with a semi-Lagrangian scheme during one time step $\Delta t$ with initial condition $f^{*}$.
        
        \item We compute $f^{n+1}$ by solving $\partial_t f +\omega_cv_1\partial_{v_2}f=0$ with a semi-Lagrangian scheme during one time step $\Delta t$ with initial condition $\hat{f}$.
    \end{itemize}
\end{enumerate}

As in Subsection~ \ref{subsec8.3}    we take $(n,m)=(1,2)$,   $\omega_{c}=0.5$, $N_x=33$ (number of points of discretization in position), $N_{v_1}=N_{v_2}=63$ (number of points of discretization in both velocity variables), $L_{v_1}=L_{v_2}=10$ (numerical truncation in both velocity variables) and, $T_{\rm f}=\frac{\pi}{ \lambda_{1,2}}$. 
In the following figures, we are comparing respectively the theoretical  perturbations, $u,F,$  that are given by \eqref{evoltnonlin}, and the numerical perturbations, 
$$
u^n=\frac{f^n- f_0}{\varepsilon \sqrt{f_0}}, \text{\rm and}\, F^n=\frac{E^n}{\varepsilon},
$$
 where $f^n$ and $E^n$ are given by the above algorithm.

\begin{figure}[H]	
	\centering
	\includegraphics[width=0.95\linewidth,height=9cm]{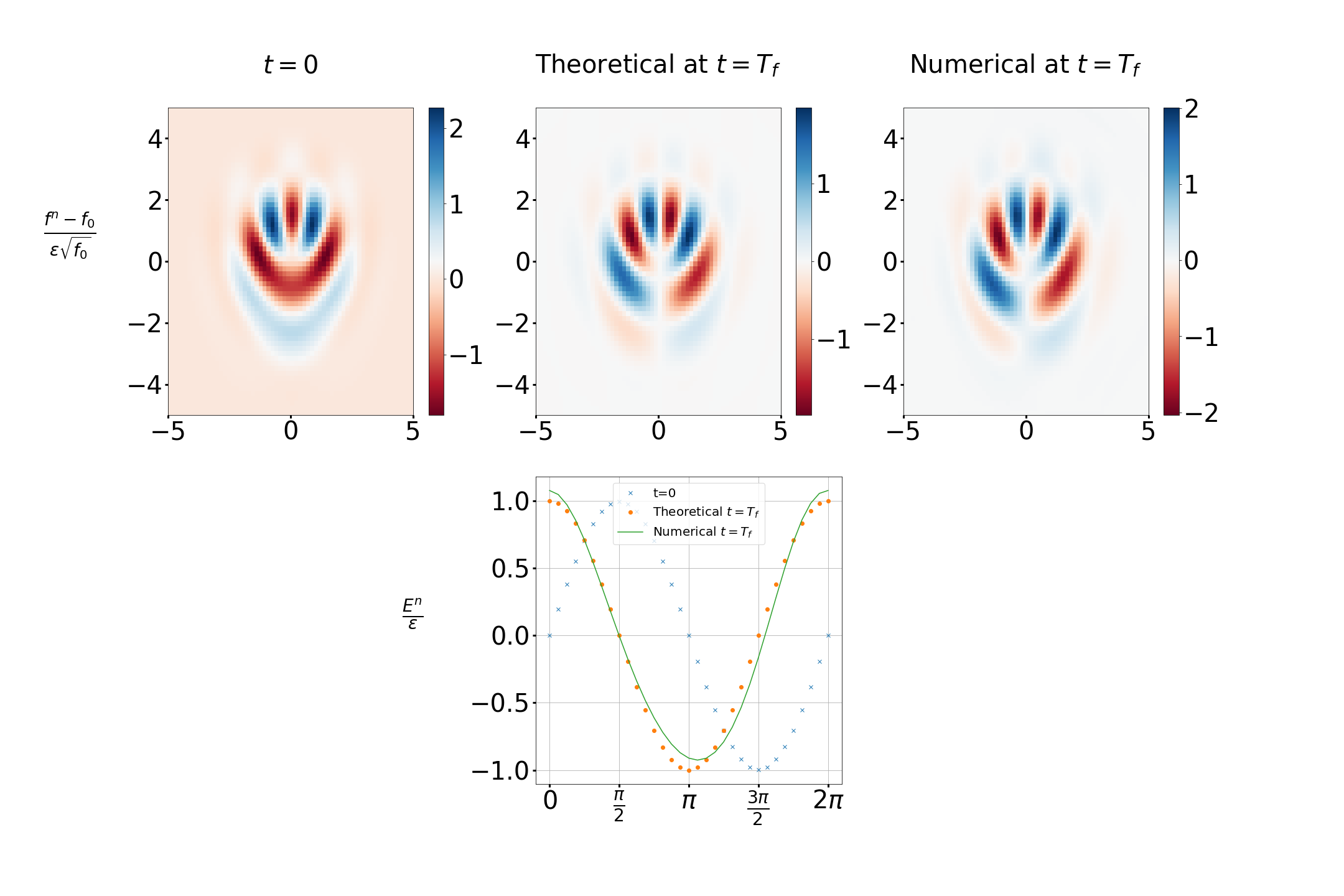}   
	\caption{$u$ in $v_1-v_2$ plane for $x=0$ and electric field $F$.}
\end{figure}

The figures show that we can approximate the solution of the non-linear \textcolor{black}{magnetized} Vlasov-Poisson system using  solutions of the linear \textcolor{black}{magnetized} Vlasov-Poisson system, initialized with   the eigenfunctions of the Hamiltonian, $\mathbf H$  of  the \textcolor{black}{magnetized} Vlasov-Amp\`ere system.

\subsection{The Bernstein-Landau paradox}
\label{paradox.aa}

In this subsection we numerically  illustrate the Bernstein-Landau paradox, and we compare it with the Landau damping,  using the above algorithm (similarly to \cite{B2010}).
In order to compare the numerical solutions to the non-linear Vlasov-Poisson system with the approximate analytical solution found in \cite{sonnen08}  in the case
$ \omega_{\rm c}=0,$ we take below the charge of the ions equal to one. With this convention the non-linear Vlasov-Poisson system  is written as, 

\begin{equation} \label{nlpoisson}
\left\{
\begin{aligned} 
&\partial_t f + v_1 \partial_x f -  E \partial_{v_1} f + \omega_c \left( -v_2 \partial_{v_1} +v_1 \partial_{v_2}  \right) f= 0. \\
&\partial_x  E(t,x)= 1 - \int_{\mathbb R^2}\, f dv.
\end{aligned} 
\right.
\end{equation}
Furthermore, also with the purpose of comparing with the approximate analytical solution of \cite{sonnen08}, we initialize with the density function $f_{LD}$ given by,

\begin{equation}
\label{iniLD}
	f_{LD}(x,v_1,v_2)=\frac{1}{2\pi}\left(1+\varepsilon \cos{kx}\right)\, \ds e^{\frac{-v^2}{2}},  \quad \varepsilon=0.001, k=0.4.
\end{equation}
In this simulation the position interval is $[0,\frac{2\pi}{k}]$, since we keep periodic solutions. To introduce the  approximate analytical solution  of \cite{sonnen08} let us consider the Vlasov-Poisson system \eqref{nlpoisson} with   $ \omega_{\rm c}=0,$
\begin{equation} \label{dd.0}
\left\{
\begin{aligned} 
&\partial_t f + v_1 \partial_x f -  E \partial_{v_1} f= 0, \\
&\partial_x  E(t,x)= 1 - \int_{\mathbb R^2}\, f dv,
\end{aligned} 
\right.
\end{equation}
 and initialized with \eqref{iniLD}.

 Let us look for a solution of the form,
\beq\label{dd.1}
f(t,x,v)= f_1(t,x,v_1)\, \frac{1}{\sqrt{2\pi}}\, e^{\frac{-v_2^2}{2}}.
\ene
Then, $f(t,x,v)$ satisfies the \eqref{dd.0} and it is initialized with \eqref{iniLD} if and only if $f_1(t,x,v_1)$ is a solution of the following Vlasov-Poisson system in one dimension in space and velocity,
\modifsF{
\begin{equation} \label{dd.2}
\left\{
\begin{aligned}
&\partial_t f_1 + v_1 \partial_x f_1 -  E_1 \partial_{v_1} f_1 = 0, \\
&\partial_x  E_1(t,x)= 1 - \int_{\mathbb R}\, f_1 dv_1,
\end{aligned}
\right.
\end{equation}
}
initialized with,
\beq\label{dd.3}
f_1(0,x,v_1)= \frac{1}{\sqrt{2\pi}}\left(1+\varepsilon \cos{kx}\right)\ds e^{\frac{-v_1^2}{2}},  \quad \varepsilon=0.001, k=0.4.
\ene
Furthermore, note that
\beq\label{dd.4}
E(t,x)= E_1(t,x).
\ene
Then, we  can compute an approximate $E(t,x)$ using the approximate solution to \eqref{dd.2}, \eqref{dd.3} given in page 58 of \cite{sonnen08}. Namely,
\begin{equation}
\label{theoNL}
E(x,t) \approx 4\varepsilon\times0, 424666\exp(-0,0661t)\sin(0,4x)\cos(1,2850t-0,3357725).
\end{equation}
We have taken the values given in the second line of the table  in page 58 of  \cite{sonnen08}. This  approximate solution is a good approximation to the exact solution for large times. 
Further, \eqref{theoNL}  is a classical test function to highlight Landau damping, more precisely the damping of the electric energy. In the figures below we report \eqref{theoNL}  in the black curves.  Moreover, the figure below illustrates how when $\omega_c\ne 0$, the damping is replaced by a recurrence phenomenon of period $T_c=\frac{2\pi}{\omega_c}$, which follows the behaviour observed in \cite{bald:rol,suku}. We take  $\omega_{\rm c}=0.1,$  and as in Subsection ~\ref{subsec8.3}, we use, $N_x=33$ (number of points of discretization in position), $N_{v_1}=N_{v_2}=63$ (number of points of discretization in both velocity variables), $L_{v_1}=L_{v_2}=10$ (numerical truncation in both velocity variables). 
\begin{figure}[H]
	\begin{minipage}[b]{0.5\linewidth}
		\centering
		\includegraphics[width=8cm,height=5cm]{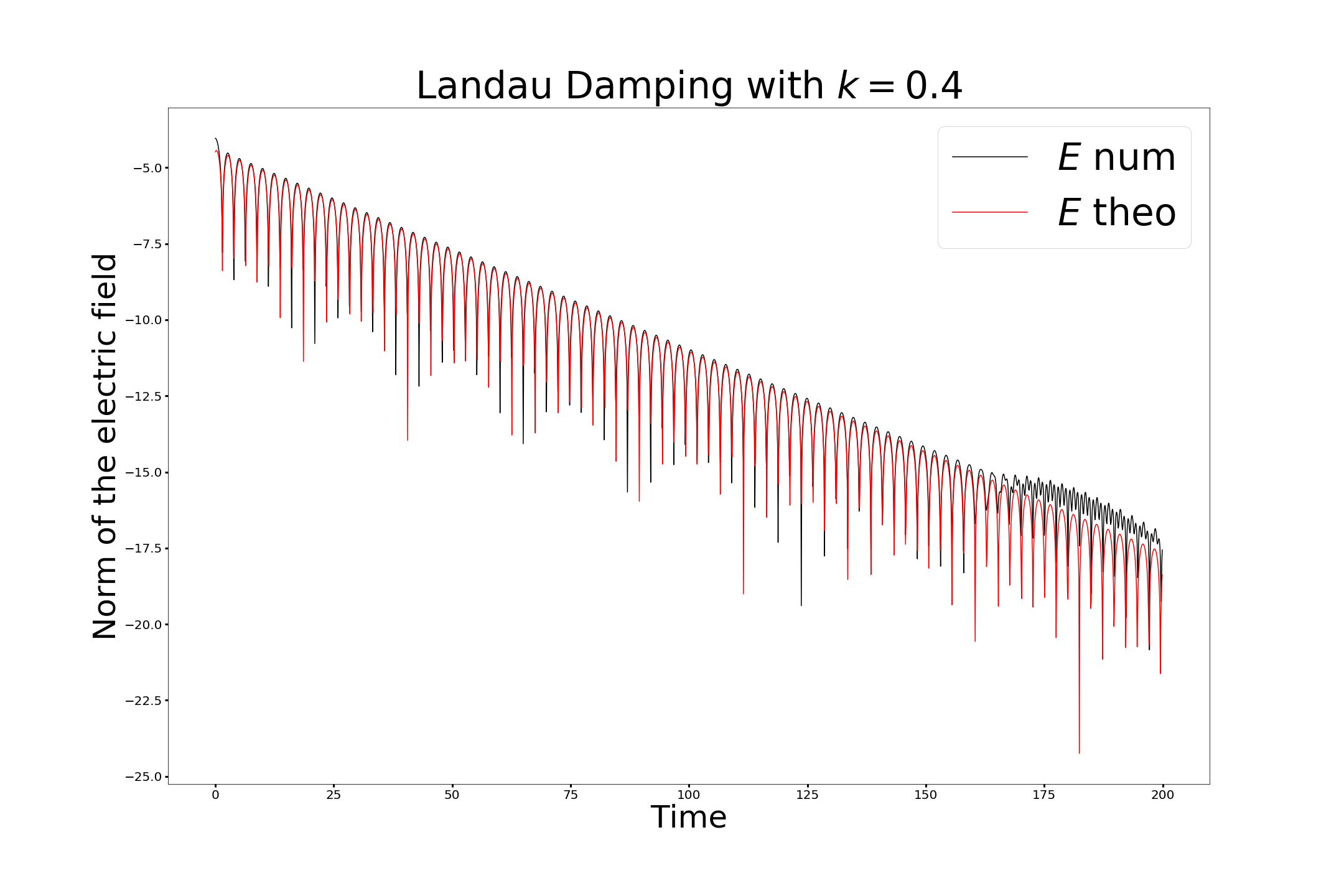}     
	\end{minipage}
	\begin{minipage}[b]{0.5\linewidth}
		\centering
		\includegraphics[width=8cm,height=5cm]{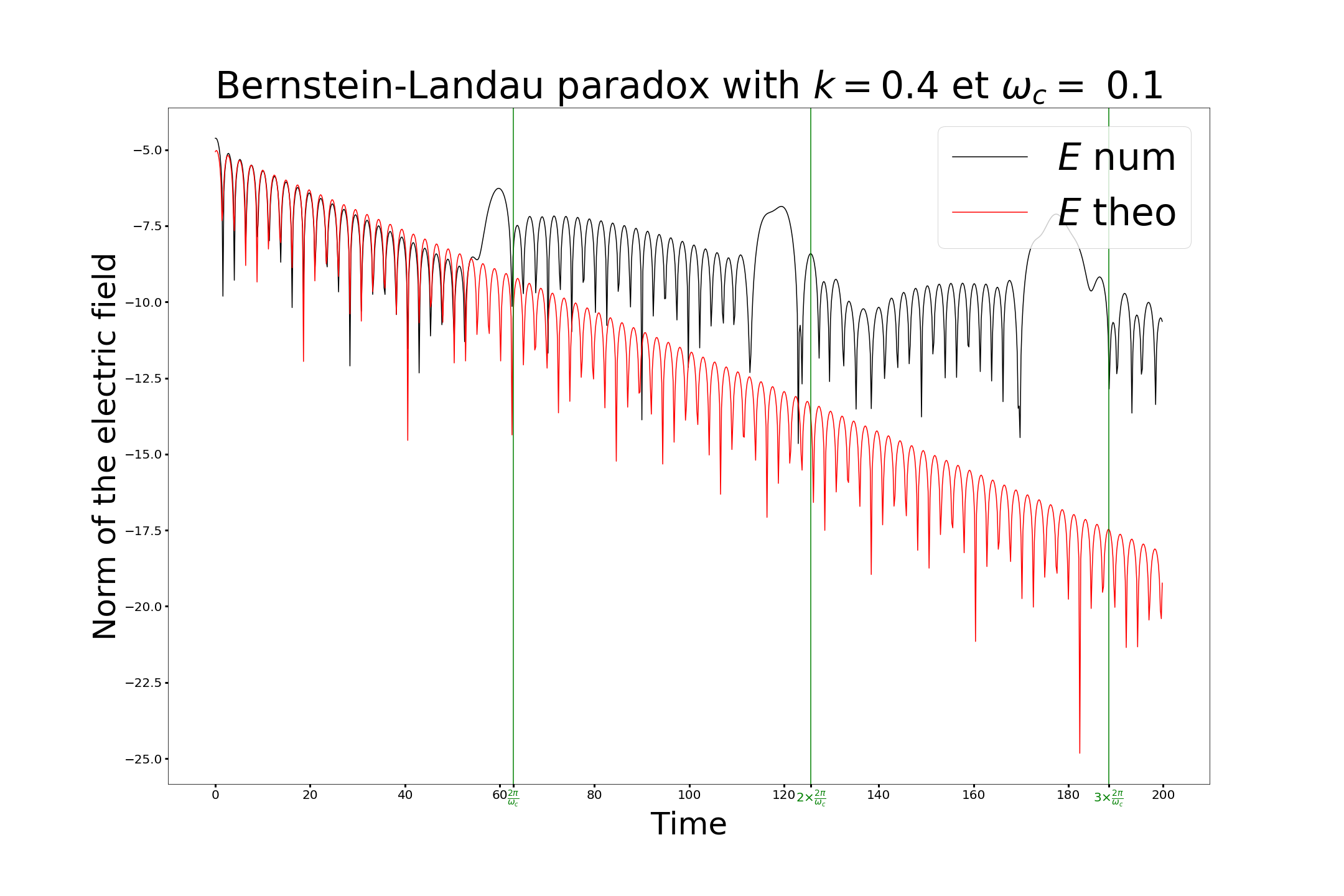}     
	\end{minipage}
	\caption{Damped and undamped electric field}
	\label{LDLB}
\end{figure}


The recurrence visible on the right-hand side figure, i.e. the Bernstein paradox,  is a fully "physical" phenomenon originating from the non-zero magnetic field and is to be distinguished from the recurrence in semi-Lagrangian schemes studied in \cite{MNP}, which deals with a purely numerical phenomenon.  \textcolor{black}{ Let us show that this  recurrence is a consequence of our  series based on the eigenvectors expansion in the regime of non zero magnetic field. For this purpose, we take the charge of the ions equal to $2 \pi,$  and solutions with period $2 \pi,$ to be able to use our results of the previous sections.
We  consider the initial data.   
\beq\label{initial.8.5}
\mathbf G_0:= \left( u_0(x,v), F_0(x)\right)=\left( e^{\frac{-v^2}{4}}\, \cos{lx}, -\frac{2\pi }{l} \sin{lx}  \right), l \in \mathcal Z^\ast,
\ene
which satisfies the Gauss law. To compute the electric field with the expansion of the solution to the magnetic Vlasov-Amp\`ere system given in \eqref{4.145}, \eqref{4.146.b}, \eqref{4.151} we only need to consider the eigenfunctions with non zero electric field, namely, $\mathbf  F^{(0}_n$ and $\mathbf Z_{n,m}.$ Recalling that the electric field is the second component of \eqref{4.145} we obtain,
\
\begin{align}\label{elc.8.5} \nonumber
F(x,t)&=   \sum_{n \in \mathbb Z^\ast}\, \left( \mathbf G_0, \mathbf F^{(0}_n\right)\,  \mathbf F_n^{(0,2)}(x) +   \sum_{n,m \in \mathbb Z^\ast}\,  e^{-i \lambda_{n,m} t}  \left(\mathbf G_0,\mathbf Z_{n,m} \right)_{\mathcal H} \, \mathbf Z^{(2)}_{n,m}(x)= \\  
& \left( \mathbf G_0, \mathbf F^{(0}_l\right)\,  \mathbf F_l^{(0,2)}(x) +   \sum_{m \in \mathbb Z^\ast}\,  e^{-i \lambda_{l,m} t}  \left(\mathbf G_0,\mathbf Z_{l,m} \right)_{\mathcal H} \, \mathbf Z^{(2)}_{l,m}(x),
 \end{align}
 where we denote by    $\mathbf F^{(0,2)}_n,$ respectively,  $\mathbf Z^{(2)}_{n,m}$ the second component of $\mathbf F^{(0)}_n,$  and of $\mathbf Z_{n,m}.$ Then, using \eqref{1.2} with $E_0=0,$ and \eqref{elc.8.5} we get the approximate formula for the electric field,
\begin{equation}
\label{theoNL:2}
E(x,t)\approx \varepsilon \Re F(t) = \varepsilon  \Re \left( \mathbf G_0, \mathbf F^{(0}_l\right)\,  \mathbf F_l^{(0,2)}(x) +   \varepsilon \, \Re \sum_{m \in \mathbb Z^\ast}\,  e^{-i \lambda_{l,m} t}  \left(\mathbf G_0,\mathbf Z_{l,m} \right)_{\mathcal H} \, \mathbf Z^{(2)}_{l,m}(x).
\end{equation}
This approximate formula for the  electric field  shows the recurrence observed in the right-hand side of figure \ref{LDLB}. For clarity, we indicate real part in \eqref{theoNL:2}, but note that since the initial data $\mathbf G_0$ in \eqref{initial.8.5} is real valued , and the solution to the  magnetic Vlasov-Amp\' ere system\eqref{vlasamp}  is unique, actually the electric field  given by \eqref{elc.8.5}  is real valued. }

\appendix

\section{Technical formulas}


\renewcommand{\theequation}{\thesection.\arabic{equation}}

\newtheorem{theorem2}{THEOREM}[section]
\renewcommand{\thetheorem}{\arabic{section}.\arabic{theorem}}

\newtheorem{prop2}[theorem2]{PROPOSITION}
\newtheorem{lemma2}[theorem2]{LEMMA}

In this Appendix we further study the properties of the secular equation (\ref{4.47}, \ref{4.48}, \ref{4.49}).
For later use we prepare the following result.

\begin{prop2}\label{prop4.4}
Let $a_{n,m}, n \in \mathbb Z^\ast, m=1,\dots,$ be the quantity defined in \eqref{4.48}. Then, there is a constant, $C$, that depends on $n,$ such that,
\beq\label{4.52}
a_{n,m} \leq C\, \frac{1}{\sqrt{m}}\, 
 \left[ \frac{e n^2}{  2 \omega_{\text{\rm c}}^2 m  }\right]^m,\qquad  m=1,\dots,
\ene
where $e$ is Euler's number. In particular, for any $p >0$ there is a constant $C,$ that depends on $n$  and  $p,$ such that,
\beq\label{4.53}
a_{n,m} \leq C \frac{1}{m^p}.
\ene     
\end{prop2}

\begin{proof} By equation (10.22.67) in page 245 of \cite{nist} 
\beq\label{4.54}
a_{n,m}= e^{\frac{-n^2}{ \omega_{\text{\rm c}}^2} }   
\, I_m\left(\frac{n^2}{ \omega_{\text{\rm c}}^2 }\right),
\ene
with $I_n(z)$ a modified Bessel function. Furthermore, by equation (10.41.1) in page 256 of \cite{nist},
\beq\label{4.55}
  I_m\left(\frac{n^2}{ \omega_{\text{\rm c}}^2}\right)=
   \frac{1}{\sqrt{2 \pi m}}\,
    \left( \frac{e n^2}{2 \omega_{\text{\rm c}}^2 m }\right)^m
  (1+ o(1)), \qquad m \to \infty.
  \ene
Equation \eqref{4.52} follows from \eqref{4.54} and \eqref{4.55}. Finally, \eqref{4.53} follows from \eqref{4.52}.
\end{proof}

We continue the analysis of the secular equation Let  $ \lambda_{n,m}, m \geq 2 $ be the root given in Lemma ~\ref{lemm4.4}.  Recall that $ \lambda_{n,m} \in (m \omega_{\rm{c}}, (m+1)\omega_{\rm{c}}).$  Then, to isolate   terms that can be large as  $ \lambda_{n,m}$ is close to $m \omega_{\rm{c}}$ or to  ($m+1)\omega_{\rm{c}},$  we decompose $g( \lambda_{n,m})$ as follows,
\beq\label{4.57}
g(\lambda_{n,m})= g^{(1)}( \lambda_{n,m})+g^{(2)}( \lambda_{n,m})+g^{(3)}( \lambda_{n,m})+g^{(4)}( \lambda_{n,m}),
\ene
where,
\beq\label{4.58}
g^{(1)}( \lambda_{n,m}):= 4 \pi \,\sum_{1\leq q \leq m-1 }\,  \frac{q^2\,\omega^2_{\text{\rm c}}}{q^2 \omega^2_{\text{\rm c}}- \lambda_{n,m}^2} a_{n,q},
\ene
\beq\label{4.59}
g^{(2)}( \lambda_{n,m}):= 4\pi\, \frac{m^2\,\omega^2_{\text{\rm c}}}{m^2 \omega^2_{\text{\rm c}}- \lambda_{n,m}^2} a_{n,m},
\ene
\beq\label{4.60}
g^{(3)}( \lambda_{n,m}):= 4 \pi\, \frac{(m+1)^2\,\omega^2_{\text{\rm c}}}{(m+1)^2 \omega^2_{\text{\rm c}}- \lambda_{n,m}^2} a_{n,m+1},
\ene
\beq\label{4.61}
g^{(4)}( \lambda_{n,m}):= 4 \pi \,\sum_{ q  \geq m+2 }\,  \frac{q^2\,\omega^2_{\text{\rm c}}}{q^2 \omega^2_{\text{\rm c}}- \lambda_{n,m}} a_{n,q}.
\ene

\begin{lemma2}\label{lemm4.5}
Let $g^{(1)}( \lambda_{n,m})$ be the quantity defined in \eqref{4.58}. Then,  there is a constant $C_n$ such that,
\beq\label{4.62}
\left|g^{(1)}( \lambda_{n,m})\right| \leq C_n\, \frac{1}{m^2}, \qquad m\geq 2.
\ene
\end{lemma2}

\begin{proof} First suppose that  $m$ is even. Then, $m/2$ is an integer, and we can  decompose $g^{(1)}( \lambda_{n,m})$ as follows,
\beq\label{4.63}
g^{(1)}(\lambda_{n,m})= g^{(1,1)}( \lambda_{n,m})+ g^{(1,2)}( \lambda_{n,m}),
\ene
where,
\beq\label{4.64}
g^{(1,1)}( \lambda_{n,m}):=  4 \pi \,\sum_{1\leq q \leq m/2 }\,  \frac{q^2\,\omega^2_{\text{\rm c}}}{q^2 \omega^2_{\text{\rm c}}- \lambda_{n,m}^2} a_{n,q},
\ene
and
\beq\label{4.65}
g^{(1,2)}( \lambda_{n,m}):=  4 \pi \,\sum_{m/2  < q \leq m-1 }\,  \frac{q^2\,\omega^2_{\text{\rm c}}}{q^2 \omega^2_{\text{\rm c}}- \lambda_{n,m}^2} a_{n,q}.
\ene
Note that,
\beq\label{4.66}
 \left|\frac{1}{q^2 \omega^2_{\text{\rm c}}- \lambda_{n,m}^2}\right| \leq  \frac{2}{m^2 \omega^2_{\text{\rm c}} }, \qquad q=1,\dots,  \frac{m}{2}. 
\ene
Then, by \eqref{4.53}, \eqref{4.64} and, \eqref{4.66}
\beq\label{4.67}
\left| g^{(1,1)}( \lambda_{n,m}) \right| \leq 4 \pi \, \frac{2}{m^2\omega^2_{\text{\rm c}}} \sum_{1}^{m/2}\, q^2\,\omega^2_{\text{\rm c}}\, a_{n,q} \leq C  \frac{1}{m^2}.
\ene
Furthermore, we have
\beq\label{4.68}
 \left|\frac{1}{q^2 \omega^2_{\text{\rm c}}- \lambda_{n,m}^2}\right| \leq  \frac{1}{\omega_{\text{\rm c}}} \, \frac{1}{m {\omega_{\text{\rm c}}}},
    \qquad q=\frac{m}{2},\dots,  m-1. 
\ene
Then, by \eqref{4.53}, \eqref{4.65} and, \eqref{4.68},
\beq\label{4.69}
\left|g^{(1,2)}( \lambda_{n,m})   \right| \leq 4 \pi \, \frac{1}{\omega_{\text{\rm c}}} \, \frac{1}{m \omega_{\text{\rm c}}}\,
\sum_{m/2  < q \leq m-1 }\,  q^2\,\omega^2_{\text{\rm c}}a_{n,q} \leq 
C_p \frac{1}{m^p}, \,p =1,\dots.
\ene
Equation \eqref{4.62} follows from  \eqref{4.63}, \eqref{4.67} and, \eqref{4.69}. In the case where $ m$  is odd, $(m-1)/2$ is an integer, and
we decompose $g^{(1)}(\lambda_{n,m})$ as in \eqref{4.63} with,
\beq\label{4.70}
g^{(1,1)}( \lambda_{n,m}):=  4 \pi \,\sum_{1\leq q \leq (m-1)/2 }\,  \frac{q^2\,\omega^2_{\text{\rm c}}}{q^2 \omega^2_{\text{\rm c}}- \lambda_{n,m}} a_{n,q},
\ene
and
\beq\label{4.71}
g^{(1,2)}( \lambda_{n,m}):=  4 \pi \,\sum_{(m-1)/2  < q \leq m-1 }\,  \frac{q^2\,\omega^2_{\text{\rm c}}}{q^2 \omega^2_{\text{\rm c}}- \lambda_{n,m}^2} a_{n,q},
\ene
and we proceed as in the case of $m$ even.
\end{proof}

In the following lemma we estimate $g^{(4)}( \lambda_{n,m}).$

\begin{lemma2}\label{lemm4.6}
Let $g^{(4)}( \lambda_{n,m})$ be the quantity defined in \eqref{4.61}. Then,  for every $ p >0$ there is a constant $C_p$ such that,
\beq\label{4.72}
\left|g^{(4)}( \lambda_{n,m})\right| \leq C_p\, \frac{1}{m^p}, \qquad m\geq 2.
\ene
\end{lemma2}

\begin{proof} Note that,
\beq\label{4.73}
 \left|\frac{1}{q^2 \omega^2_{\text{\rm c}}- \lambda_{n,m}^2}\right| \leq  \frac{1}{\omega_{\text{\rm c}}} \, \frac{1}{(m+3){\omega_{\text{\rm c}}}},
    \qquad q \geq m+2. 
\ene
Equation \eqref{4.72} follows from \eqref{4.53}, \eqref{4.61} and, \eqref{4.73}.
\end{proof}

In the following lemma we estimate  how $ \lambda_{n,m} $ approaches $m   \omega_{\text{\rm c}}$ as $ m \to\pm  \infty.$

\begin{lemma2}\label{lemm4.7}
We have,
\beq\label{4.74}
 \lambda_{n,m}= m   \omega_{\text{\rm c}}+  2 \pi m  \, \omega_{\text{\rm c}}\, \frac{a_{n,|m|}}{n^2}+ a_{n,|m|}\, O\left(\frac{1}{|m|}\right), \qquad m \to \pm \infty.
\ene
\end{lemma2}

\begin{proof} Note that since $ \lambda_{n,-m}= - \lambda_{n,m}$ it is enough to prove  equation \eqref{4.74} when $m \to \infty.$

 Using \eqref{4.62} and \eqref{4.72}  we write \eqref{4.49} as follows
\beq\label{4.75}
 4 \pi\, \frac{m^2\,\omega^2_{\text{\rm c}}}{ \lambda_{n,m}^2 - m^2 \omega^2_{\text{\rm c}}} a_{n,m}  =n^2 + g^{(3)}( \lambda_{n,m}) +O\left(\frac{1}{m^2}\right), \qquad m \to \infty.
\ene
Moreover, as $ g^{(3)}(\lambda_{n,m}) \geq 0,$ we get,
$$
4 \pi\, \frac{m^2\,\omega^2_{\text{\rm c}}}{ \lambda_{n,m}^2 - m^2 \omega^2_{\text{\rm c}}} a_{n,m} \geq n^2+O\left(\frac{1}{m^2}\right) ,  \qquad m \to \infty.
$$
Then, there is an $ m_0$ such that
$ 4 \pi\, \frac{m^2\,\omega^2_{\text{\rm c}}}{ \lambda_{n,m}^2 - m^2 \omega^2_{\text{\rm c}}} a_{n,m} \geq \frac{\pi}{4}$, $ m \geq m_0$,  
and then,
$
 \lambda_{n,m}^2 \leq   m^2 \omega^2_{\text{\rm c}}  + 16  m^2\,\omega^2_{\text{\rm c}} \, a_{n,m}$, $ m \geq m_0$, 
and taking the square root we obtain
\beq\label{4.76}
 m \omega_{\text{\rm c}} \leq   \lambda_{n,m}  \leq m \omega_{\text{\rm c}} \sqrt{ 1  + 16 \, a_{n,m}},  \qquad m \geq m_0.
\ene
This already shows that $ \lambda_{n,m}$ is asymptotic to $\omega_{\text{\rm c}} $ for large $m.$ However, we can improve this estimate to obtain \eqref{4.74}.  
By \eqref{4.53} and  \eqref{4.76} for every $ p >0,$
\beq\label{4.77}
\left( (m+1)   \omega_{\text{\rm c}} -   \lambda_{n,m}   \right)^{-1}= \frac{1}{  \omega_{\text{\rm c}}}\left(1+ O\left(\frac{1}{m^p}\right)\right),  \qquad m \to \infty.
\ene
Further, introducing \eqref{4.60} and \eqref{4.77} into \eqref{4.75}, and using \eqref{4.53} we obtain,
\beq \label{4.78}
4 \pi\, \frac{m^2\,\omega^2_{\text{\rm c}}}{ \lambda_{n,m}^2 - m^2 \omega^2_{\text{\rm c}}} a_{n,m}  =n^2 +O\left(\frac{1}{m^2}\right), \qquad m \to \infty.
\ene
We rearrange  \eqref{4.78} as follows,
\beq\label{4.79}
 \lambda_{n,m} - m \omega_{\text{\rm c}}= \frac{4\pi}{n^2}\, \frac{m^2\,\omega^2_{\text{\rm c}}}{ \lambda_{n,m}+    m \omega_{\text{\rm c}}} a_{n,m}+    \frac{1}{n^2}\, ( \lambda_{n,m} - m \omega_{\text{\rm c}})\, O\left(\frac{1}{m^2}\right), \qquad m \to \infty.
\ene
By \eqref{4.76}
\beq\label{4.80}
 \lambda_{n,m}- m \omega_{\text{\rm c}} \leq  m \omega_{\text{\rm c}} \, O\left( a_{n,m} \right), \qquad m \to \infty.
\ene
Further,
\beq\label{4.81}
\left( \lambda_{n,m}+ m \omega_{\text{\rm c}}\right)^{-1}= \left(2 m \omega_{\text{\rm c}} + \lambda_{n,m}-  m \omega_{\text{\rm c}}\right)^{-1}= \frac{1}{2m  \omega_{\text{\rm c}}}\, (1+ O\left( a_{n,m} \right)), \qquad m \to \infty.
\ene
Expansion \eqref{4.74} follows from  \eqref{4.80} and, \eqref{4.81}.
\end{proof}

\section{A family of stationary solutions} 

In this appendix we construct explicitly a family of   time-independent   solutions  to the  linearized magnetized Vlasov-Poisson system.
We first construct the family in dimension 1+2 (one dimension in space, two dimensions in velocity), 
that is the situation that we consider in our work.
Then,  we generalize our family of solutions to the case  dimension 3+3 (three dimensions in space, three dimensions in velocity) that is the case considered by  Bedrossian and Wang \cite{bedro}.

\subsection{Dimension 1+2}
 For the purpose of making the comparison with \cite{bedro} more transparent we consider the Vlasov equation,
\begin{equation} \label{nnn.1}
\partial_t f + v_1 \partial _ x f +\mathbf F \cdot \nabla_v f=0,
\end{equation}
with the electromagnetic Lorenz force,
 \beq\label{nnn.2}
 \mathbf F(t,  x) = \frac{q}{m} \left( \mathbf E (t, x)+ v \times \mathbf B_0(t, x)\right).
 \ene
 Taking $ \mathbf B_0\neq 0$ 
the plasma is magnetized. The variable $x$ is in the periodic torus $x\in \mathcal T=[0,2\pi]_{\rm per}$. The velocity variables are $(v_1,v_2)\in \mathbb R^2$.
 We take the charge $ q >0$ (Remark \ref{rem:plusmoins})  the mass  $m>0,$ and we assume, as before,  that the magnetic field $\mathbf B (t,\mathbf x)=\mathbf B_0$ is constant in space-time. We suppose again that the two-dimensional velocity $v$ is perpendicular to the constant magnetic field, i.e., $\mathbf B_0= (0,0,  B_0), B_0 >0.$ Moreover, we assume that the electric field is directed along the first coordinate axis, $\mathbf E(t,x)=(E(t,x), 0, 0 ),$  that it has mean zero,
 $$
 \int_\mathcal T\, E(t,x)\, dx=0,
 $$
  and that it satisfies the Gauss law,
\beq\label{nnn.3}
\partial_x  E(t,x)=   \frac{1}{4\pi} \,\left( -1 +  q \int_{\mathbb R^2}\, f dv\right),
\ene
where as in \cite{bedro} we introduced the factor $\frac{1}{4 \pi}$ in the right-hand side of the Gauss law, and we have taken the charge of the heavy particles equal to minus one. See Remark \ref{rem:plusmoins}.

We linearize  the equations around a homogeneous Maxwellian equilibrium state $ \tilde{f}_0(v),$  where, 
$$
\tilde{f}_0(v):= \frac{1}{2\pi}  e^{\frac{-v^2}{2}}.
$$ 
We take as equilibrium state $\tilde{f}_0(v):= \frac{1}{2 \pi} e^{\frac{-v^2}{2}}$  to make the comparison with the results of \cite{bedro} more transparent.
 This corresponds  to  the expansion,
\beq\label{ccc.1}
f(t,x,v)=\tilde{f}_0(v) + \varepsilon h(t,x,v) +O(\varepsilon^2), 
\ene
and 
\beq\label{ccc.2}
 E(t,x)=E_0 + \varepsilon F(t,x)+O(\varepsilon^2),
\ene
with a null reference electric field $E_0=0$.
Inserting \eqref{ccc.1} and \eqref{ccc.2} into \eqref{nnn.1}-\eqref{nnn.3}, and keeping  the terms up to linear in $\varepsilon, $ we obtain the linearized magnetized Vlasov-Poisson system

\begin{equation} \label{nnn.4}
\left\{
\begin{aligned}
&\partial_t h+ v_1 \partial_x h - \frac{q}{m}   F v_1  {\tilde{f}_0} - \frac{q}{m} B_0 \left( -v_2 \partial_{v_1} +v_1 \partial_{v_2}  \right) h= 0, \\
&\partial_x  F= \frac{q}{4\pi}  \, \int_{\mathbb R^2} h \, dv_1 dv_2,   \int_\mathcal T\, F(t,x)\, dx=0.
\end{aligned} 
\right.
\end{equation}
As we look for time-independent solutions, we have to solve,

\begin{equation} \label{bd:app1}
\left\{
\begin{aligned}
& v_1 \partial_x h - \frac{q}{m} F v_1  {\tilde{f}_0}- \frac{q}{m} B_0 \left( -v_2 \partial_{v_1} +v_1 \partial_{v_2}  \right) h= 0, \\
&\partial_x  F= \frac{q}{4 \pi} \int_{\mathbb R^2} h\, dv_1 \,dv_2, \qquad \int_\mathcal T\, F(t,x)\, dx=0.
\end{aligned} 
\right.
\end{equation}
Note that,
\begin{equation} \label{bd:app2}
\left( -v_2 \partial_{v_1} +v_1 \partial_{v_2}  \right) \tilde{f}_0(v_1,v_2) =0.
\end{equation}
Our objective hereafter is to construct a family of non trivial smooth solutions to \eqref{bd:app1} that have fast decay  in velocity.

\begin{lemma2} \label{lemm:app}
There exists an explicit family of non  trivial smooth solutions $(h,F)$ to the time-independent  linearized magnetized Vlasov-Poisson system (\ref{bd:app1}),  where $ F= - \varphi'(x),$ with $ \varphi \in C^\infty(\mathcal T),$ and where the function $h$  can be taken with $l$ continuous derivates with respect to $v, l=1,2,\dots,$ or infinitely
differentiable with respect to $v.$ Moreover, for each fixed  $ x \in \mathcal T,  h\in L^1(\mathbb R^2).$ Further,  the absolute value of $h$ and of all its derivatives  can be taken bounded by  Gaussian functions of $v,$ uniformly in $ x \in \mathcal T.$ Moreover,   $ h+\frac{q}{m} \varphi f_0$ can be taken with compact support in $\mathbf v$, uniformly in $x \in \mathcal T.$   
 \end{lemma2}

\begin{proof} We introduce an electric potential  $ \varphi \in C^1(\mathcal T)$ as 
$$
F(x)=-\varphi'(x),
$$
 with $\varphi(2\pi)=\varphi(0) $ and $\varphi'(2 \pi)= \varphi'(0).$ 
 Plugging in the first equation in \eqref{bd:app1}, we obtain 
$$
v_1\partial_x
\left[
h(x,v_1,v_2) +  \frac{q}{m}  \varphi(x) \tilde{f}_0(v_1,v_2) 
\right]-  \frac{q}{m}  B_0 \left( -v_2 \partial_{v_1} +v_1 \partial_{v_2}  \right) h= 0.
$$
Let us define
\begin{equation} \label{bd:app3}
G(x,v_1,v_2)=h(x,v_1,v_2) + \frac{q}{m}    \varphi(x) \tilde{f}_0(v_1,v_2) .
\end{equation}
Since we have (\ref{bd:app2}), $G$ satisfies the equation 
\beq\label{ddd.1}
v_1\partial_x G(x,v_1,v_2) -  \frac{q}{m}   B_0 \left( -v_2 \partial_{v_1} +v_1 \partial_{v_2}  \right) G(x,v_1,v_2)=0.
\ene
Let us make another change of function which is valid since $B_0\neq 0,$
\beq\label{ddd.2}
H(x,v_1,v_2)= G(x-  \frac{m \,v_2}{\, qB_0 },v_1,v_2)
\Longleftrightarrow
G(x,v_1,v_2)= H(x+  \frac{m \, v_2}{q \,B_0},v_1,v_2).
\ene
The equation (\ref{ddd.1}) is rewritten as 
\begin{equation} \label{bd:app4}
   \frac{q}{m}   \omega_c \left( -v_2 \partial_{v_1} +v_1 \partial_{v_2}  \right) H(x,v_1,v_2)=0.
 \end{equation}
 Under this form, it is easy to find a general solution which 
 writes
 $$
 H(x,v_1,v_2)= K(x, v_1^2+v_2^2)
 $$
 where $K$ is an arbitrary smooth  function which decreases sufficiently fast at infinity with respect to its second variable.
For example, we can take,
\begin{equation} \label{bd:app7}
K(x, v_1^2+v_2^2)= e^{inx}  g( v_1^2+v_2^2), \qquad n \in \mathcal Z\setminus\{0\},
\end{equation}
 where $g(v_1^2+v_2^2) \in C^l(\mathbb R^2), l=1,2,\dots,$ or $g(v_1^2+v_2^2) \in C^\infty(\mathbb R^2),$ and $  g(v_1^2+v_2^2) \in L^1(\mathbb R^2).$  For example, $g$ can be taken with compact support, or a Gaussian. Going back to the perturbation $h$, one obtains
the representation formula
$$
h(x,v_1,v_2) = -\frac{q}{m} \varphi(x) \tilde{f}_0(v_1,v_2) + K(x+  \frac{m \,v_2}{q\, B_0}, v_1^2+v_2^2),
$$
where the electric potential  $\varphi$ remains to be determined. All functions $h$ of this form satisfy 
the first  equation of (\ref{bd:app1}).
\\
It remains to verify the Gauss law, that is the   second equation   of (\ref{bd:app1}).
The right-hand side of the Gauss law is
$$
 \frac{q}{4\pi }  \int_{\mathbb R^2} h \,dv_1\,dv_2=  \frac{q}{4\pi } \left(-\frac{q}{m}\, \varphi(x) +\int_{\mathbb R^2}\, K(x+  \frac{m\, v_2}{q\, B_0}, v_1^2+v_2^2) \,dv_1 \, dv_2\right).
$$
So the Gauss law is rewritten as 
$$
-\varphi''(x) + \frac{q^2}{4\pi  m} \,  \varphi(x) =   \frac{q}{4\pi } \, \int_{\mathbb R^2} K(x+  \frac{m \,v_2}{q\, B_0}, v_1^2+v_2^2) \,dv_1 \,dv_2.
$$
Using  (\ref{bd:app7}),  one gets
$$
-\varphi''(x)+ \frac{q^2}{4\pi m } \varphi(x) =  \frac{q}{4\pi }   e^{inx}  \int_{\mathbb R^2} \ds e^{ in  \frac{m v_2}{q B_0}}  g( v_1^2+v_2^2) \, dv_1 \, dv_2.
$$
This is an  equation for the electric potential.
The periodic solution is explicit,
\begin{equation} \label{eq:1}
\varphi(x) = \frac1{n^2+\frac{q^2}{4\pi m}}\, \frac{q}{4\pi }  \,e^{inx} \, \int_{\mathbb R^2} \ds e^{in  \frac{m v_2}{q B_0} }  g( v_1^2+v_2^2) dv_1dv_2.
\end{equation}
The remaining properties of the solution $(h,F)$ follow immediately from the explicit representation of $(h,F).$ 
\end{proof}

Remark that the solutions given by Lemma~\ref{lemm:app} satisfy,
$$
\int_{\mathcal T} h(x,v_1,v_2)\, dx=0,
$$
and in particular,
$$
\int_{\mathcal T \times \mathbb R^2}\, h(x,v_1,v_2)\, dx\, dv_1\, dv_2=0.
$$
The solutions given by Lemma~\ref{lemm:app} are in agreement with \eqref{4.144}, \eqref{4.145}, \eqref{4.146.b}, \eqref{4.151}, and also with \eqref{4.152.0}, \eqref{4.152.1}, \eqref{4.152}.


\subsection{Dimension 3+3}

We now consider solutions to the magnetized Vlasov-Poisson system in $\mathcal T^3 \times \mathbb R^3,$ where $\mathcal T^3$ is the three-dimensional torus, $\mathcal T^3:=
[0,2\pi]_{\rm per}^3.$  We denote $\mathbf x:=(x,y,z)\in \mathcal T^3,$ and $\mathbf v=(v_1,v_2,v_3)\in \mathbb R^3.$  The Vlasov equation is given by,

\begin{equation} \label{ppp.1}
\partial_t f(t,\mathbf x,\mathbf v) + \mathbf v \cdot \nabla _ {\mathbf x} f(t, \mathbf x,\mathbf v) +\mathbf F(t, \mathbf x) \cdot \nabla_{\mathbf v} f(t,\mathbf x,\mathbf v)=0,
\end{equation}
with the electromagnetic Lorenz force,
 \beq\label{ppp.2}
 \mathbf F(t,  \mathbf x) = \frac{q}{m} \left( \mathbf E (t, \mathbf x)+ \mathbf v \times \mathbf B_0(t,\mathbf x)\right),
 \ene
 where  $ q >0,  m >0,$  and we assume, as before,  that the magnetic field $\mathbf B (t,\mathbf x)=\mathbf B_0$ is constant in space-time, and that it is directed along the third coordinate, i.e., $\mathbf B_0= (0,0,  B_0), B_0 >0.$ Moreover, we assume that the electric field  satisfies the Gauss law,
\beq\label{ppp.3}
\nabla_{\mathbf x}\cdot  \mathbf E(t, \mathbf x)=   \frac{1}{4\pi}  \rho(t,\mathbf x), \qquad    \rho(t,\mathbf x):=  \,[-1 +  q \int_{\mathbb R^3}\, f (t,\mathbf x, \mathbf v)\, d^3
\mathbf v].
\ene
 Further, we assume that the electric field has mean zero,
\beq\label{ppp3.b}
\int_{\mathcal T^3}\, \mathbf E(t,\mathbf x)\, d^3\mathbf x=0.
\ene
We linearize   equations   \eqref{ppp.1}-\eqref{ppp3.b}  around a homogeneous Maxwellian equilibrium state $ f^0(\mathbf v),$  where, 
$$
f^0(\mathbf v):= \tilde{f}_0(v_1,v_2)\   \frac{1}{\sqrt{2\pi T_\parallel  }}\,  e^{\frac{-v^2_3}{2 T_\parallel}},
$$ 
where $T_\parallel >0$ is the temperature along the magnetic field.
This amounts to take  $\tilde f=0$  in (\ref{eq:ftilde}).
 This corresponds  to  the expansion,
\beq\label{ppp.4}
f(t,\mathbf x,\mathbf v)=f^0(\mathbf v) + \varepsilon \mathcal G(t,\mathbf x,\mathbf v) +O(\varepsilon^2), 
\ene
and 
\beq\label{ppp.5}
\mathbf  E(t,\mathbf x)=\mathbf E_0 + \varepsilon \mathcal F(t, \mathbf x)+O(\varepsilon^2),
\ene
with a null reference electric field $\mathbf E_0=0$.
Inserting \eqref{ppp.4} and \eqref{ppp.5} into \eqref{ppp.1}-\eqref{ppp.3}, and keeping  the terms up to linear in $\varepsilon, $ we obtain the linearized  magnetized   
  Vlasov-Poisson system,

\begin{equation} \label{ppp.6}
\left\{
\begin{aligned}
&\partial_t \mathcal G(t,\mathbf x,\mathbf  v)+ \mathbf v \cdot \nabla_{\mathbf x} \mathcal G(t,\mathbf x, \mathbf v) + \frac{q}{m}   \mathcal F(t,\mathbf x)\cdot 
 \nabla_{\mathbf v} {f^0(\mathbf v)}+ \frac{q}{m}   \mathbf v \times \mathbf B_0\cdot
\nabla_{\mathbf v} \mathcal G(t,\mathbf x,\mathbf v)= 0, \\
&  \nabla_{\mathbf x} \cdot  \mathcal  F(t,\mathbf x)=  \,\frac{q}{4\pi}  \, \rho(t,\mathbf x), \qquad \rho(t,\mathbf x):=  \int_{\mathbb R^3}\, \mathcal G(t,\mathbf x,\mathbf  v)\, d^3\mathbf v, \qquad   \int_{\mathcal T^3}\, \mathcal  F(t,\mathbf x)\, d^3\mathbf x=0.
\end{aligned} 
\right.
\end{equation}
We look for solutions to \eqref{ppp.6} that satisfy,
\beq\label{ppp.7}
\int_{\mathcal T^3 \times \mathbb R^3}\, \mathcal G(t,\mathbf x,\mathbf  v)\, d^3\mathbf x\,d^3 \mathbf v=0.
\ene
Under the condition \eqref{ppp.7} the Gauss law, that is the second equation in \eqref{ppp.6}, is equivalent to the following equation,

\beq\label{ppp.8}
\mathcal F(t,\mathbf x)=- \nabla_{\mathbf x} \, \int_{\mathcal T^3}\, W(\mathbf x-\mathbf y)\, \rho(t,\mathbf y)\, d^3\mathbf y, 
\ene
where,
$$
W(\mathbf x):= \frac{q}{4\pi}\, \frac{1}{(2\pi)^3}\, \sum_{
\mathbf k \in \mathbb Z^3\setminus\{0\}}\, \frac{1}{| \mathbf k|^2}\, e^{ik\cdot\mathbf x}.
$$
Since we are looking for time-independent solutions we have to solve,

\begin{equation} \label{ppp.9}
\left\{
\begin{aligned}
& \mathbf v \cdot \nabla_{\mathbf x} \mathcal G(t,\mathbf x, \mathbf v) + \frac{q}{m}   \mathcal F(t,\mathbf x)\cdot  \nabla_{\mathbf v} {f^0(\mathbf v)}+ \frac{q}{m}  \mathbf v \times \mathbf B_0\cdot
\nabla_{\mathbf v} \mathcal G(t,\mathbf x,\mathbf v)= 0, \\
&\nabla_{\mathbf x}\cdot  \mathcal  F(t,\mathbf x)= \frac{q}{4\pi}  \,  \rho(t,\mathbf x), \qquad \rho(t,\mathbf x):=  \int_{\mathbb R^3}\, \mathcal G(t,\mathbf x, \mathbf v)\, d^3\mathbf v, \qquad  \int_{\mathcal T^3}\, \mathcal  F(t,\mathbf x)\, d^3\mathbf x=0.
\end{aligned} 
\right.
\end{equation}
Let $(h,F)$ be one of the  solutions to \eqref{bd:app1} given by Lemma~\ref{lemm:app}. We define, 
\beq\label{ppp.10}
\mathcal G (\mathbf x,\mathbf v)= h(x,v_1,v_2)\, \frac{1}{\sqrt{2\pi T_\parallel  }}\,  e^{\frac{-v^2_3}{2 T_\parallel}},
\mbox{and }
\mathcal F(\mathbf x)=\left(
\begin{array}{ccc}
F(x) \\
0 \\
0
\end{array}
\right).
\ene

\begin{lemma2}\label{sol.3d}
Let $(h(x,v_1,v_2),F(x))$ be one of the solutions to the time-independent  linearized   magnetized Vlasov-Poisson system \eqref{bd:app1} given by Lemma~\ref{lemm:app}. Then, the pair $(\mathcal G(\mathbf x,\mathbf v),\mathcal F(\mathbf x))$ defined
in \eqref{ppp.10} is a  solution to the time-independent  linearized  magnetized Vlasov-Poisson system \eqref{ppp.9} in $\mathcal T^3\times \mathbb R^3,$ with $ \mathcal F \in C^\infty(\mathcal T^3),$  and where the function $h$  can taken with $l$ continuous derivates with respect to $\mathbf v, l=1,2,\dots,$ or infinitely
differentiable with respect to $\mathbf v.$ Moreover, for each fixed  $ x \in \mathcal T^3,  h\in L^1(\mathbb R^3).$ Further,  the absolute value of $h$ and of all its derivatives  can be taken bounded by  Gaussian functions of $\mathbf v,$ uniformly in $ \mathbf x \in \mathcal T^3.$  Further, the solution $(\mathcal G(\mathbf x,\mathbf v),\mathcal F(\mathbf x))$ satisfies \eqref{ppp.7}.
\end{lemma2}

\begin{proof}
We detail the calculations for the convenience of the reader. One has 
\beq\label{ppp.11}
\begin{array}{lll}
\mathbf v \cdot \nabla_{\mathbf x} \mathcal G (x,y,z,v_1,v_2,v_3)= v_1 \partial_x h  (x,v_1,v_2) \   \frac{1}{\sqrt{2\pi T_\parallel  }}\,  e^{\frac{-v^2_3}{2 T_\parallel}}, \\\\
\nabla_{\mathbf v} \mathcal G (x,y,z,v_1,v_2,v_3)= \left(
\begin{array}{ccc}
\partial_{v_1} h(x,v_1,v_2) \   \frac{1}{\sqrt{2\pi T_\parallel  }}\,  e^{\frac{-v^2_3}{2 T_\parallel}}  \\\\
\partial_{v_2} h (x,v_1,v_2)  \   \frac{1}{\sqrt{2\pi T_\parallel  }}\,  e^{\frac{-v^2_3}{2 T_\parallel}}   \\\\
h(x,v_1,v_2)  \  \partial_{v_3} \frac{1}{\sqrt{2\pi T_\parallel  }}\,  e^{\frac{-v^2_3}{2 T_\parallel}} \\
\end{array}
\right), 
\\ \\ 
\mathcal F(x,y,z)\cdot \nabla_{\mathbf v} f^0 = -F(x) \,v_1 \,f^0, \\
 \mathbf v \times \mathbf B_0= \left(
\begin{array}{ccc}
B_0 v_2 \\\\
-
B_0 v_1 \\\\
0
\end{array}
\right), \\ \\
(\mathbf v \times \mathbf B_0)\cdot \nabla_{\mathbf v} \mathcal G (x,y,z,v_1,v_2,v_3)=
B_0 \left(v_2\partial_{v_1} - v_1\partial_{v_2} \right) h(x, v_1, v_2)  \frac{1}{\sqrt{2\pi T_\parallel  }}\,  e^{\frac{-v^2_3}{2 T_\parallel}}  .
\end{array}
\ene
Therefore, by \eqref{ppp.11} one gets the first equation in the linearized  magnetized Vlasov-Poisson system \eqref{ppp.9},
\begin{equation} \label{bd:app10}
\mathbf v \cdot \nabla_{\mathbf x} \mathcal G(t,\mathbf x, \mathbf v) + \frac{q}{m}   \mathcal F(t,\mathbf x)\cdot  \nabla_
{\mathbf v} {f^0(v)}+ \frac{q}{m}   \mathbf v \times \mathbf B_0\cdot
\nabla_{\mathbf v} \mathcal G(t,\mathbf x,\mathbf v)= 0.
\ene
Moreover, one has $\nabla_{\mathbf x} \cdot \mathcal F(\mathbf x)= \partial_x F(x),$ and 
$$
\begin{array}{lll}
\int_{\mathbb R^3} \mathcal G (\mathbf x,\mathbf v)\, d^3\mathbf v& =&
\int_{\mathbb R^2} h(x,v_1,v_2)dv_1dv_2\times \int_{\mathbb R} \, \frac{1}{\sqrt{2\pi T_\parallel  }}\,  e^{\frac{-v^2_3}{2 T_\parallel}}  \,  dv_3 \\
& =&
\int_{\mathbb R^2} h(x,v_1,v_2)dv_1dv_2 .
\end{array}
$$
So one obtains immediately the Gauss law
\begin{equation} \label{bd:app12}
\nabla_{\mathbf x} \cdot \mathcal F(\mathbf x )=\frac{q}{4\pi}\, \rho(t,\mathbf x).
\end{equation}
The fact that \eqref{ppp.7} holds, and the properties of the solution $(\mathcal G, \mathcal F)$ stated in the lemma
hold, follow immediate from the definition of the pair $(h,F).$
\end{proof}
\textcolor{black}{ The solutions $(\mathcal G, \mathcal F),$  given by Lemma~\ref{sol.3d},
to the time-independent  linearized magnetized Vlasov-Poisson system \eqref{ppp.9}, and that fulfill \eqref{ppp.7}, satisfy the assumption of Theorem~1 of Bedrossian and Wang, \cite{bedro} (see Theorem~\ref{theob} above). For example, we can  take $ g(v_1^2+v_2^2)= e^{-v_1^2+v_2^2}.$   Moreover, the charge density fluctuation is independent of $y$ and $z,$ and is given by,
\beq\label{density.1}
\rho(x)= e^{inx} \left[    \frac{1}{n^2+\frac{q^2}{4\pi m}}\, \frac{-q^2}{4\pi  m} + 1 \right]  \int_{\mathbb R^2} \,  e^{in \frac{m \,v_2}{q\, B_0}}\,   g( v_1^2+v_2^2) \, dv_1 dv_2.   
\ene
 The Fourier coefficient  $\hat{\rho}(n,0,0)$ is, in general,  not zero, and it is given by,
 \beq\label{fourho.999}
 \hat{\rho}(n,0,0)= (2 \pi)^3\,  \left[    \frac{1}{n^2+\frac{q^2}{4\pi m}}\, \frac{-q^2}{4\pi  m} +  1 \right]  \int_{\mathbb R^2} \,  e^{in \frac{m \,v_2}{q\, B_0}}\,   g( v_1^2+v_2^2) \, dv_1 dv_2.   
\ene
}
\subsection{Limit $ B_0 \to 0$}

An interesting question is  
passing to the limit $B_0 \to 0$   in the right-hand side of (\ref{eq:1}).
 One has weak convergence to zero of the right-hand side under standard integrability conditions on $g$ since 
$$
\underset{B_0 \to 0} \lim \int_{\mathbb R^2} e^{i n\frac{ mv_2}{q B_0}}  g( v_1^2+v_2^2) dv_1dv_2 =0
$$
 because of the oscillating term $e^{i n\frac{ mv_2}{q B_0}}.$
Therefore, by \eqref{eq:1}  the solutions of (\ref{bd:app1}) given by Lemma~\ref{lemm:app}. satisfy
$$
\underset{ B_0 \to 0}  \varphi= 0.
$$ 
The weak limit recovers  the classical results in the non magnetized case.

\vspace{1cm}
\noindent{\bf Acknowledgement}
This paper was partially written while Ricardo Weder was visiting the Institut de Math\'ematique d'Orsay, Universit\'e  Paris-Saclay.  Ricardo Weder thanks Christian G\'erard for his kind hospitality. 
Fr\'ed\'erique Charles, Bruno Despr\'es and Alexandre Rege acknowledge the support of the MUFFIN ANR project
under contract number ANR-19-CE46-0004. 

{ 
 
 }

\end{document}